\documentclass[onecolumn,12pt]{IEEEtran}
\usepackage{setspace}  
\usepackage{fullpage}

\usepackage{color}
\usepackage{ifpdf}
\usepackage{graphicx}
\usepackage{amsmath}
\usepackage{amsthm}
\usepackage{amsfonts}
\usepackage{amssymb}
\usepackage{eucal}
\usepackage{subfigure}
\usepackage{array}
\usepackage{url}
\usepackage{multirow}
\usepackage{dblfloatfix}
\usepackage[colorlinks=true,pdfstartview=FitV,linkcolor=black,citecolor=black,urlcolor=blue,plainpages=false]{hyperref}
\usepackage[numbers,sort&compress]{natbib}
\graphicspath{{pic//}}
\usepackage{enumerate}
\newtheorem{theorem}{Theorem}
\newtheorem{lemma}{Lemma}
\newtheorem{remark}{Remark}

\newtheorem{corollary}{Corollary}

\usepackage{comment}
\usepackage[normalem]{ulem}

\usepackage{tikz}
\usetikzlibrary{arrows}
\usetikzlibrary{arrows,positioning}
\usetikzlibrary{decorations.markings}
\tikzset{
    >=stealth',
    punkt/.style={
           rectangle,
           rounded corners,
           draw=black, very thick,
           text width=6.5em,
           minimum height=2em,
           text centered},
    pil/.style={
           ->,
           thick,
           shorten <=2pt,
           shorten >=2pt,}
    pir/.style={
           <-,
           thick,
           shorten <=2pt,
           shorten >=2pt,}
}
\definecolor{KeynoteRed}{rgb}{.678,.051, .051}
\definecolor{KeynoteBlue}{rgb}{0.008, 0.443, 0.60}
\definecolor{KeynoteLightblue}{rgb}{.635, .914, .973}
\definecolor{KeynoteYellow}{rgb}{0.859, 0.584, 0.212}
\definecolor{KeynoteYellow}{rgb}{0.859, 0.584, 0.212}
\definecolor{KeynoteSlate}{rgb}{0.239, 0.271, 0.322}
\definecolor{KeynoteGray}{rgb}{0.498, 0.529, 0.529}
\definecolor{KeynoteTextGray}{rgb}{0.325, 0.325, 0.325}
\definecolor{KeynoteLightGray}{rgb}{0.706, 0.706, 0.706}
\definecolor{KeynoteBlueGray}{rgb}{0.471, 0.533, 0.620}
\definecolor{ECEpurple}{rgb}{.169, .18, .455}
\definecolor{ECEcyan}{rgb}{.41, .62, .72}
\definecolor{ECEgray}{rgb}{.788, .827, .859}
\definecolor{ECEblueGray}{rgb}{61.2, 70.6, 70.6}
\definecolor{ECEblueGray}{rgb}{61.2, 70.6, 70.6}
\definecolor{RiceBlue}{rgb}{0, .14, .41}

\title{Vector Bin-and-Cancel for MIMO Distributed Full-Duplex
}
\date {}

    \author{
      Jingwen Bai, Chris Dick and Ashutosh Sabharwal, \emph{Fellow, IEEE}\footnote{J. Bai and A. Sabharwal are with Department of Electrical and Computer Engineering
      Rice University, Houston, TX 77005, USA, e-mail:\{jingwen.bai,ashu\}@rice.edu. C. Dick is with Xilinx Inc., San Jose, CA, 95124 USA, e-mail: chris.dick@xilinx.com. This work was partially supported by NSF CNS-1012921, NSF CNS-1161596 Xilinx and Intel.}}

\vspace{15pt}

\begin{document}
\maketitle


\begin{abstract} 
In a multi-input multi-output~(MIMO) full-duplex network, where an in-band full-duplex infrastructure node communicates with two half-duplex mobiles supporting simultaneous up- and downlink flows, the inter-mobile interference between the up-  and downlink mobiles limits the system performance. We study the impact of leveraging an out-of-band side-channel between mobiles in such network under different channel models. For time-invariant channels, we aim to characterize the  generalized degrees-of-freedom~($\mathsf{GDoF}$) of the side-channel assisted MIMO full-duplex network. For slow-fading channels, we focus on the diversity-multiplexing tradeoff (DMT) of the system with various assumptions as to the availability of channel state information at the transmitter~(CSIT). The key to the optimal performance is a vector bin-and-cancel strategy leveraging Han-Kobayashi message splitting, which is shown to achieve the system capacity region to within a constant bit.
We quantify how the side-channel improve the $\mathsf{GDoF}$ and DMT compared to a system without the extra orthogonal spectrum. 
The insights gained from our analysis reveal: i) the tradeoff between spatial resources from multiple antennas at different nodes and spectral resources of the side-channel, and ii) the interplay between the channel uncertainty at the transmitter and use of the side-channel. 
\end{abstract} 


\section{introduction}
Increasingly, mobile devices have multiple radios to simultaneously access different parts of the spectrum, e.g. cellular and ISM bands. The ability of simultaneous  access to multiple parts of the spectrum provides an opportunity to use multiple bands in new and unique ways. A common method is to use the two bands to access both cellular and ISM band networks (notably WiFi) at the same time and is now an integral part of cellular provider data strategy to offload cellular traffic to WiFi networks~\cite{3Goffload}. In this paper, we will consider the use of  device-to-device~(D2D) wireless channels between mobile devices, to serve as side-channels to \emph{aid} main-channels communication with the infrastructure nodes.
For example, the main network could be on a cellular band while the wireless side-channel could be on an unlicensed ISM band. The conventional use of D2D involves establishing peer-to-peer communication \cite{D2Doverlay}, forming virtual MIMO by cooperative communication \cite{D2Dvmimo} or offloading cellular traffic~\cite{han2010cellular}. In contrast, we propose to use the D2D side-channel for interference management to improve the cellular capacity, a scenario which we labeled as \emph{ISM-in-cellular} communication~\cite{Jingwen,JingwenTWC,cellnet14}.

In this paper, we will study how the side-channel will impact the system performance in a two-user MIMO full-duplex network. 
In-band full-duplex operation promises to double the spectral efficiency as compared to the half-duplex counterpart which uses either time division or frequency division for transmission and reception. 
It is in fact feasible to design near-perfect full-duplex base stations owing to the available freedom (bigger size, non-battery-powered operation) in their designs~(e.g., see \cite{duarte2012design,Everett13Paper,ashu13} and the references therein).  And in-band full-duplex has already become part of the ongoing standard both in 3GPP~\cite{3gppFD}  and 802.11-ax~\cite{wifiMassiveMIMOFD}. Thus, we envision that the first use of full-duplex capabilities might be in small cell infrastructure~\cite{Smallcell}, supporting legacy half-duplex mobile nodes.

In Fig.~\ref{fig.1},  a full-duplex capable base station~(BS) communicates with two half-duplex mobiles simultaneously to support one uplink~(UL) and one downlink~(DL) flow. A major bottleneck in this network is the inter-mobile interference from uplink mobile (node M1) to downlink mobile node (node M2), because of which the degrees-of-freedom of the network collapse to one when all nodes are equipped with single antenna (SISO)~\cite{JingwenTWC}. As a result, we proposed a \emph{distributed full-duplex architecture}~\cite{JingwenTWC} to leverage the wireless side-channels to mitigate inter-mobile interference. In the case of MIMO scenario, one driving question is if and how the spatial degree-of-freedom, i.e. number of antennas at the base station and mobiles, will be correlated to the spectral degrees-of-freedom offered by the side-channel. 
\begin{figure}[h!]
  \centering
    \includegraphics[width=0.35\textwidth,trim = 70mm
      50mm 65mm 35mm, clip]{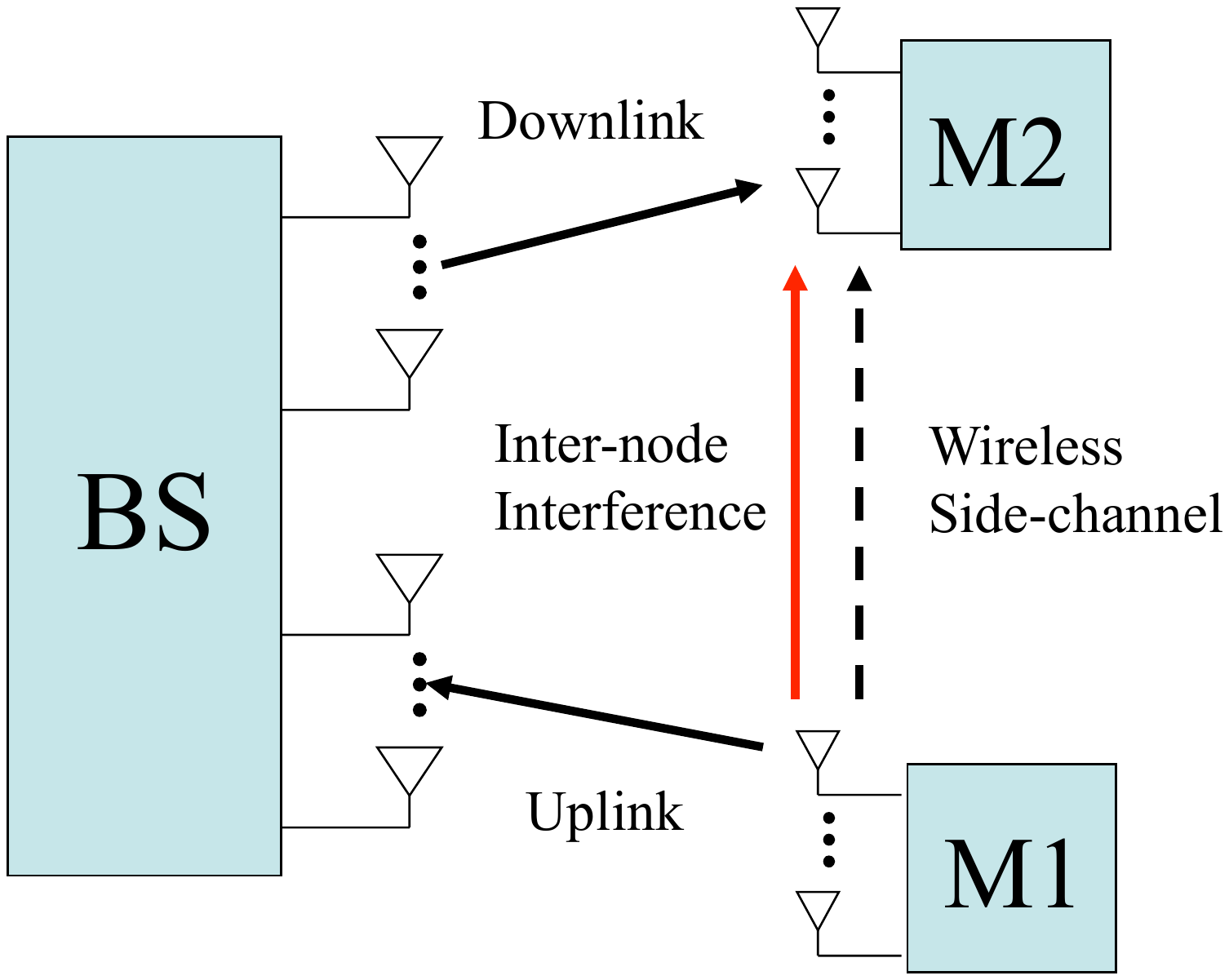}
  \caption{\emph{MIMO full-duplex network: inter-mobile interference becomes an important factor when the full-duplex infrastructure node communicates with uplink and downlink mobile nodes simultaneously. }}
\label{fig.1}
\end{figure}

In our setup, we assume that uplink node~M1 has $M_{\mathrm{ul}}$ transmit antennas, the downlink node~M2 has $N_{\mathrm{dl}}$ receive antennas, the full-duplex BS has $M_{\mathrm{dl}}$ and $N_{\mathrm{ul}}$ transmit and receive antennas, respectively. The bandwidth of the side-channel between the mobiles is $W$-fold compared to the main-channel. We summarize the main results in this work as follows.
\begin{enumerate}
\item In the time-invariant channels, we obtain the capacity region to within a constant bit achieved by a vector bin-and-cancel scheme. We also analyze the role of channel uncertainty at the transmitter and characterize the $\mathsf{GDoF}$ as a function of antenna numbers and side-channel bandwidth under different assumptions of CSIT. The insights gained from $\mathsf{GDoF}$ reveal the tradeoff between spatial resources from multiple antennas and spectral resources of the side-channels as well as the interplay between the channel uncertainty at the transmitter and the use of side-channel. In the case when BS has more antennas than mobiles, if there are more downlink receive antennas than uplink transmit antennas, i.e., $N_{\mathrm{dl}}\geq M_{\mathrm{ul}}$, there is no benefit to obtain CSIT since with and without CSIT achieve the same degrees-of-freedom.
On the other hand, if $M_{\mathrm{ul}}>N_{\mathrm{dl}}$, having CSIT require less side-channel bandwidth to achieve no-interference performance. Thus we conclude that having more spatial degree-of-freedom at the interfered downlink receiver or larger side-channel bandwidth can simplify transceiver design by ruling out the necessity to obtain CSIT. 

\item In slow-fading channels, we derive the general DMT regarding different assumptions of CSIT. Specifically, we quantify the bandwidth of the side-channel required to compensate for \emph{lack of} CSIT such that the DMT without CSIT achieves the optimal DMT with CSIT. Interestingly, in the case when $M_{\rm dl}=N_{\rm ul}=M\geq M_{\rm ul},N_{\rm dl}$, the required bandwidth is inversely proportional to the number of antennas at the BS, i.e., $W\propto\frac{1}{M}$. The caveat is that the side-channel channel SNR, in the meantime, has to grow with the number of antennas at BS. The result provides guidance towards system design: larger number of BS antennas, e.g. recent discussions on massive MIMO~\cite{larsson2013massive}, can help reduce the required side-channel bandwidth to combat inter-mobile interference. 

We also observe the dependency of CSIT and the antenna number ratio between the mobiles. For the symmetric DMT,  when $M_{\mathrm{ul}}>N_{\mathrm{dl}}$, without side-channel, the lack of CSIT will result in performance loss. However, larger side-channel bandwidth will help bridge the performance gap. On the other hand, when $N_{\mathrm{dl}}\geq M_{\mathrm{ul}}$, there is no benefit to obtain CSIT to achieve no-interference DMT since, with and without CSIT, one requires the same amount of side-channel bandwidth to completely eliminate the effect of interference. Hence in the protocol design, the scheduler could possibly group downlink user with more receive antennas to eliminate the overhead of acquiring CSIT.

\item We evaluate the required side-channel bandwidth to achieve the no-interference $\mathsf{GDoF}$ and DMT under different channel models such that the effect of inter-mobile interference can be completely eliminated via side-channel. The key difference in the findings between the two channel scenarios, for instance,  when $M_{\rm dl}=N_{\rm ul}=M\geq M_{\rm ul}, N_{\rm dl}$, is that in $\mathsf{GDoF}$ analysis under time-invariant channels, the required $W$ does not depend on the antenna number ratio between the mobiles; while in DMT analysis under slow-fading channels, required $W$ is a function of the antennas number ratio $A=\frac{\max(M_{\mathrm{ul}},N_{\mathrm{dl}})}{\min(M_{\mathrm{ul}},N_{\mathrm{dl}})}$ and $W\propto\frac{1}{A}$. 
The impact on the system design is that we should schedule up- and downlink user pair with higher antenna ratio to cancel out interference with reduced side-channel bandwidth.

\end{enumerate}

The rest of paper is organized as follows. Section~\ref{sec2} presents the system model.
In Section~III, we show that a vector bin-and-cancel scheme achieves within a constant gap of the capacity region in time-invariant channels. We give a characterization of $\mathsf{GDoF}$ which reveals tradeoff between spatial resources from multiple antennas and spectral resources of the side-channels under both CSIT and no-CSIT assumptions. In Section IV, we derive the general DMT with and without CSIT in slow-fading channels. We also study the spatial and spectral tradeoff between multiple antennas and side-channel on the symmetric DMT.
Section V concludes the paper.

\emph{Notations:} We use $A^\dagger$ to denote Hermitian of $A$, and $|A|$ to denote the determinant of $A$. We use $(x)^+$ to denote $\max(x,0)$. We use $\mathcal{CN}(0,Q)$ to denote a circularly symmetric complex Gaussian distribution with zero mean and covariance matrix $Q$. We use ${I_{N}}$ to denote identity matrix of rank $N$. We use $f(\rho)\doteq g(\rho)$ to denote that $\lim_{\rho\rightarrow \infty}\frac{\mathrm{log}f(\rho)}{\mathrm{log}g(\rho)}=1$. We use $A\preceq B$ to denote that matrix $B-A$ is a positive-semidefinite positive~(p.s.d) matrix. 

\section{System Model} \label{sec2}
In this section, we describe the system model to be used for the rest of the paper. We assume the full-duplex BS is equipped with $M_{\mathrm{dl}}$ transmit antennas for the downlink and $N_{\mathrm{ul}}$ receive antennas for the uplink. The uplink mobile M1 is equipped with $M_{\mathrm{ul}}$ transmit antennas and downlink mobile M2 is equipped with $N_{\mathrm{dl}}$ receive antennas. Besides the main-channel which includes uplink, downlink and interference link, there also exists an out-of-band wireless side-channel between the uplink mobile and downlink mobile.
\begin{figure}[h!]
  \centering
    \includegraphics[width=0.5\textwidth,trim = 45mm
      65mm 55mm 50mm, clip]{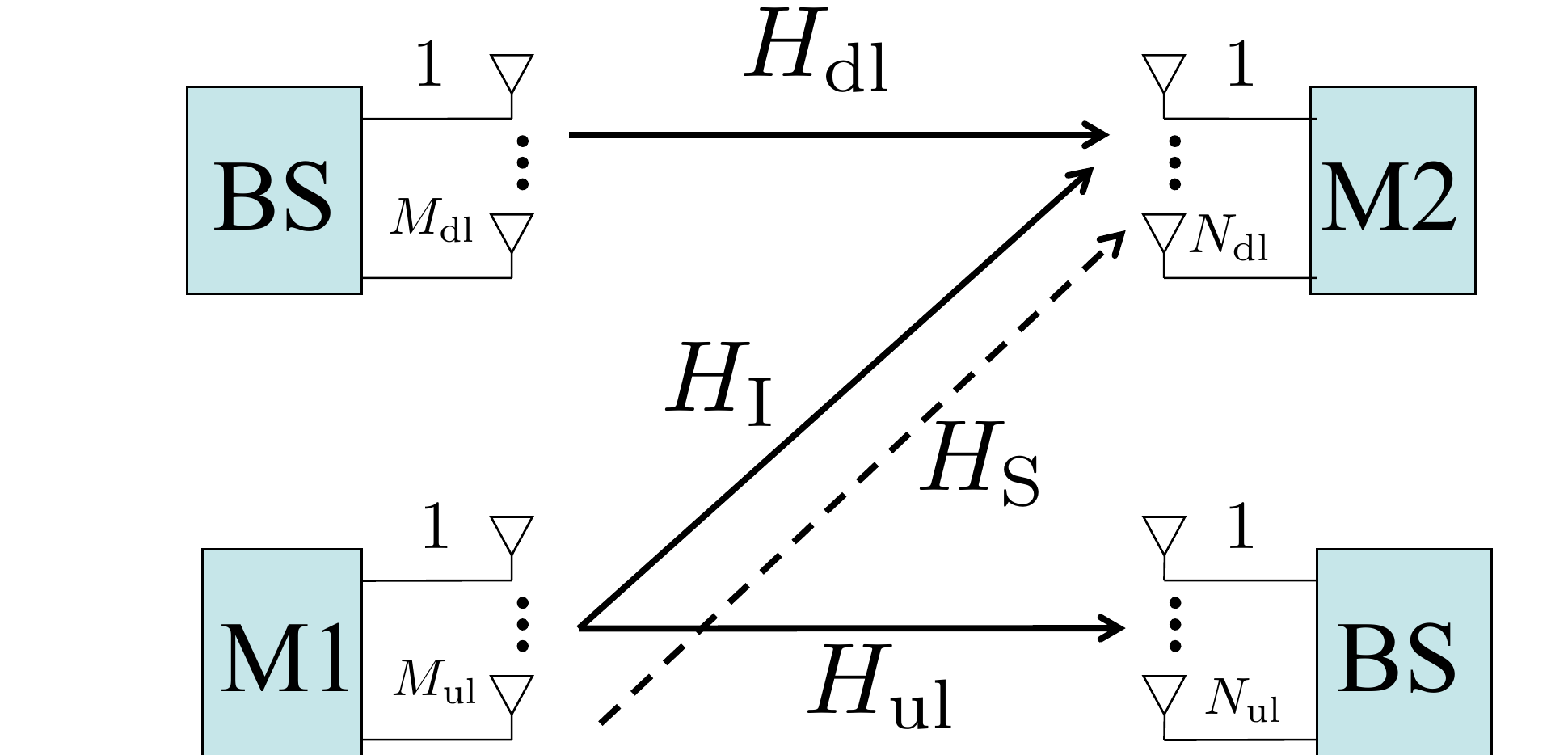}
  \caption{Channel model: $(M_{\mathrm{dl}},N_{\mathrm{dl}},M_{\mathrm{ul}},N_{\mathrm{ul}})$ side-channel assisted MIMO full-duplex network.}
\label{SystemModel}
\end{figure}
We refer to the channel model shown in Fig.~\ref{SystemModel} as $(M_{\mathrm{dl}},N_{\mathrm{dl}},M_{\mathrm{ul}},N_{\mathrm{ul}})$ side-channel assisted MIMO full-duplex network.
Let $W_m$ and $W_s$ denote the bandwidth of the main-channel and side-channel, respectively.
Parameter $W=\frac{W_s}{W_m}$ represents the bandwidth ratio of the side-channel to that of the main-channel.

Since one of the transmitter and receiver is co-located in the same node, the base-station BS, the uplink message received by the BS is causally known to the BS transmitter for downlink transmission. As a result, the side-channel assisted full-duplex network can be viewed as a Z-interference channel with implicit feedback and an out-of-band side-channel.

We assume that the channel parameters in our system model consist of two components: a small-scale fading factor due to multipath and a large-scale fading factor due to path loss. 
We denote the small-scale fading channels matrix as $\mathcal{H}=\{H_{\mathrm{dl}},H_{\mathrm{ul}},H_{\mathrm{I}},H_{\mathrm{S}}\}$, where each entry in $\mathcal{H}$ represents the small-scale fading channel matrix for the downlink, uplink, inter-mobile interference channel and the side-channel, as shown in Fig.~\ref{SystemModel}. We assume that all entries in $H_{k}$, where $k\in\{\mathrm{dl,ul,I,S}\}$, are mutually independent and identically distributed~(i.i.d.) according to $\mathcal{C}\mathcal{N}(0,1)$  and all channel matrices are full rank with probability one. We will consider two different scenarios for the small-scale fading. 
\begin{itemize}
\item \emph{Time-invariant channels}: $\mathcal{H}$ is fixed during the entire communication period.
\item \emph{Slow-fading channels}: $\mathcal{H}$ remains unchanged during each fade duration or coherence time, and varies i.i.d. between distinct fade periods. 
\end{itemize}

As for the large-scale fading factor, it captures the channel attenuation due to distance. Thus the channel attenuation between the transmitter and receiver is the same for every transmit-receive antenna pair. Hence  the channel attenuation for each channel is denoted by a scalar $\gamma_k$, where $k\in\{\mathrm{dl,ul,I,S}\}$.    
The transmitter at BS and uplink node M1 have a maximum power budget $P_{\rm dl}$ and $P_{\rm ul}$, respectively.
To simplify the notation, let $\rho_{\mathrm{dl}}=\gamma_{\mathrm{dl}}P_{\mathrm{dl}}$, $\rho_{\mathrm{ul}}=\gamma_{\mathrm{ul}}P_{\mathrm{ul}}$, $\rho_{\mathrm{S}}=\gamma_{\mathrm{S}}P_{\mathrm{ul}}$ and $\rho_{\mathrm{I}}=\gamma_{\mathrm{I}}P_{\mathrm{ul}}$, which denotes the average signal-to-noise ratio and interference-to-noise ratio at each receive antenna with additive Gaussian noise of unit variance.

Next, we describe the channel input-output relationships as follows.
\subsubsection{Uplink}
The node M1 will split the transmit power between main-channel and side-channel, i.e., $\bar{\lambda}P_{\rm ul}$ and ${\lambda}P_{\rm ul}$ for  main-channel and side-channel data transmission, respectively. We define $\bar{\lambda}=1-\lambda,\lambda\in[0,1]$.
Thus the received uplink signal $Y_{\mathrm{ul}}\in\mathbb{C}^{N_{\mathrm{ul}}\times 1}$ at BS is given by
\begin{gather}
\begin{aligned}
Y_{\mathrm{ul}}(t)&=\sqrt{\bar{\lambda}\rho_\mathrm{ul}}H_{\mathrm{ul}}X_{\mathrm{ul}}(t)+Z_{\mathrm{ul}}(t),
\end{aligned}
\end{gather}
where $X_{\mathrm{ul}}(t)\in\mathbb{C}^{M_{\mathrm{ul}}\times 1}$ is the uplink vector signal; $H_{\mathrm{ul}}\in\mathbb{C}^{N_{\mathrm{ul}}\times M_{\mathrm{ul}}}$ represents uplink channel and $Z_{\mathrm{ul}}(t)\in\mathbb{C}^{N_{\mathrm{ul}}\times 1}$ is the receiver additive Gaussian noise which contains i.i.d. $\mathcal{C}\mathcal{N}(0,1)$ entries.
\subsubsection{Downlink}
The received downlink signal $Y_{\mathrm{dl}}\in\mathbb{C}^{N_{\mathrm{dl}}\times 1}$
at the node M2 is a combination of the downlink signal and the interfering uplink signal, and is given by
\begin{gather}
\begin{aligned}
Y_{\mathrm{dl}}(t)&=\sqrt{\rho_\mathrm{dl}}H_{\mathrm{dl}}X_{\mathrm{dl}}(t)+\sqrt{\bar{\lambda}\rho_{\mathrm{I}}}H_{\mathrm{I}}X_{\mathrm{ul}}(t)+Z_{\mathrm{dl}}(t),
\end{aligned}
\end{gather}
where $X_{\mathrm{dl}}(t)\in\mathbb{C}^{M_{\mathrm{dl}}\times 1}$ is the downlink vector signal; $H_{\mathrm{dl}}\in\mathbb{C}^{N_{\mathrm{dl}}\times M_{\mathrm{dl}}}$ is the downlink channel matrix and $H_{\mathrm{I}}\in\mathbb{C}^{N_{\mathrm{dl}}\times M_{\mathrm{ul}}}$ is the inter-mobile interference channel matrix; $Z_{\mathrm{dl}}(t)\in\mathbb{C}^{N_{\mathrm{dl}}\times 1}$ is the receiver additive Gaussian noise which contains i.i.d. $\mathcal{C}\mathcal{N}(0,1)$ entries. 
\subsubsection{Side-channel}
We assume that the number of side-channel antennas are same as the main-channel. 
Thus the received signal $Y_{\mathrm{S}}\in\mathbb{C}^{N_{\mathrm{dl}}\times 1}$ at the node M2 is given by
\begin{gather}
\begin{aligned}
Y_{\mathrm{S}}(t)&=\sqrt{\lambda\rho_{\mathrm{S}}}H_{\mathrm{S}}X_{\mathrm{S}}(t)+Z_{\mathrm{S}}(t),
\end{aligned}
\end{gather}
where $X_{\mathrm{S}}(t)\in\mathbb{C}^{M_{\mathrm{ul}}\times 1}$ is the side-channel vector signal; $H_{\mathrm{S}}\in\mathbb{C}^{N_{\rm dl}\times M_{\rm ul}}$ is the channel matrix of the side-channel; $Z_{\mathrm{dl}}(t)\in\mathbb{C}^{N_{\mathrm{dl}}\times 1}$ is the Gaussian noise added to the side-channel which contains i.i.d. $\mathcal{C}\mathcal{N}(0,W)$ entries. Note that the noise variance of each entry in the side-channel is $W$ times larger than that in the main-channel.

The power constraint of the input signals is given as:
\begin{gather}
\frac{1}{L}\sum_{t=1+Lk}^{L(k+1)}\text{Trace}\bigg(\mathbb{E}[X_i(t)X_i(t)^\dagger]\bigg)\leq 1,~k\in\mathbb{N}, i\in\{\rm dl,ul,S\},\label{pcc}
\end{gather}
where in time-invariant channels, $k=0$, and $L$ denotes the entire communication duration; in slow-fading channels, $L$ denotes the coherence time.\footnote{In the rest of the paper, we omit the time-index t in the expressions.}

We define the strength level of different links with respect to nominal SNR, $\rho$, in decibels\footnote{We can set $\rho=\rho_{\mathrm{dl}}$ or $\rho_{\mathrm{ul}}$ such that either $\alpha_{\mathrm{dl}}=1$ or $\alpha_{\mathrm{ul}}=1$.}
\begin{gather}
\alpha_i=\frac{\mathrm{log}\rho_{i}}{\mathrm{log}\rho},~i\in\{\mathrm{dl,ul,I,S}\}.
\label{level}
\end{gather}
Note that the above normalization allows different links to have disparate strength.


\section{Vector Bin-and-cancel Scheme} \label{bcscheme}
A full-duplex node can be viewed as ``two nodes," with a co-located transmitter and receiver, that are connected by an \emph{infinite} capacity link. Inspired by this interpretation, in \cite{JingwenTWC}, we proposed a \emph{distributed full-duplex} architecture which is enabled by a wireless side-channel of \emph{finite} bandwidth when the transmitter and interfered receiver are not co-located.
When channel knowledge is known globally, we showed that a bin-and-cancel scheme achieves the capacity region to within 1~bit/s/Hz of the capacity region for all channel parameters in SISO case~\cite{JingwenTWC}. 

In this section, we will study the capacity region in MIMO case under different assumptions of channel uncertainty at the transmitter.
CSIT plays a critical role in MIMO interference channels. With CSIT, the transmitter can design the precoding matrix to steer the direction of the transmit signal to achieve higher rate. However, the cost of obtaining CSIT is also prohibitive since the receiver has to feed back the channel knowledge within the coherence time which incurs operational overhead. Thus it is crucial to explore the role of channel uncertainty at the transmitter in system performance. We assume that the receiver-side channel information is always available as the receiver can track the instantaneous channel from the training pilots. In what follows, we will study the capacity region in time-invariant channels. Next, we will present how CSIT and the use of side-channel is correlated, we also characterize the spatial and spectral tradeoff between multiple antennas at different nodes and spectral resources provided by side-channel.
\subsection{Capacity Region to Within a Constant Gap With CSIT}
\subsubsection{Outer Bound}
\begin{lemma}\label{outerbound}
Given the channel realization $\mathcal{H}$, the capacity region $\mathcal{C(H)}$ of the side-channel assisted MIMO full-duplex network is outer bounded by
\begin{gather}
\begin{aligned}
R_{\mathrm{dl}}&\leq W_m\bigg(\mathrm{log}\left|I_{N_{\mathrm{dl}}}+\rho_{\mathrm{dl}}H_{\mathrm{dl}}H_{\mathrm{dl}}^\dagger\right|\bigg)\triangleq \overline{C}_{\mathrm{dl}},\\
R_{\mathrm{ul}}&\leq W_m\bigg( \mathrm{log}\left|I_{N_{\mathrm{ul}}}+\bar{\lambda}\rho_{\mathrm{ul}}H_{\mathrm{ul}}H_{\mathrm{ul}}^\dagger\right|\bigg)\triangleq\overline{C}_{\mathrm{ul}},\\
R_{\mathrm{dl}}+R_{\mathrm{ul}}&\leq W_m\bigg(\mathrm{log}\left|I_{N_{\mathrm{dl}}}+\rho_{\mathrm{dl}}H_{\mathrm{dl}}H_{\mathrm{dl}}^\dagger+\bar{\lambda}\rho_{\mathrm{I}}H_{\mathrm{I}}H_{\mathrm{I}}^\dagger\right|+W\mathrm{log}\left|I_{N_{\mathrm{dl}}}+\frac{\lambda\rho_{\mathrm{S}}}{W}H_{\mathrm{S}}H_{\mathrm{S}}^{\dagger}\right|\\
&+\mathrm{log}\left|I_{N_{\mathrm{ul}}}+\bar{\lambda}\rho_{\mathrm{ul}}H_{\mathrm{ul}}(I_{M_{\mathrm{ul}}}+\bar{\lambda}\rho_{\mathrm{I}}H_{\mathrm{I}}^\dagger H_{\mathrm{I}})^{-1}H_{\mathrm{ul}}^\dagger\right|+N_{\mathrm{dl}}\bigg)\triangleq\overline{C}_{\mathrm{sum}},
 \label{ubequ}
\end{aligned}
\end{gather}
\end{lemma}
\begin{proof}
See Appendix~\ref{outer}. Note that if the interference channel~($\rho_{\rm I}$) or side-channel quality ($W\rho_{\rm S}$) exceeds certain threshold such that $C_{\rm sum}\geq C_{\rm dl}+C_{\rm ul}$, the capacity is just trivially outer bounded by the first two individual constraints in (\ref{ubequ}).
\end{proof}

\subsubsection{Achievable Rate Region}
A vector bin-and-cancel scheme based on a simple Han-Kobayashi coding strategy achieves the following rate region when CSIT is available. The scheme will be elucidated later in Section~\ref{vbc}.
\begin{lemma}\label{innerbound}
The achievable rate region $\mathcal{R}_{\rm BC}(\mathcal{H})$ of the side-channel assisted MIMO three-node full-duplex network for time-invariant channels is
\begin{gather}
\begin{aligned}
R_{\mathrm{dl}}&\leq \overline{C}_{\mathrm{dl}}-W_m c_1,\\
R_{\mathrm{ul}}&\leq \overline{C}_{\mathrm{ul}}-W_m c_2,\\
R_{\mathrm{dl}}+R_{\mathrm{ul}}&\leq \overline{C}_{\mathrm{sum}}-W_m(c_1+c_2),
 \label{ubequ1}
\end{aligned}
\end{gather}
where \begin{gather}
\begin{aligned}
c_{\mathrm{1}}&=\min\{M_{\mathrm{dl}}+M_{\mathrm{ul}},N_{\mathrm{dl}}\}\mathrm{log}(\max\{M_{\mathrm{dl}},M_{\mathrm{ul}}\})+\hat{m}_{\mathrm{I}},\\
c_{\mathrm{2}}&=(m_{\mathrm{ul}}+Wm_{\mathrm{I}})\mathrm{log}M_{\mathrm{ul}}+m_{X}\mathrm{log}(M_{\mathrm{ul}}+1),\hat{m}_{\mathrm{I}}=m_{\mathrm{I}}\mathrm{log}\left(1+\frac{1}{M_{\mathrm{ul}}}\right),\\
m_{\mathrm{dl}}&=\min\{M_{\mathrm{dl}},N_{\mathrm{dl}}\}, m_{\mathrm{ul}}=\min\{M_{\mathrm{ul}},N_{\mathrm{ul}}\}, m_{X}\!=\max\{M_{\mathrm{ul}},N_{\mathrm{dl}}\},m_{\mathrm{I}}=\min\{M_{\mathrm{ul}},N_{\mathrm{dl}}\}.\label{gap}
\end{aligned}
\end{gather}
\end{lemma}
\begin{proof}
See Section~\ref{vbc} for description of the achievability and Appendix~\ref{inner} for the rate calculation.
\end{proof}
Based on the lemmas above, we will state the result of constant-bit gap to capacity region under time-invariant channels in the following theorem.
\begin{theorem} \label{mainthe1}
For the side-channel assisted two-user MIMO full-duplex network under time-invariant channels, the achievable rate region $\mathcal{R}_{\rm BC}(\mathcal{H})$ is within $\max\{c_1,c_2\}$ bit/s/Hz of the capacity region $\mathcal{C(H)}$, where $c_i,~i=1,2$ is given in (\ref{gap}).
\end{theorem}
\begin{proof}
The proof is straightforward. From Lemma~\ref{outerbound} and Lemma~\ref{innerbound}, we can calculate the rate difference and divide it by the total bandwidth $W_m+W_s$ of the system.
In other word, for any given rate pair $(R_{\mathrm{dl}},R_{\mathrm{ul}})\in\mathcal{C(H)}$~(bit/s), the rate pair $\big((R_{\mathrm{dl}}-(W_m+W_s)c_1)^+,(R_{\mathrm{ul}}-(W_m+W_s)c_2)^+\big)$ is achievable in $\mathcal{R}_{\rm BC}(\mathcal{H})$.
\end{proof}
In the SISO case, we can easily verify that the vector bin-and-cancel achieves the capacity region to within one bit.
\subsection{Achievability}\label{vbc}
In this section, we will describe the vector bin-and-cancel scheme used to show the achievability in Lemma~\ref{innerbound}. 
In vector bin-and-cancel, we use Han-Kobayashi \cite{han1981new} style common-private message splitting with a simple power splitting.
The common message can be decoded at both receivers while the private message can only be decoded at the intended receiver. The downlink message $\omega_{\mathrm{dl}}$ only consists of 
private message for the downlink receiver which is of size $2^{nR_{\rm dl}}$, and is encoded into codeword $X_{\mathrm{dl}}$. The uplink message is divided into the common part $\omega_{\mathrm{ul},c}$ of size $2^{nR_{\mathrm{ul},c}}$
and the private part  $\omega_{\mathrm{ul},p}$ of size $2^{nR_{\mathrm{ul},p}}$. The uplink codeword is then obtained by superposition of the codewords of both $\omega_{\mathrm{ul},c}$ and $\omega_{\mathrm{ul},p}$,
\begin{gather}
\begin{aligned}
X_{\mathrm{ul}}=S_{\mathrm{ul}}+U_{\mathrm{ul}}, \nonumber
\end{aligned}
\end{gather}
where $S_{\mathrm{ul}}$ and $U_{\mathrm{ul}}$ are the codewords of uplink common message $\omega_{\mathrm{ul},c}$ and private message~$\omega_{\mathrm{ul},p}$, respectively.

Next, we partition the uplink common message $\omega_{\mathrm{ul},c}$: the common message set is divided into $2^{nR_{\rm S}}$ equal size bins such that $\mathcal{B}(l)=\left[(l-1)2^{n(R_{\mathrm{ul},c}-R_{\rm S})}+1:l2^{n(R_{\mathrm{ul},c}-R_{\rm S})}\right], l\in[1:2^{nR_{\rm S}}]$. The total number of bin indices $2^{nR_{\rm S}}$ is determined by the strength of the side-channel,~$\alpha_{\mathrm{S}}$, and the bandwidth ratio~$W$. The bin index $l$ is then encoded into codeword $X_{\mathrm{S}}$ and sent from the uplink transmit antenna arrays over the side-channel, which is shown in Fig.~\ref{BCP1}.

All the codewords are mutually independent complex Gaussian random vectors with covariance matrices given as follows to satisfy the power constraint given in (\ref{pcc}):
\begin{gather}
\begin{aligned}
\mathbb{E}(X_{\mathrm{dl}}X_{\mathrm{dl}}^\dagger)&=\frac{1}{M_{\mathrm{dl}}}I_{M_{\mathrm{dl}}},~~ \mathbb{E}(U_{\mathrm{ul}}U_{\mathrm{ul}}^\dagger)=\frac{1}{M_{\mathrm{ul}}}(I_{M_{\mathrm{ul}}}+\bar{\lambda}\rho_{\mathrm{I}}H_{\mathrm{I}}^\dagger H_{\mathrm{I}})^{-1}\\
\mathbb{E}(S_{\mathrm{ul}}S_{\mathrm{ul}}^\dagger)&=\frac{1}{M_{\mathrm{ul}}}I_{M_{\mathrm{ul}}}-\mathbb{E}(U_{\mathrm{ul}}U_{\mathrm{ul}}^\dagger),~~\mathbb{E}(X^s_{\mathrm{ul}}X_{\mathrm{ul}}^{s\dagger})=\frac{1}{M_{\mathrm{ul}}}I_{M_{\mathrm{ul}}},
\label{powerallocate}
\end{aligned}
\end{gather}
where $\lambda\in(0,1), \bar{\lambda}+\lambda=1$.
The parameter $\lambda$ denotes the fraction of power allocated to the side-channel. For the power splitting between the uplink private and common message, we set the power of the private message such that its received signal strength is below the noise floor at each unintended receiver's antenna. And we allocate the power of the codewords equally among the transmit antenna array.
\begin{figure}[h!]
\begin{minipage}[b]{0.5\linewidth}
  \centering
     \scalebox{0.3}{\includegraphics[trim = 50mm
      50mm 45mm 52mm, clip]{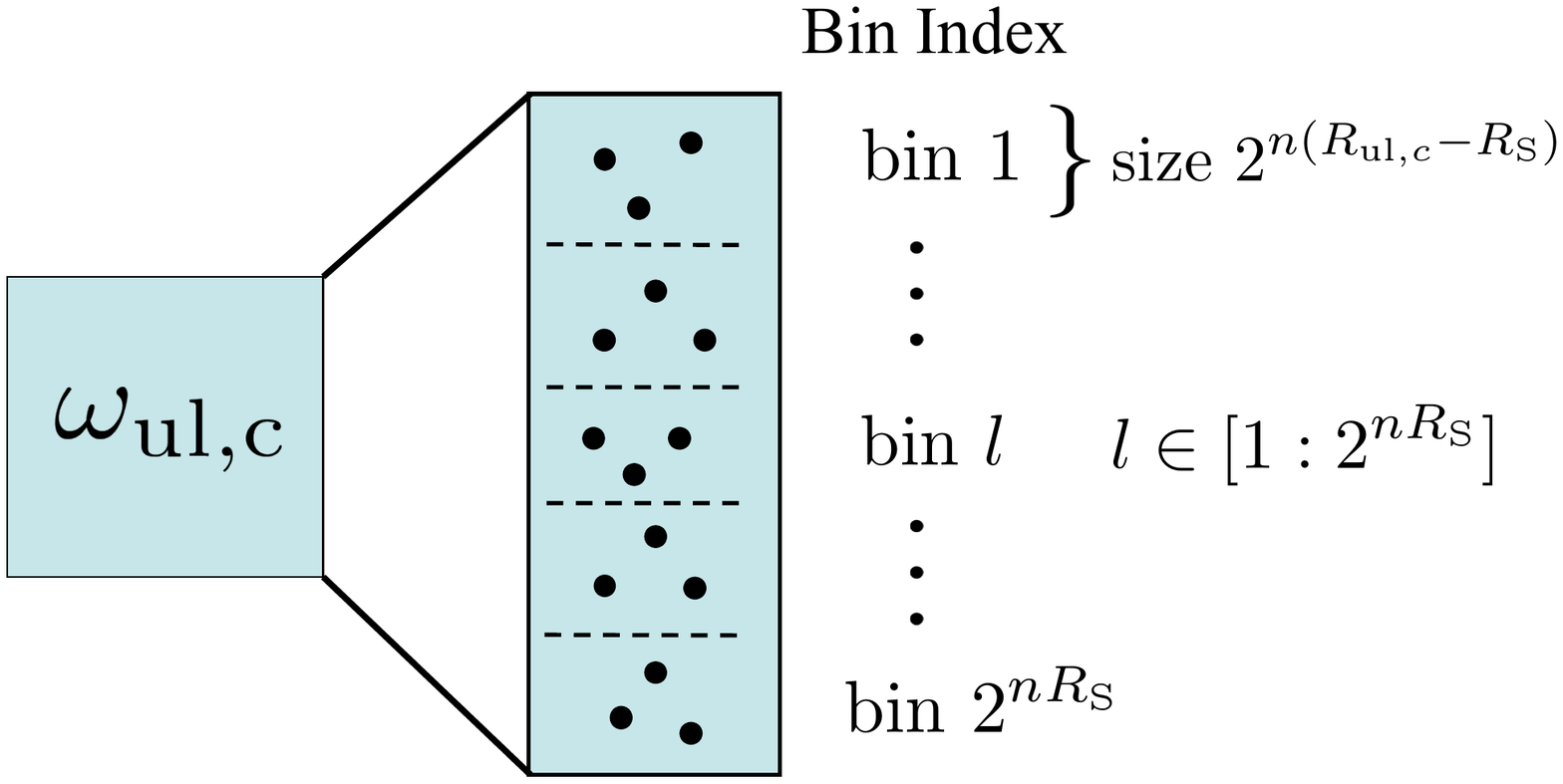}}
  \caption{Binning of the common message at uplink transmitter.}
\label{BCP1}
\end{minipage}
\hspace{0.1cm}
\begin{minipage}[b]{0.4\linewidth}
  \centering
    \scalebox{0.3}{\includegraphics[trim = 50mm
      55mm 50mm 53mm, clip]{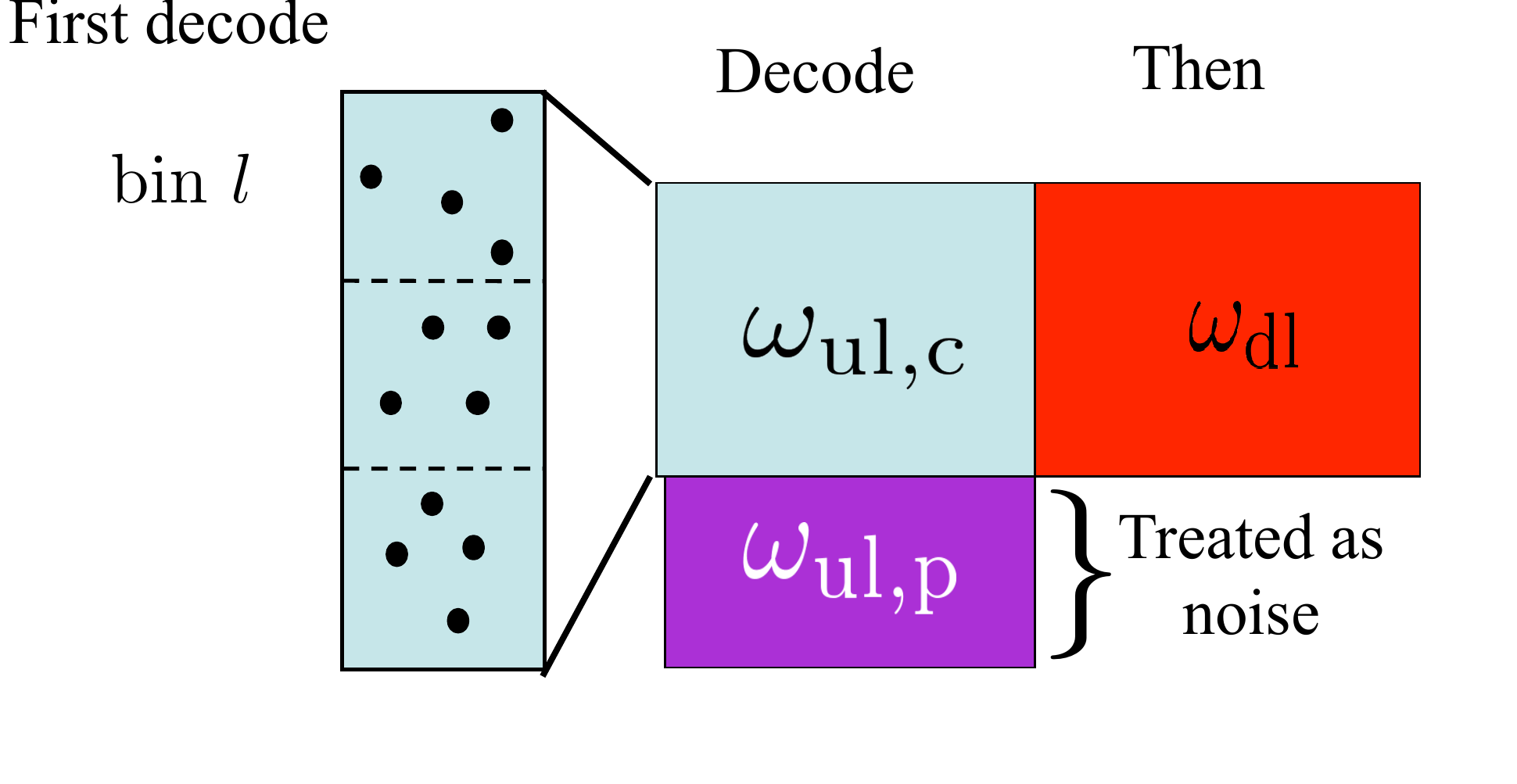}}
  \caption{Decoding at downlink receiver.}
\label{BCP2}
\end{minipage}
\end{figure}

Now we describe the decoding process. The decoding at the BS is straightforward. Upon receiving $Y_{\mathrm{ul}}$, the BS decodes ($\omega_{\mathrm{ul},c}, \omega_{\mathrm{ul},p}$). The achievable rate region of ($R_{\mathrm{ul},c},R_{\mathrm{ul},p}$) is the capacity region of multiple-access channel denoted as $\mathcal{C}_1$, where
\begin{gather}
\begin{aligned}
R_{\mathrm{ul},c}&\leq I(S_{\mathrm{ul}};Y_{\mathrm{dl}}|X_{\mathrm{dl}})\\
R_{\mathrm{ul},p}&\leq I(U_{\mathrm{ul}};Y_{\mathrm{ul}}|S_{\mathrm{ul}})\\
R_{\mathrm{ul},c}+R_{\mathrm{ul},p}&\leq I(S_{\mathrm{ul}},U_{\mathrm{ul}};Y_{\mathrm{ul}})
\label{discreteAR}
\end{aligned}
\end{gather}

The decoding at the downlink receiver has two stages as shown in Fig.~\ref{BCP2}. In stage one, upon receiving  $Y_{\mathrm{S}}$, the downlink receiver first decodes the bin index $l$ from the side-channel. In stage two, upon receiving $Y_{\mathrm{dl}}$, the downlink receiver decodes ($\omega_{\rm dl}, \omega_{\mathrm{ul},c}$) with the help of side-channel information while treating uplink private message $\omega_{\mathrm{ul},p}$ as noise.\footnote{With the assistance of the bin index, more uplink common message can be decoded which otherwise is restricted by the interference link.} This is a multiple-access channel~(MAC) with side-channel whose capacity region denoted as $\mathcal{C}_2$ is given in~\cite{JingwenTWC} (see Lemma 1), hence  we have
\begin{gather}
\begin{aligned}
R_{\mathrm{dl}}&\leq I(X_{\mathrm{dl}};Y_{\mathrm{dl}}|S_{\mathrm{ul}})\\
R_{\mathrm{ul},c}&\leq I(S_{\mathrm{ul}};Y_{\mathrm{dl}}|X_{\mathrm{dl}})+I(X_{\mathrm{S}};Y_{\mathrm{S}})\\
R_{\mathrm{dl}}+R_{\mathrm{ul},c}&\leq I(X_{\mathrm{dl}},S_{\mathrm{ul}};Y_{\mathrm{dl}})+I(X_{\mathrm{S}};Y_{\mathrm{S}}).
\label{discreteAR}
\end{aligned}
\end{gather}
The achievable rate region of side-channel assisted full-duplex network is the set of all $(R_{\rm dl},R_{\rm ul})$ such that $R_{\rm dl}, R_{\rm ul} = R_{\mathrm{ul},c} +R_{\mathrm{ul},p}$ satisfying that $(R_{\mathrm{ul},c},R_{\mathrm{ul},p})\in \mathcal{C}_1$ and $(R_{\mathrm{dl}},R_{\mathrm{ul},c})\in \mathcal{C}_2$. Using Fourier-Motzkin elimination, the achievable rate pairs $(R_{\rm dl},R_{\rm ul})$ are constrained by the following rate region
\begin{gather}
\begin{aligned}
R_{\mathrm{dl}}&\leq I(X_{\mathrm{dl}};Y_{\mathrm{dl}}|S_{\mathrm{ul}})\\
R_{\mathrm{ul}}&\leq \min \{I(S_{\mathrm{ul}},U_{\mathrm{ul}};Y_{\mathrm{ul}}),I(U_{\mathrm{ul}};Y_{\mathrm{ul}}|S_{\mathrm{ul}})+I(S_{\mathrm{ul}};Y_{\mathrm{dl}}|X_{\mathrm{dl}})+I(X_{\mathrm{S}};Y_{\mathrm{S}})\}\\
R_{\mathrm{dl}}+R_{\mathrm{ul}}&\leq I(U_{\mathrm{ul}};Y_{\mathrm{ul}}|S_{\mathrm{ul}})+I(X_{\mathrm{dl}},S_{\mathrm{ul}};Y_{\mathrm{dl}})+I(X_{\mathrm{S}};Y_{\mathrm{S}}).
\label{discreteAR}
\end{aligned}
\end{gather}
The achievable rate region given above is calculated in Appendix~\ref{inner}, thus we can obtain the explicit achievable rate expression in Lemma~2.
\subsection{High SNR Approximation}
From Theorem~\ref{mainthe1}, vector bin-and-cancel scheme achieves the capacity region to within a constant bit for \emph{all} values of channel parameters under time-invariant channels. In the high SNR limit, a constant number of bits~(which do not vary with respect to SNR) are insignificant and
can be ignored on the scale of interest. Therefore we can establish the high SNR capacity region approximation to within~$\mathcal{O}(1)$ in the following corollary.
\begin{corollary}\label{capacity}
For a given the channel realization $\mathcal{H}$, vector bin-and-cancel is asymptotically capacity achieving and the asymptotic capacity approximation $\mathcal{C(H)}$ is given by
\begin{gather}
\begin{aligned}
\mathcal{C(H)}\doteq\Bigg\{(R_{\mathrm{dl}},R_{\mathrm{ul}}):R_{\mathrm{dl}}&\leq W_m\mathrm{log}\left|I_{N_{\mathrm{dl}}}+\rho_{\mathrm{dl}}H_{\mathrm{dl}}H_{\mathrm{dl}}^\dagger\right|\triangleq C_{\mathrm{dl}},\\
R_{\mathrm{ul}}&\leq W_m\mathrm{log}\left|I_{N_{\mathrm{ul}}}+\bar{\lambda}\rho_{\mathrm{ul}}H_{\mathrm{ul}}H_{\mathrm{ul}}^\dagger\right|\triangleq C_{\mathrm{ul}},\\
R_{\mathrm{dl}}+R_{\mathrm{ul}}&\leq W_m\bigg(\mathrm{log}\left|I_{N_{\mathrm{dl}}}+\rho_{\mathrm{dl}}H_{\mathrm{dl}}H_{\mathrm{dl}}^\dagger+\bar{\lambda}\rho_{\mathrm{I}}H_{\mathrm{I}}H_{\mathrm{I}}^\dagger\right|+W\mathrm{log}\left|I_{N_{\mathrm{dl}}}+\frac{\lambda\rho_{\mathrm{S}}}{W}H_{\mathrm{S}}H_{\mathrm{S}}^{\dagger}\right|\\
&+\mathrm{log}\left|I_{N_{\mathrm{ul}}}+\bar{\lambda}\rho_{\mathrm{ul}}H_{\mathrm{ul}}(I_{M_{\mathrm{ul}}}+\bar{\lambda}\rho_{\mathrm{I}}H_{\mathrm{I}}^\dagger H_{\mathrm{I}})^{-1}H_{\mathrm{ul}}^\dagger\right|\bigg)\triangleq C_{\mathrm{sum}}\Bigg\}.
\label{hignSNRcap}
\end{aligned}
\end{gather}
\end{corollary}

The high SNR capacity approximation can be used to derive the generalized degrees
of freedom~($\mathsf{GDoF}$). The $\mathsf{GDoF}$ captures the asymptotic behavior of the capacity and the corresponding optimal schemes, allowing different links to grow at disparate rates.

The $\mathsf{GDoF}$ region is defined as follows \footnote{Notice that our definition deviates slightly from the conventional definition of $\mathsf{GDoF}$ in that we account for the asymmetric bandwidths of different links and the rate is calculated as bit/s instead of bit/s/Hz.}
\begin{gather}
\begin{aligned}
\bigg\{(\mathsf{DoF}_{\mathrm{dl}},\mathsf{DoF}_{\mathrm{ul}}):\mathsf{DoF}_{i}=\lim_{\rho\rightarrow\infty}\frac{R_i(\rho_i)}{W_m\mathrm{log}\rho},~i\in\{\mathrm{dl,ul}\}~ \text{and}~(R_{\mathrm{dl}},R_{\mathrm{ul}})\in\mathcal{C(H)}\bigg\},
\end{aligned}
\end{gather}
where $W_m\mathrm{log}\rho$ is the point-to-point main-channel capacity with nominal $\mathsf{SNR}$ in bit/s. $\mathsf{DoF}_{\mathrm{dl}}$ and $\mathsf{DoF}_{\mathrm{ul}}$ denote the degrees of freedom~($\mathsf{DoF}$) of downlink and uplink, respectively. Using high SNR capacity approximation, we state the $\mathsf{GDoF}$ region as follows.
\begin{corollary}\label{csitgdof}
Assuming $\alpha_{\mathrm{dl}}=\alpha_{\mathrm{ul}}=1$, the $\mathsf{GDoF}$ region of $(M_{\mathrm{dl}},N_{\mathrm{dl}},M_{\mathrm{ul}},N_{\mathrm{ul}})$ side-channel assisted MIMO full-duplex network satisfies the following constraints
\begin{gather}
\begin{aligned}
  &\mathsf{DoF}_{\mathrm{dl}}\leq m_{\mathrm{dl}},~~\mathsf{DoF}_{\mathrm{ul}}\leq m_{\mathrm{ul}},\\
  &\mathsf{DoF}_{\mathrm{dl}}+\mathsf{DoF}_{\mathrm{ul}}\leq f\Big(N_{\mathrm{ul}},\big((1-\alpha_{\mathrm{I}})^+,m_{\mathrm{I}}\big),\big(1,(M_{\mathrm{ul}}-N_{\mathrm{dl}})^+\big)\Big)\\
  &+f\big(N_{\mathrm{dl}},(\alpha_{\mathrm{I}},M_{\mathrm{ul}}),(1,M_{\mathrm{dl}})\big)+Wf\big(N_{\mathrm{dl}},(\alpha_{\mathrm{S}},M_{\mathrm{ul}})\big),
\end{aligned}
\end{gather}
where $m_{\mathrm{dl}}=\min\{M_{\mathrm{dl}},N_{\mathrm{dl}}\}, m_{\mathrm{ul}}=\min\{M_{\mathrm{ul}},N_{\mathrm{ul}}\}, m_{\mathrm{I}}=\min\{M_{\mathrm{ul}},N_{\mathrm{dl}}\}$ as defined in (\ref{gap}); function $f\big(x,(y_1,x_1),(y_2,x_2)\big)=\min\{x,x_1\}y_1^++\min\{(x-x_1)^+,x_2\}y_2^+$ for $y_1\geq y_2$.
\end{corollary}
\begin{proof}
The proof is akin to~\cite{GDOFMIMOIC}~(see Appendix C), so we will only provide an interpretation of the $\mathsf{GDoF}$ result here.

First, the $\mathsf{DoF}$ of downlink and uplink is limited by the number of transmit and receive antennas, much like the point-to-point MIMO channel.
Next we will explain the sum $\mathsf{GDoF}$. Let $\mathsf{DoF}_{\mathrm{ul},c}$ and $\mathsf{DoF}_{\mathrm{ul},p}$ denote the $\mathsf{DoF}$ of the uplink common message and private message, respectively. 


Adopting the singular value decomposition~(SVD), we can decompose the interference channel as $H_{\mathrm{I}}=U\Lambda V^{\dagger}$, where $U$ and $V$ are $N_{\mathrm{dl}}\times N_{\mathrm{dl}}$ and $M_{\mathrm{ul}}\times M_{\mathrm{ul}}$ unitary matrices, respectively, $\Lambda$ is $N_{\mathrm{dl}}\times M_{\mathrm{ul}}$ diagonal matrix containing singular values of $H_{\mathrm{I}}$. Thus $H_{\mathrm{I}}$ is decomposed into $m_{\mathrm{I}}$ parallel channels, leaving $(M_{\rm ul}-m_{\rm I})^+=(M_{\mathrm{ul}}-N_{\mathrm{dl}})^+$ effective inputs at uplink transmitter that do not cause any interference to the downlink receiver.
The uplink transmitter divides the private streams into two parts.  The first part is sent along the $(M_{\mathrm{ul}}-N_{\mathrm{dl}})^+$-dimensional null space of interference channel $H_{\mathrm{I}}$ and reaches BS at an SNR of $\rho$ with $N_{\mathrm{ul}}$ receive antennas. In the remaining $m_{\rm I}$ dimensions, the second part is transmitted at a power level of $\rho^{-\alpha_{\mathrm{I}}}$ such that it reaches the unintended receiver at the noise floor and reaches BS at an SNR of $\rho^{(1-\alpha_{\mathrm{I}})^+}$.  The process can be viewed as a combination of signal space and signal scale interference alignment.
Thus the $\mathsf{DoF}$ of the uplink private message is
\begin{gather}
\begin{aligned}
\mathsf{DoF}_{\mathrm{ul},p}=f\Big(N_{\mathrm{ul}},\big((1-\alpha_{\mathrm{I}})^+,m_{\mathrm{I}}\big),\big(1,(M_{\mathrm{ul}}-N_{\mathrm{dl}})^+\big)\Big).
\label{sum1}\end{aligned}
\end{gather}

Since the common message can be decoded at both receivers, the downlink receiver with $N_{\mathrm{dl}}$ receive antennas is a side-channel assisted multiple access channel receiver. The downlink message $\omega_{\mathrm{dl}}$ reaches the downlink receiver at an SNR of $\rho$ with $M_{\mathrm{dl}}$ transmit antennas. The uplink common message $\omega_{\mathrm{ul},c}$ reaches the downlink receiver through both main-channel at an SNR of $\rho^{\alpha_{\mathrm{I}}}$ and side-channel as an orthogonal spectral space at an SNR of $\rho^{W\alpha_{\mathrm{S}}}$ with $M_{\mathrm{ul}}$ transmit antennas. Thus we have 
\begin{gather}
\begin{aligned}
\mathsf{DoF}_{\mathrm{dl}}+\mathsf{DoF}_{\mathrm{ul},c}=f\big(N_{\mathrm{dl}},(\alpha_{\mathrm{I}},M_{\mathrm{ul}}),(1,M_{\mathrm{dl}})\big)+Wf\big(N_{\mathrm{dl}},(\alpha_{\mathrm{S}},M_{\mathrm{ul}})\big). \label{gdof2}
\end{aligned}
\end{gather}
Combining (\ref{sum1}) and (\ref{gdof2}) leads to the sum $\mathsf{GDoF}$.
\end{proof}

\begin{remark}
When $W=0$, i.e., there is no side-channel, the $\mathsf{GDoF}$ is the same as that of MIMO Z-interference channel in \cite{GDOFMIMOIC}, hence we conclude that the implicit feedback at the full-duplex capable BS does not help improve $\mathsf{GDoF}$ regime in the two-user MIMO full-duplex network. 
This is due to the fact there is only one-sided interference. When $W>0$, the implicit feedback is still not useful in terms of $\mathsf{GDoF}$. because our scheme does not rely on any feedback.
\end{remark}

\subsection{Special Cases}
In this section, we give several special cases to illustrate the $\mathsf{GDoF}$ results above.

\begin{theorem}(Case A)\label{sc1}
When $M_{\rm dl}=M_{\rm ul}=M, N_{\rm dl}=N_{\rm ul}=N$, and $\alpha_{\rm ul}=\alpha_{\rm dl}=1$, the sum $\mathsf{GDoF}$ per antenna denoted as $\frac{\mathsf{GDoF}_{\rm sum}}{\min(M,N)}$ for the symmetric side-channel assisted MIMO full-duplex network is given by
\begin{gather}
\begin{aligned}
\frac{\mathsf{GDoF}_{\rm sum}}{\min(M,N)}=
  \begin{cases}
\min\Big\{2, 2-\left(2-\frac{\max(M,N)}{\min(M,N)}\right)^+\alpha_{\rm I}+W\alpha_{\rm S}\Big\}&\alpha_{\rm I}<1, \\
  \min\Big\{2,\alpha_{\rm I}+\frac{\max(M,N)}{\min(M,N)}-1+W\alpha_{\rm S})\Big\}&\alpha_{\rm I}\geq 1.
   \end{cases}\nonumber
\end{aligned}
\end{gather}
\end{theorem}
In this case, one can observe that the sum $\mathsf{GDoF}$ per antenna increases linearly with the antenna ratio $\frac{\max(M,N)}{\min(M,N)}$ and side-channel quality $W\alpha_{\rm S}$.

Another case of interest is when the BS has more  antennas than mobile clients, i.e., $M_{\mathrm{dl}}, N_{\mathrm{ul}}\geq M_{\mathrm{ul}}, N_{\mathrm{dl}}$. This scenario is almost always true in practical systems and the ongoing trend is that the BS can accommodate many antennas such as in massive MIMO systems~\cite{larsson2013massive}, while the small-form factor mobiles will have a relatively fewer antennas due to its physical size constraint.
\begin{theorem}(Case B)\label{sc2}
When BS has more antennas than mobiles, i.e., $M_{\mathrm{dl}}, N_{\mathrm{ul}}\geq M_{\mathrm{ul}}, N_{\mathrm{dl}}$ with $\alpha_{\rm ul}=\alpha_{\rm dl}=1$, the sum $\mathsf{GDoF}$ per antenna denoted as $\frac{\mathsf{GDoF}_{\rm sum}}{\min(M_{\rm ul},N_{\rm dl})}$ is given as
\begin{gather}
\begin{aligned}
\frac{\mathsf{GDoF}_{\rm sum}}{\min(M_{\rm ul},N_{\rm dl})}=
  \begin{cases}
\min\Big\{\frac{m_X}{m_{\rm I}}+1, \frac{m_X}{m_{\rm I}}+1-\alpha_{\rm I}+W\alpha_{\rm S}\Big\}&\alpha_{\rm I}<1 \\
  \min\Big\{\frac{m_X}{m_{\rm I}}+1,\frac{m_X}{m_{\rm I}}-1+\alpha_{\rm I}+W\alpha_{\rm S}\Big\}&\alpha_{\rm I}\geq 1.
   \end{cases}\nonumber
\end{aligned}
\end{gather}
where $m_X=\max(M_{\rm ul},N_{\rm dl}), m_{\rm I}=\min (M_{\rm ul},N_{\rm dl})$.
\end{theorem}
\begin{figure}[h!]
  \centering
    \includegraphics[width=0.4\textwidth,trim = 0mm
      53mm 8mm 65mm, clip]{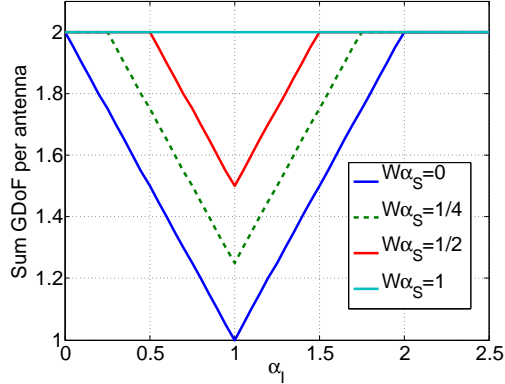}
  \caption{The sum $\mathsf{GDoF}$ per antenna for $M_{\rm ul}=N_{\rm dl}$ when BS has an excess of antennas. }
\label{gdofgain}
\end{figure}
Figure~\ref{gdofgain} illustrates how the sum $\mathsf{GDoF}$ per antenna varies as the side-channel quality changes when $M_{\rm ul}=N_{\rm dl}$ given an excess of antennas at BS. When $W\alpha_{\rm S}= 0$, i.e., there is no side-channel, the curve maintains ``V'' shape as in the Z-interference channel. When $W\alpha_{\rm S}$ increases, the curve gradually becomes a lifted ``V'' and finally reach the maximum sum $\mathsf{GDoF}$ per antenna of 2 for all regimes that one can achieve without interference.

We also give an example to clarify the  $\mathsf{DoF}$ of  vector bin-and-cancel in Case B assuming $\alpha_{\rm I}=\alpha_{\rm S}=1$. Using the standard MIMO SVD of channel matrices, the interference channel and side-channel can be converted to $ m_{\mathrm{I}}=\min\{N_{\mathrm{dl}},M_{\mathrm{ul}}\}$ parallel paths from uplink node $\mathrm{Tx_U}$ to downlink node $\mathrm{Rx_D}$.
In Fig.~\ref{zfdia}, the diagonalized interference and side-channel paths are depicted in bold.
\begin{figure}[h!]
  \centering
    \includegraphics[width=0.52\textwidth,trim = 51mm
      40mm 35mm 52mm, clip]{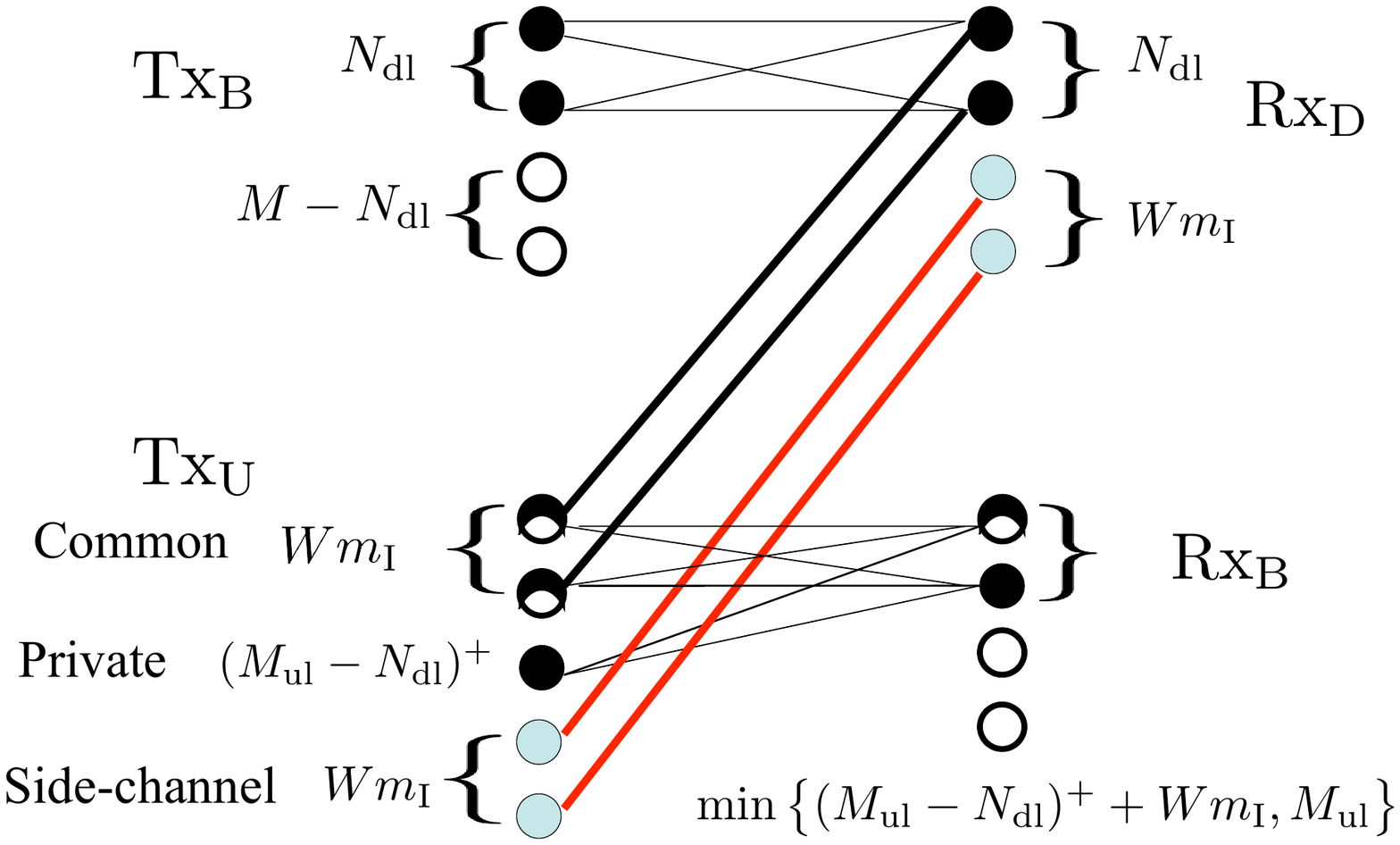}
  \caption{The $\mathsf{DoF}$-optimal scheme of two-user side-channel assisted MIMO full-duplex network when $M_{\rm ul}\geq N_{\rm dl}$.}
\label{zfdia}
\end{figure}

In Fig.~\ref{zfdia}, the base station $\mathrm{Tx_B}$ sends $N_{\mathrm{dl}}$ independent streams to downlink node $\mathrm{Rx_D}$, which is indicated by the black circles. Uplink node $\mathrm{Tx_U}$ sets $(1-W)m_{\mathrm{I}}$ effective inputs \footnote{The effective input is a product of the unitary matrices by SVD and the initial input vector.} to zero, which is indicated by the white circles; $\mathrm{Tx_U}$ then sends $(M_{\mathrm{ul}}-N_{\mathrm{dl}})^+$ independent private streams in the null space of the signal from $\mathrm{Tx_B}$, and $Wm_{\mathrm{I}}$ common message which can be heard at $\mathrm{Rx_D}$. Using vector bin-and-cancel, each transmitter sends $Wm_{\mathrm{I}}$ streams of its common message to the interfering receiver through the side-channel, which is indicated by the blue circles. At the downlink receiver $\mathrm{Rx_D}$, $Wm_{\mathrm{I}}$ streams of the interfering message can be canceled out, thus downlink can achieve $N_{\mathrm{dl}}$ $\mathsf{DoF}$s and uplink can achieve $\min\big\{(M_{\mathrm{ul}}-N_{\mathrm{dl}})^++Wm_{\mathrm{I}},M_{\mathrm{ul}}\big\}$ $\mathsf{DoF}$s. Thus, in total, we can obtain $\min\big\{\max\{N_{\mathrm{dl}},M_{\mathrm{ul}}\}+Wm_{\mathrm{I}},N_{\mathrm{dl}}+M_{\mathrm{ul}}\big\}$ $\mathsf{DoF}$s.

\subsection{GDoF Without CSIT}

Acquiring the CSIT incurs a large overhead, especially in a MIMO system with many antennas. Hence it is of practical interest to study the  $\mathsf{GDoF}$ performance of the system without CSIT. 

We first describe the encoding and decoding strategy under the no-CSIT assumption. Both transmitters encode their messages using independent Gaussian codebooks for the main-channel. The uplink transmitter sends common message only, and applies vector bin-and-cancel scheme. The side-channel bins all the uplink message and encodes the bin indices using an independent Gaussian codebook. From the downlink user's perspective, the channel is a MAC with side-channel. At the decoding process, the downlink user uses joint maximum likelihood~(ML) decoder to decode both downlink message and uplink messages with the help of side-channel. Hence we can obtain the achievable rate region $\mathcal{R}^{\text{No-CSIT}}$ as 
\begin{gather}
\begin{aligned}
\mathcal{R}^{\text{No-CSIT}}&=\Bigg\{(R_{\mathrm{dl}},R_{\mathrm{ul}}):R_{\mathrm{dl}}\leq W_m\mathrm{log}\left|I_{N_{\mathrm{dl}}}+\frac{\rho_{\mathrm{dl}}}{M_{\mathrm{dl}}}H_{\mathrm{dl}}H_{\mathrm{dl}}^\dagger\right|,\\
R_{\mathrm{ul}}&\leq W_m\min\left\{\mathrm{log}\left|I_{N_{\mathrm{ul}}}+\frac{\bar{\lambda}\rho_{{\mathrm{ul}}}}{M_{\mathrm{ul}}}H_{\mathrm{ul}}H_{\mathrm{ul}}^\dagger\right|,\mathrm{log}\left|I_{N_{\mathrm{dl}}}+\frac{\bar{\lambda}\rho_{\mathrm{I}}}{M_{\mathrm{ul}}}H_{\mathrm{I}}H_{\mathrm{I}}^\dagger\right|+W\mathrm{log}\left|I_{N_{\mathrm{dl}}}+\frac{\lambda\rho_{\mathrm{S}}}{W M_{\mathrm{ul}}}H_{\mathrm{S}}H_{\mathrm{S}}^{\dagger}\right|\right\},\\
R_{\mathrm{dl}}+R_{\mathrm{ul}}&\leq W_m\bigg(\mathrm{log}\left|I_{N_{\mathrm{dl}}}+\frac{\rho_{\mathrm{dl}}}{M_{\mathrm{dl}}}H_{\mathrm{dl}}H_{\mathrm{dl}}^\dagger+\frac{\bar{\lambda}\rho_{\mathrm{I}}}{M_{\mathrm{ul}}}H_{\mathrm{I}}H_{\mathrm{I}}^\dagger\right|+W\mathrm{log}\left|I_{N_{\mathrm{dl}}}+\frac{\lambda\rho_{\mathrm{S}}}{W M_{\mathrm{ul}}}H_{\mathrm{S}}H_{\mathrm{S}}^{\dagger}\right|\bigg)\Bigg\}, \label{nocsitar}
\end{aligned}
\end{gather}
where $\lambda\in(0,1)$, for instance, we can fix $\lambda=\bar{\lambda}=0.5$. The achievable rate region given above can be calculated easily from Equation~(\ref{discreteAR}) with uplink private message set to null and equal power allocation among transmit antennas which does not require any CSIT.

Now we can obtain the lower bound of the $\mathsf{GDoF}$ under the no-CSIT assumption.
\begin{corollary}\label{nocsitgdof}
Assuming $\alpha_{\mathrm{dl}}=\alpha_{\mathrm{ul}}=1$ and no-CSIT, the achievable $\mathsf{GDoF}$ region of $(M_{\mathrm{dl}},N_{\mathrm{dl}},M_{\mathrm{ul}},N_{\mathrm{ul}})$ side-channel assisted MIMO full-duplex network satisfies the following constraints
\begin{gather}
\begin{aligned}
  &\mathsf{DoF}_{\mathrm{dl}}\leq m_{\mathrm{dl}},~~\mathsf{DoF}_{\mathrm{ul}}\leq \min\left\{m_{\mathrm{ul}},\alpha_{\rm I}m_{\rm I}+W\alpha_{\rm S}m_{\rm I}\right\},\\
  &\mathsf{DoF}_{\mathrm{dl}}+\mathsf{DoF}_{\mathrm{ul}}\leq f\big(N_{\mathrm{dl}},(\alpha_{\mathrm{I}},M_{\mathrm{ul}}),(1,M_{\mathrm{dl}})\big)+Wf\big(N_{\mathrm{dl}},(\alpha_{\mathrm{S}},M_{\mathrm{ul}})\big).
\end{aligned}
\end{gather}
\end{corollary}
\begin{proof}
The achievable $\mathsf{GDoF}$ region without CSIT can be derived following the same argument as in the case with CSIT.
\end{proof}
\begin{remark}
Comparing the Corollaries~\ref{csitgdof} and \ref{nocsitgdof}, we conclude that when $\alpha_{\rm I}\geq 1$ and $N_{\rm dl}\geq M_{\rm ul}$, acquiring CSIT is of no use as the $\mathsf{GDoF}$ without CSIT achieves the optimal $\mathsf{GDoF}$ with CSIT. 
In the strong interference regime where $\mathsf{INR}>\mathsf{SNR}$, larger number of receiver antennas is sufficient to null out the interference to achieve the optimal $\mathsf{GDoF}$ regime.

\end{remark}

\subsection{Spatial and Spectral Tradeoff in $\mathsf{GDoF}$}
In this section, we will compare three systems: (i) the side-channel assisted full-duplex network with CSIT, (ii) the side-channel assisted full-duplex network without CSIT, and (iii) an idealized full-duplex network without interference, i.e., a parallel uplink and a downlink channel; the last network provides us the benchmark for the best possible performance.  By comparing these three systems, we aim to quantify the relationship between the spatial resources of multiple antennas and spectral resources of the side-channel. We start by presenting several corollaries to Theorems~\ref{sc1} and~\ref{sc2}.


\begin{corollary}(Case A with CSIT)\label{coroIF}
The effect of interference can be completely eliminated if the bandwidth ratio of the side-channel to main-channel satisfies the following condition, 
\begin{gather}
\begin{aligned}
W_{\text{CSIT}}=
 \begin{cases}
 \frac{\alpha_{\rm I}}{\alpha_{\rm S}}\left(2-\frac{\max(M,N)}{\min(M,N)}\right)^+, &\text{for}~ \alpha_{\rm I}<1,\\
\frac{1}{\alpha_{\rm S}}\left(3-\frac{\max(M,N)}{\min(M,N)}-\alpha_{\rm I}\right)^+, &\text{for}~\alpha_{\rm I}\geq 1.\label{gamma} \end{cases}
\end{aligned}
\end{gather}
\end{corollary}
From Corollary~\ref{coroIF}, we can see that the required bandwidth ratio is a linearly decreasing function of the antenna number ratio $\frac{\max(M,N)}{\min(M,N)}$ to achieve the interference-free performance. 
Therefore the spatial resources of the number of antennas at transmitters and receivers is interchangeable with the spectral resources of the side-channel bandwidth to eliminate interference. 
The intuition behind it is that the additional spatial signaling dimension to perform transmit/receive beamforming is equivalent to leveraging the extra spectral signaling dimension of the side-channel for interference cancellation.

From Corollary~\ref{nocsitgdof}, we can also find out the required bandwidth ratio under the no-CSIT assumption in order to achieve the no-interference upper bound. The required bandwidth ratio without CSIT in Case A for $\alpha_{\rm I}=1$ is given by
\begin{gather}
\begin{aligned}
W_{\text{No-CSIT}}=
 \begin{cases}
\frac{1}{\alpha_{\rm S}}\left(2-\frac{N}{M}\right)^+, &\text{for}~ N\geq M,\\
\frac{1}{\alpha_{\rm S}}, &\text{for}~M>N. \end{cases}
\end{aligned}
\end{gather}
\begin{figure}[h!]
\begin{minipage}[b]{0.5\linewidth}
  \centering
     \scalebox{0.4}{\includegraphics[trim = 80mm
      40mm 75mm 60mm, clip]{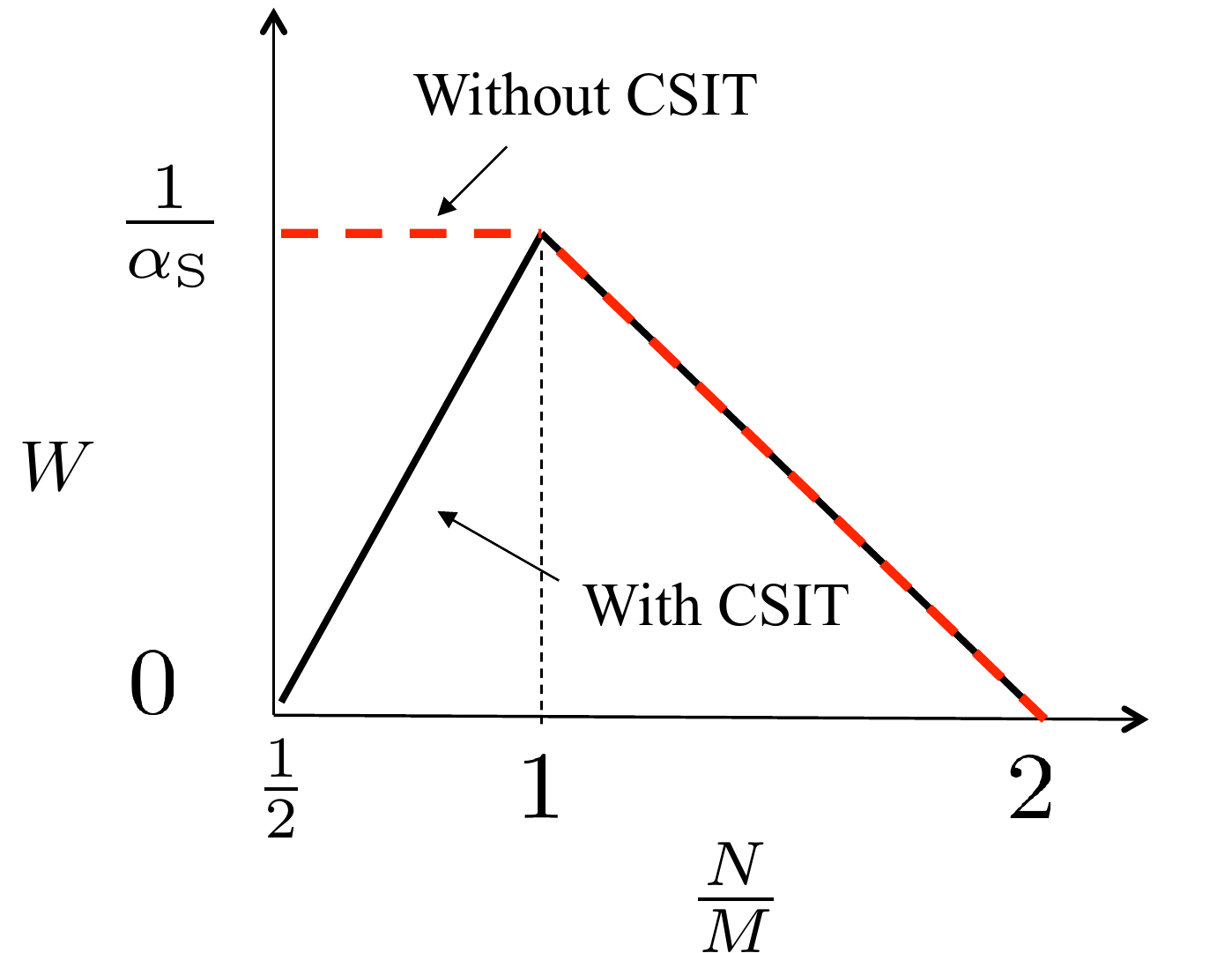}}
  \caption{Spatial spectral tradeoff in Case A when $\alpha_{\rm I}=1$.}
\label{fig1}
\end{minipage}
\begin{minipage}[b]{0.4\linewidth}
  \centering
    \scalebox{0.4}{\includegraphics[trim = 60mm
      35mm 50mm 55mm, clip]{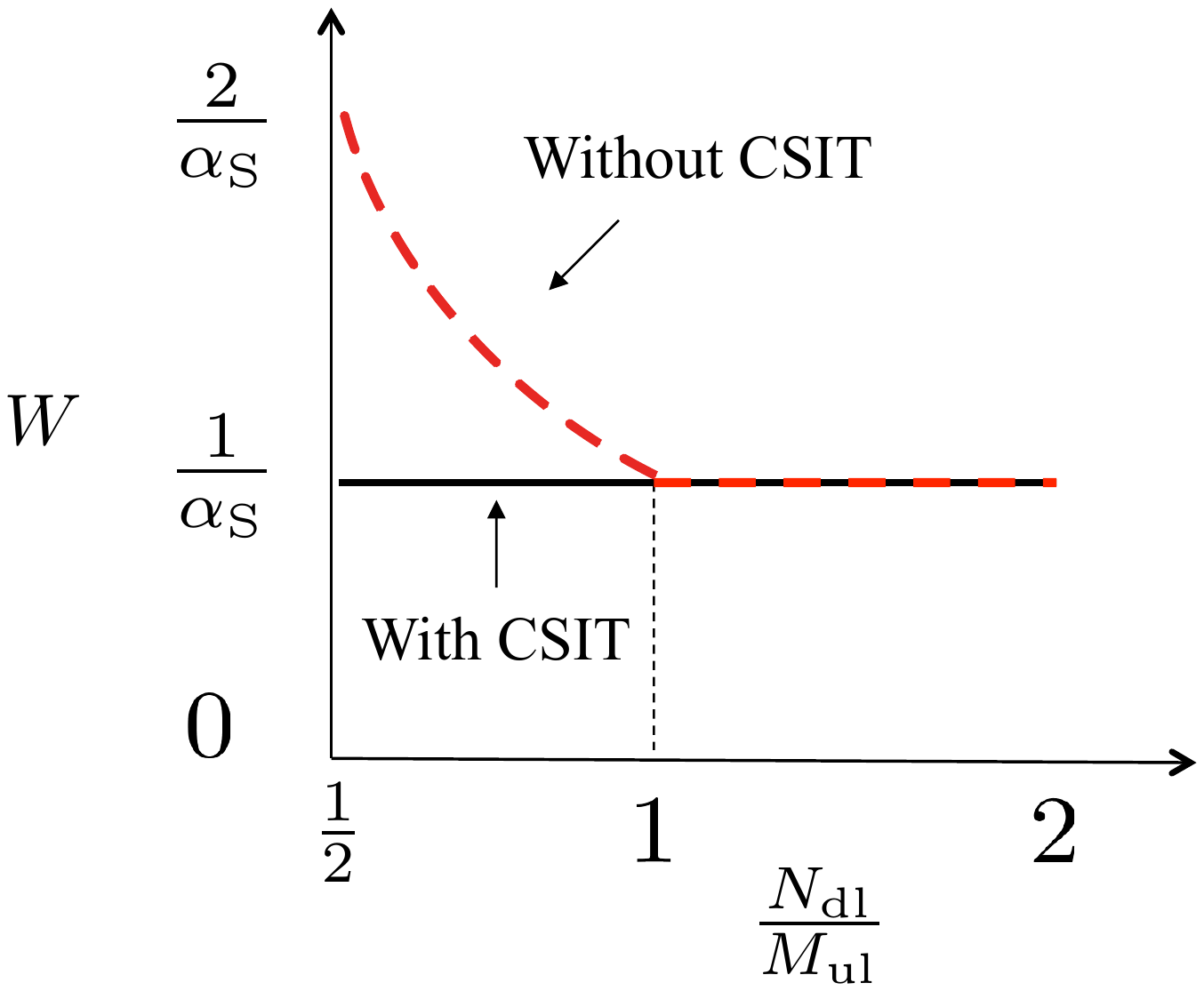}}
  \caption{Spatial spectral tradeoff in Case B when $\alpha_{\rm I}=1$.}
\label{fig2}
\end{minipage}
\end{figure}
\begin{corollary}(Case B with CSIT)\label{coroIF2}
The effect of interference can be completely eliminated if the bandwidth ratio of the side-channel to main-channel satisfies the following condition, 
\begin{gather}
\begin{aligned}
W_{\text{CSIT}}=
 \begin{cases}
 \frac{\alpha_{\rm I}}{\alpha_{\rm S}}, &\text{for}~ \alpha_{\rm I}<1,\\
\frac{(2-\alpha_{\rm I})^+}{\alpha_{\rm S}}, &\text{for}~\alpha_{\rm I}\geq 1.\nonumber \end{cases}
\end{aligned}
\end{gather}
\end{corollary}
We observe that in Case B, the required side-channel bandwidth to achieve the no-interference sum $\mathsf{GDoF}$ is not affected by the number of antennas in the system but received interference signal strength and side-channel signal strength levels.
For $\alpha_{\rm I}<1$, lower interference level
requires less side-channel bandwidth while for $\alpha_{\rm I}\geq1$, higher interference level leads to smaller side-channel bandwidth requirement.

In Case B, we can also derive the required bandwidth ratio under the no-CSIT assumption from Corollary~\ref{nocsitgdof}, to achieve the no-interference performance. The required bandwidth ratio without CSIT in Case B for $\alpha_{\rm I}=1$ is given by
\begin{gather}
\begin{aligned}
W_{\text{No-CSIT}}=
 \begin{cases}
\frac{1}{\alpha_{\rm S}}, &\text{for}~ N_{\rm dl}\geq M_{\rm ul},\\
\frac{M_{\rm ul}}{N_{\rm dl}\alpha_{\rm S}}, &\text{for}~M_{\rm ul}>N_{\rm dl}. \end{cases}
\end{aligned}
\end{gather}

In Figs.~\ref{fig1} and~\ref{fig2}, we show the spatial and spectral tradeoff in both Case A and Case B when $\alpha_{\rm I}=1$. We observe that when there are more downlink receive antennas than uplink transmit antennas, obtaining CSIT is unavailing since with and without CSIT require the same amount of side-channel bandwidth to completely eliminate interference. However, when we have more uplink transmit antennas, if we do not have CSIT, the extra spatial degrees-of-freedom are wasted and we need more side-channel bandwidth to achieve the no-interference performance.
\begin{figure}
\centering
\subfigure[$N_{\mathrm{dl}}\geq M_{\mathrm{ul}}$]{\scalebox{0.62}{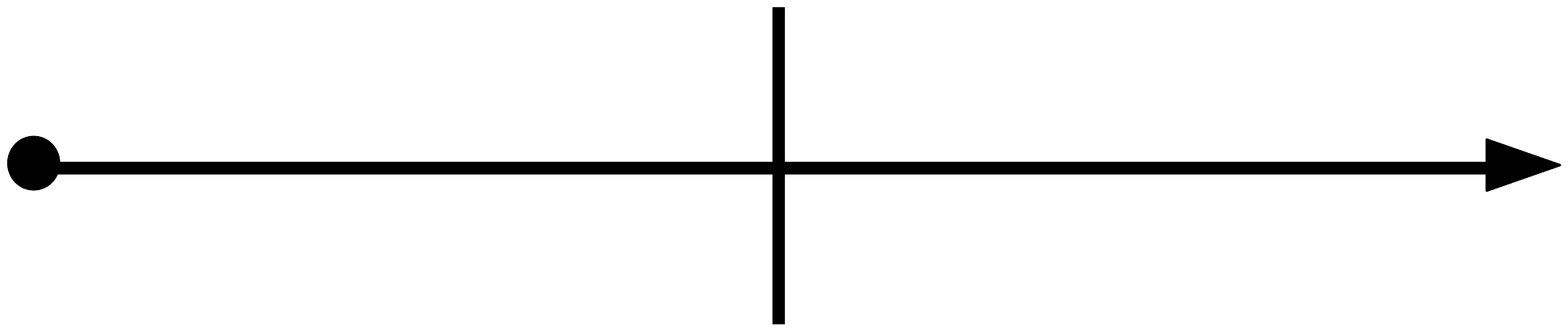}\label{fig:sub1}}
\hspace{0.6cm}
\subfigure[$M_{\mathrm{ul}}> N_{\mathrm{dl}}$]{\scalebox{0.62}{\input{pic/gdofm2gn1.pdftex_t}}\label{fig:sub2}}
\caption{Comparison of the three systems in $\mathsf{DoF}$ as a function of the side-channel bandwidth when $\alpha_{\rm I}=1$.}\label{seccom}
\end{figure}

In Fig.~\ref{seccom}, we give an illustration of the comparisons of the three systems in $\mathsf{DoF}$ as a function of the side-channel bandwidth when there is an excess of BS antennas.

\section{Diversity and multiplexing tradeoff of MIMO distributed full-duplex} \label{optimalDMT}

In this section, we consider a slow-fading scenario. When the channel experiences slow fading,  an important metric to characterize the MIMO system performance is the diversity and multiplexing tradeoff~(DMT), which delineates the asymptotic tradeoff between data rate and reliability in the high SNR limit.
The optimal DMT, first introduced in MIMO point-to-point channels~\cite{DMTTse}, represents the optimal diversity gain $d^*(r)$ for each multiplexing gain $r$ among all possible schemes.  Similar to our definition of $\mathsf{GDoF}$, we define the multiplexing gain of both downlink and uplink channel in our system as follows
\begin{gather}
r_i=\lim_{\rho\rightarrow\infty}\frac{R_i(\rho_i)}{W_m\mathrm{log}\rho},~i\in\{\mathrm{dl,ul}\},
\end{gather}
where $R_{\mathrm{dl}}$ and $R_{\mathrm{ul}}$ are the achievable rates (bit/s) of downlink and uplink, respectively.

Assuming the overall average error probability is $P_e(r_{\mathrm{dl}},r_{\mathrm{ul}})$, the DMT is
\begin{gather}
d(r_{\mathrm{dl}},r_{\mathrm{ul}})=\lim_{\rho\rightarrow\infty}\frac{-\mathrm{log}P_e(r_{\mathrm{dl}},r_{\mathrm{ul}})}{\mathrm{log}\rho}.
\end{gather}
We define $d^{\text{opt}}(r_{\mathrm{dl}},r_{\mathrm{ul}})$ as the supremum of $d(r_{\mathrm{dl}},r_{\mathrm{ul}})$ computed over all possible schemes. Thus $d^{\text{opt}}(r_{\mathrm{dl}},r_{\mathrm{ul}})$ is the optimal DMT of the system. 

In this section, we will study the DMT performance under different assumptions regarding the availability of CSIT. We assume that the channel knowledge is known at the receivers.
In the following, we will first obtain the optimal DMT with CSIT which can be achieved by vector bin-and-cancel as described in Section \ref{bcscheme}.
Next, we study the case without CSIT and derive the corresponding achievable DMT. Finally, based on the DMT result, we will investigate the spatial and spectral tradeoff as well as the interplay between CSIT and side-channel.
\subsection{With CSIT Case}\label{csitcase}
In a slow fading scenario, the channel matrices remain fixed over a fade period with a short-term power constraint given in~(\ref{pcc}), thus the capacity region in time-invariant channels can serve as instantaneous capacity region in each fade period. We define the outage event as the target rate pair not contained in the instantaneous capacity region: $\mathfrak{B}\triangleq\{(R_{\mathrm{dl}},R_{\mathrm{ul}})\notin \mathcal{C(H)} \}$, where $\mathcal{C(H)}$ is given in Corollary~\ref{capacity}. 
From \cite{DMTTse}, it can be easily shown that $P^*_e(r_{\mathrm{dl}},r_{\mathrm{ul}})\doteq\mathrm{Pr}(\mathfrak{B})$, where $P^*_e(r_{\mathrm{dl}},r_{\mathrm{ul}})$ is the infimum of the overall average error probability among all possible schemes. In the high SNR limit, we can obtain that
\[\mathrm{Pr}(\mathfrak{B})\doteq\max_{i\in\{\rm dl,ul,{\mathrm{sum}}\}}\mathrm{Pr}\left(C_{i}< R_i\right),~\implies\rho^{-d^*(r_{\mathrm{dl}},r_{\mathrm{ul}})}\doteq\max_{i\in\{\rm dl,{\mathrm{ul}},\mathrm{sum}\}}\mathrm{Pr}\left(C_{i}< R_i\right),\]
where $C_i$ is given in (\ref{hignSNRcap}) and $R_{\mathrm{sum}}=R_{\mathrm{dl}}+R_{\mathrm{ul}}$.
Thus the optimal diversity order is
\begin{gather}
\begin{aligned}
d^*(r_{\mathrm{dl}},r_{\mathrm{ul}})=\min_{i\in\{{\mathrm{dl}},{\mathrm{ul}},{\mathrm{sum}}\}}d_{\mathfrak{B}_i}(r_i),~~
\label{optimald}
\text{where}~d_{\mathfrak{B}_i}(r_i)=\lim_{\rho\to\infty}-\frac{\mathrm{log}\mathrm{Pr}(C_{i}< W_m r_i\mathrm{log}\rho)}{\mathrm{log}\rho},
\end{aligned}
\end{gather}

In Section~\ref{bcscheme}, we showed that vector bin-and-cancel achieves the asymptotic capacity region. Hence in the asymptotic DMT characterization,  the optimal DMT with CSIT can be achieved by vector bin-and-cancel which only requires CSIT of the interference channel between the up- and downlink nodes since the uplink message splitting depends on the interference channel.
The derivation of the optimal DMT curve of side-channel assisted MIMO full-duplex network follows from two steps. 
In \cite{DMTTse}, we know that the optimal DMT for MIMO point-to-point channel is $d_{M,N}(r)=(M-r)(N-r)$, which is a piecewise linear curve joining the integer point $r\in[0,\min(M,N)]$. For a general channel level $\alpha_{i}, i\in\{\mathrm{dl},\mathrm{ul}\}$ of a point-to-point channel, we will invoke Lemma~\ref{lemmacite}~(in Appendix~\ref{citelemma}) for our calculation. Hence we first obtain the optimal diversity order of each individual downlink and uplink given as
\begin{gather}
\begin{aligned}
d_{\mathfrak{B}_i}(r_i)=\alpha_{i}d_{M_i,N_i}\left(\frac{r_i}{\alpha_{i}}\right), \forall r_i\in[0,\min\{M_i,N_i\}\alpha_{i}], i\in\{{\mathrm{dl}},{\mathrm{ul}}\}. \label{diP2P}
\end{aligned}
\end{gather}
Next, we evaluate $d_{\mathfrak{B}_{\mathrm{sum}}}(r_{\mathrm{sum}})$ in the following lemma.
\begin{lemma}\label{doss}
The diversity order with CSIT given the sum multiplexing gain of both uplink and downlink is the minimum of the following objective function:
\begin{gather}
\begin{aligned}
d_{\mathfrak{B}_{\mathrm{sum}}}(r_{\mathrm{sum}})=&\min_{\bar{\mu},\bar{\sigma},\bar{\theta},\bar{\nu}}\sum_{i=1}^{m_{\mathrm{dl}}}(M_{\mathrm{dl}}+N_{\mathrm{dl}}+1-2i)\mu_i+\sum_{j=1}^{m_{\mathrm{ul}}}(M_{\mathrm{ul}}+N_{\mathrm{ul}}+1-2j)\sigma_j-(M_{\mathrm{dl}}+N_{\mathrm{ul}})m_{\mathrm{I}}\alpha_{\mathrm{I}}\\
&+\sum_{k=1}^{m_{\mathrm{I}}}(M_{\mathrm{dl}}+N_{\mathrm{ul}}+M_{\mathrm{ul}}+N_{\mathrm{dl}}+1-2k)\theta_k+\sum_{l=1}^{m_{\mathrm{I}}}(M_{\mathrm{ul}}+N_{\mathrm{dl}}+1-2l)\nu_l\\
&+\sum_{i=1}^{m_{\mathrm{dl}}}\sum_{k=1}^{\min\{N_{\mathrm{dl}}-i,M_{\mathrm{ul}}\}}(\alpha_{\mathrm{I}}-\mu_i-\theta_k)^++\sum_{j=1}^{m_{\mathrm{ul}}}\sum_{k=1}^{\min\{M_{\mathrm{ul}}-j,N_{\mathrm{dl}}\}}(\alpha_{\mathrm{I}}-\sigma_j-\theta_k)^+;\\
\mathrm{Subject~to}\quad&\sum_{i=1}^{m_{\mathrm{dl}}}(\alpha_{1}-\mu_i)^++\sum_{j=1}^{m_{\mathrm{ul}}}(\alpha_{2}-\sigma_j)^++\sum_{k=1}^{m_{\mathrm{I}}}(\alpha_{\mathrm{I}}-\theta_k)^++W\sum_{l=1}^{m_{\mathrm{I}}}(\alpha_{\mathrm{S}}-\nu_l)^+< r_{\mathrm{sum}};\\
&0\leq\mu_{\mathrm{1}}\leq\cdots\mu_{m_{\mathrm{dl}}};~0\leq\sigma_{\mathrm{1}}\leq\cdots\sigma_{m_{\mathrm{ul}}};~0\leq\theta_{\mathrm{1}}\leq\cdots\theta_{m_{\mathrm{I}}};~0\leq\nu_{\mathrm{1}}\leq\cdots\nu_{m_{\mathrm{I}}};\\
&\mu_i+\theta_k\geq\alpha_{\mathrm{I}},~\forall (i+k)\geq N_{\mathrm{dl}}+1;\\
&\sigma_j+\theta_k\geq\alpha_{\mathrm{I}},~\forall (j+k)\geq M_{\mathrm{ul}}+1, \label{disum}
\end{aligned}
\end{gather}
where $\bar{\mu}=\{\mu_{\mathrm{1}},\cdots,\mu_{m_{\mathrm{dl}}}\},\bar{\sigma}=\{\sigma_{\mathrm{1}},\cdots,\sigma_{m_{\mathrm{ul}}}\},\bar{\theta}=\{\theta_{\mathrm{1}},\cdots,\theta_{m_{\mathrm{I}}}\}, \bar{\nu}=\{\nu_{\mathrm{1}},\cdots,\nu_{m_{\mathrm{I}}}\}$ and $m_{\mathrm{dl}},~m_{\mathrm{ul}}$ and $m_{\mathrm{I}}$ are defined in (\ref{gap}).
\end{lemma}
\begin{proof}
We provide the proof in Appendix~\ref{dos}.
\end{proof}
With $d_{\mathfrak{B}_i}$ for $i\in\{{\mathrm{dl}},{\mathrm{ul}},{\mathrm{sum}}\}$ derived above, we have the following theorem which gives the optimal DMT in its most general form, allowing different channel parameters and multiplexing gains for uplink and downlink with arbitrary number of antennas at each node.

\begin{theorem}\label{the1}
The optimal DMT of $(M_{\mathrm{dl}},N_{\mathrm{dl}},M_{\mathrm{ul}},N_{\mathrm{ul}})$ side-channel assisted MIMO full-duplex network with CSIT denoted as $d^{\text{CSIT,opt}}$ is given by
 \[
d^{\text{CSIT,opt}}_{(M_{\mathrm{dl}},N_{\mathrm{dl}},M_{\mathrm{ul}},N_{\mathrm{ul}})}(r_{\mathrm{dl}},r_{\mathrm{ul}})=\min_{i\in\{{\mathrm{dl}},{\mathrm{ul}},{\mathrm{sum}}\}}d_{\mathfrak{B}_i}(r_i),
\]where $d_{\mathfrak{B}_i}(r_i)$ is given in (\ref{diP2P}) and Lemma~\ref{doss}.
\end{theorem}
The optimization problem in Lemma~\ref{doss} is a convex optimization problem~\cite{convex} with linear constraints, which can be solved using linear programming. The general form of the optimal DMT with CSIT in Theorem~\ref{the1}, though can be calculated using numerical methods, does not result in a closed-form solution. In the following corollary, a closed-form DMT result is derived in the case of single-antenna mobiles communicating with multiple-antenna BS with $M$ transmit and receive antennas, i.e., $M_{\rm dl}=N_{\rm ul}=M$.  
\begin{corollary}\label{coroex1}
In the case of  $(M,1,1,M)$ with symmetric DMT $r_{\rm ul}=r_{\rm dl}=r$ when $\alpha_{\mathrm{dl}}=\alpha_{\mathrm{ul}}=\alpha_{\mathrm{I}}=1$. 
The closed-form optimal DMT with CSIT is given which completely characterizes the optimal DMT under all side-channel conditions:
\begin{itemize}
\item when $W\leq \frac{1}{2M+1}$ and $W\alpha_{\rm S}<1$, 
\begin{gather}
\begin{aligned}
d^{\text{CSIT,opt}}_{(M,1,1,M)}(r)=\left\{
  \begin{array}{l l}
   M(1-r),&0\leq r\leq \frac{M+1+(2M+1)W\alpha_{\rm S}}{3M+2}\\
 (2M+1)(1+W\alpha_{\mathrm{S}})-(4M+2)r,& \frac{M+1+(2M+1)W\alpha_{\rm S}}{3M+2}\leq r \leq \frac{1+W\alpha_{\rm S}}{2}
   \end{array} \right.
\end{aligned}
\end{gather}
\item when $ \frac{1}{2M+1}\leq W<\frac{2}{M}, \alpha_{\rm S}\geq \frac{M}{2}$, and $W\alpha_{\rm S}< 1$,
\begin{gather}
\begin{aligned}
d^{\text{CSIT,opt}}_{(M,1,1,M)}(r)=\left\{
  \begin{array}{l l}
   M(1-r),&0\leq r\leq \beta^*\\
 \alpha_{\mathrm{S}}+\frac{1}{W}(1-2r),&\beta^*\leq r\leq \frac{1+W\alpha_{\rm{S}}}{2}
   \end{array} \right.
\end{aligned}
\end{gather}
\item when $W\geq \frac{1}{2M+1}, \alpha_{\rm S}< \frac{M}{2}$, and $ W\alpha_{\rm S}< 1$,
\begin{gather}
\begin{aligned}
d^{\text{CSIT,opt}}_{(M,1,1,M)}(r)=\left\{
  \begin{array}{l l}
   M(1-r),&0\leq r\leq \frac{M+1+\alpha_{\rm S}}{3M+2}\\
    2M+1+\alpha_{\mathrm{S}}-(4M+2)r,& \frac{M+1+\alpha_{\rm S}}{3M+2}\leq r\leq \frac{1}{2}\\
    \alpha_{\mathrm{S}}+\frac{1}{W}(1-2r),&\frac{1}{2}\leq r\leq \frac{1+W\alpha_{\rm{S}}}{2}
   \end{array} \right.
\end{aligned}
\end{gather}
\item  when $W\geq \frac{1}{2M+1}, \alpha_{\rm S}< \frac{M}{2}$, and $ W\alpha_{\rm S}\geq 1$,
\begin{gather}
\begin{aligned}
d^{\text{CSIT,opt}}_{(M,1,1,M)}(r)=\left\{
  \begin{array}{l l}
   M(1-r),&0\leq r\leq \frac{M+1+\alpha_{\rm S}}{3M+2}\\
    2M+1+\alpha_{\mathrm{S}}-(4M+2)r,& \frac{M+1+\alpha_{\rm S}}{3M+2}\leq r\leq \frac{1}{2}\\
    \alpha_{\mathrm{S}}+\frac{1}{W}(1-2r),&\frac{1}{2}\leq r\leq \beta^*\\
    M(1-r), & \beta^*\leq r \leq 1
       \end{array} \right.
\end{aligned}
\end{gather}
\item  when $\alpha_{\rm S}\geq \frac{M}{2}$ and $W\alpha_{\rm S}\geq 1$,
\begin{gather}
\begin{aligned}
d^{\text{CSIT,opt}}_{(M,1,1,M)}(r)=
    M(1-r), 0\leq r \leq 1
\end{aligned}
\end{gather}
where $\beta^*=\frac{\alpha_{\rm S}+\frac{1}{W}-M}{\frac{2}{W}-M}.$
\end{itemize}
\end{corollary}
\begin{proof}
The DMT of the point-to-point channel is $M(1-r), \forall r\in[0,1]$. Thus we only need to solve for the optimization problem given sum multiplexing gain.
One way to find the minimum of the optimization problem in Lemma~\ref{doss} is to apply the Karush-Kuhn-Tucker condition. Here we will provide another approach which is the key to the proof of a general case.
The method we adopt is gradient descent which finds the local optimum. Since the optimization problem we have is convex with linear constraints, the local optimum is actually the global optimum in convex optimization~\cite{convex}. Hence we can obtain the global optimum via gradient descent algorithm.

We first simplify the objective function of the diversity order in Lemma~\ref{doss} given sum multiplexing gain. By substituting $\nu_l^\prime=W\nu_l$ in (\ref{disum}), we can express the objective function as
\begin{gather}
\begin{aligned}
d^{\text{CSIT}}_{\mathrm{sum}}=& \min M\mu_1+M\sigma_1+(2M+1)\theta_1+\frac{\nu_1^\prime}{W}-2M,\\ \label{m1dmt1}
\mathrm{Subject~to}\quad& (1-\mu_1)^++(1-\sigma_1)^++(1-\theta_1)^++(W\alpha_{\mathrm{S}}-\nu_1^\prime)^+< r_{\mathrm{sum}};\\
&\mu_1,\sigma_1,\theta_1,\nu_1^\prime\geq 0;\\
&\mu_1+\theta_1\geq1;~\sigma_1+\theta_1\geq1.
\end{aligned}
\end{gather}
Next, we differentiate the objective function  in (\ref{m1dmt1}) with respect to different variables 
\begin{eqnarray}
\frac{ \partial d^{\text{CSIT}}_{\mathrm{sum}}}{\partial \nu_1^\prime }&=&\frac{1}{W};\\
\frac{ \partial d^{\text{CSIT}}_{\mathrm{sum}}}{\partial \theta_1}&=&2M+1;\\
\frac{ \partial d^{\text{CSIT}}_{\mathrm{sum}}}{\partial \mu_1}&=&
\frac{ \partial d^{\text{CSIT}}_{\mathrm{sum}}}{\partial \sigma_1}=M<\frac{ \partial d^{\text{CSIT}}_{\mathrm{sum}}}{\partial \theta_1}.
\end{eqnarray}
Comparing the gradient of each variable, when $W\leq\frac{1}{2M+1}$, the steepest descent of the objective function is along the decreasing value of $\nu_1^\prime$ with $\theta_1=\mu_1=\sigma_1=1$, for $r_{\rm sum}\leq W\alpha_s$. Thus we have
$d^{\text{CSIT}}_{{\mathrm{sum}}}(r)=2M+1+\alpha_{\mathrm{S}}-\frac{r_{\rm sum}}{W},\forall r_{\rm sum}\in[0,W\alpha_{\mathrm{S}}].$ This also implies that for $r_{\rm sum}\geq W\alpha_s$, $\nu_1^\prime=0$ in the optimal solution. Now the steepest descent of the objective function in (\ref{m1dmt1}) is along the decreasing value of $\theta_1$ with $\mu_1=\sigma_1=1$, and the corresponding minimum is 
$d^{\text{CSIT}}_{{\mathrm{sum}}}(r_{\rm sum})=(2M+1)(1+W\alpha_{\mathrm{S}})-(2M+1)r_{\rm sum},$ $\forall r_{\rm sum}\in[W\alpha_{\mathrm{S}},1+W\alpha_{\rm S}].$

When $W\geq \frac{1}{2M+1}$, the steepest descent of the objective function is along the decreasing value of $\theta_1$ with $\mu_1=\sigma_1=1,\nu_1^\prime=W\alpha_{\rm S}$, for $r_{\rm sum}\leq 1$. Thus we have
$d^{\text{CSIT}}_{{\mathrm{sum}}}(r_{\rm sum})=2M+1+\alpha_{\mathrm{S}}-(2M+1)r_{\rm sum},\forall r_{\rm sum}\in[0,1].$
Again, for $r_{\rm sum}\geq 1$,  the optimal solution has $\theta_1=0$. We will rewrite the objective function as
\begin{gather}
\begin{aligned}
 d^{\text{CSIT}}_{\mathrm{sum}}=&\min M\mu_1+M\sigma_1+\frac{\nu_1^\prime}{W}-2M,\\ 
\mathrm{Subject~to}\quad& (1-\mu_1)^++(1-\sigma_1)^++(W\alpha_{\mathrm{S}}-\nu_1^\prime)^+\leq r_{\mathrm{sum}}-1;\\
&\mu_1,\sigma_1,\nu_1^\prime\geq 0;\\
&\mu_1\geq1;~\sigma_1\geq1.
\end{aligned}
\end{gather}
To minimize the objective function above, we should let $\mu_1=\sigma_1=1.$  Hence the minimum of the objective function is $d^{\text{CSIT}}_{{\mathrm{sum}}}(r_{\rm sum})=\alpha_{\mathrm{S}}+\frac{1}{W}(1-r_{\rm sum}),\forall r_{\rm sum}\in[1,1+W\alpha_{\rm S}].$
Now combining all the results above, we have 
\begin{equation}
d^{\text{CSIT,opt}}_{(M,1,1,M)}(r)=\min\{M(1-r), d^{\text{CSIT}}_{{\mathrm{sum}}(M,1,1,M)}(r)\}~\text{for}~0\leq r\leq 1. \label{csitex}
\end{equation}
where $d^{\text{CSIT}}_{{\mathrm{sum}}(M,1,1,M)}(r)$ is given as
\begin{itemize}
\item when $W\leq \frac{1}{2M+1}$
\begin{gather}
\begin{aligned}
d^{\text{CSIT}}_{{\mathrm{sum}}(M,1,1,M)}(r)=\left\{
  \begin{array}{l l}
  2M+1+\alpha_{\mathrm{S}}-\frac{2r}{W},&0\leq r\leq \frac{W\alpha_{\mathrm{S}}}{2}\\
  (2M+1)(1+W\alpha_{\mathrm{S}})-(4M+2)r,&\frac{W\alpha_{\mathrm{S}}}{2}\leq r\leq \frac{1+W\alpha_{\rm{S}}}{2}
   \end{array} \right.
\end{aligned}
\end{gather}
\item when $W\geq \frac{1}{2M+1}$
\begin{gather}
\begin{aligned}
d^{\text{CSIT}}_{{\mathrm{sum}}(M,1,1,M)}(r)=\left\{
  \begin{array}{l l}
  2M+1+\alpha_{\mathrm{S}}-(4M+2)r,&0\leq r\leq \frac{1}{2}\\
  \alpha_{\mathrm{S}}+\frac{1}{W}(1-2r),&\frac{1}{2}\leq r\leq \frac{1+W\alpha_{\rm{S}}}{2}
   \end{array} \right.
\end{aligned}
\end{gather}
\end{itemize}
Further simplification of (\ref{csitex}) will lead to the analytical expression in Corollary~\ref{coroex1}.
\end{proof}
\begin{remark} \label{remark1}
The optimal DMT with CSIT in the no side-channel case is a special case of Corollary~\ref{coroex1} when $W=0$ , and is given as 
\begin{gather}
\begin{aligned}
d^{\text{No-SC,CSIT,opt}}_{(M,1,1,M)}(r)=\left\{
  \begin{array}{l l}
   M(1-r),&0\leq r\leq \frac{M+1}{3M+2}\\
  (2M+1)(1-2r),& \frac{M+1}{3M+2}\leq r \leq \frac{1}{2}
   \end{array} \right.
\end{aligned}
\end{gather}
\end{remark}

From Corollary~\ref{coroex1}, we can completely quantify the improvement of DMT with side-channel under all side-channel conditions. Fig.~\ref{plotcase1} depicts the comparison of DMT with/without~(w/wo) side-channel when $W=\frac{1}{2M+1},\alpha_{\mathrm{S}}=\frac{M}{2}$. We define the light loading threshold as the multiplexing gain threshold within which the system error event is dominated by single-user performance. In the case with CSIT, the light loading threshold of the system without side-channel is $B$ shown in Fig.~\ref{plotcase1}. When $r> B$, the dominant error event is that all users are in error. With the help of side-channel, the light loading threshold is increased by $\Delta_1$, where $\Delta_1=\frac{(2M+1)W\alpha_{\rm S}}{3M+2}$.  Moreover, we can see that the side-channel also improves system maximum multiplexing gain~(when the diversity order is zero) by $\Delta_3$, where $\Delta_3=\frac{W\alpha_{\rm S}}{2}.$ Both improvement amount $\Delta_1$ and $\Delta_3$ will scale with side-channel quality~$W\alpha_{\rm S}$~( for $W\leq \frac{1}{2M+1}$) till 
either point $C$ or $D$ reaches the symmetric maximum multiplexing gain of one which corresponds to the no-interference point. 
\begin{figure}[h!]
  \centering
    \includegraphics[width=0.5\textwidth,trim = 30mm
      25mm 32mm 20mm, clip]{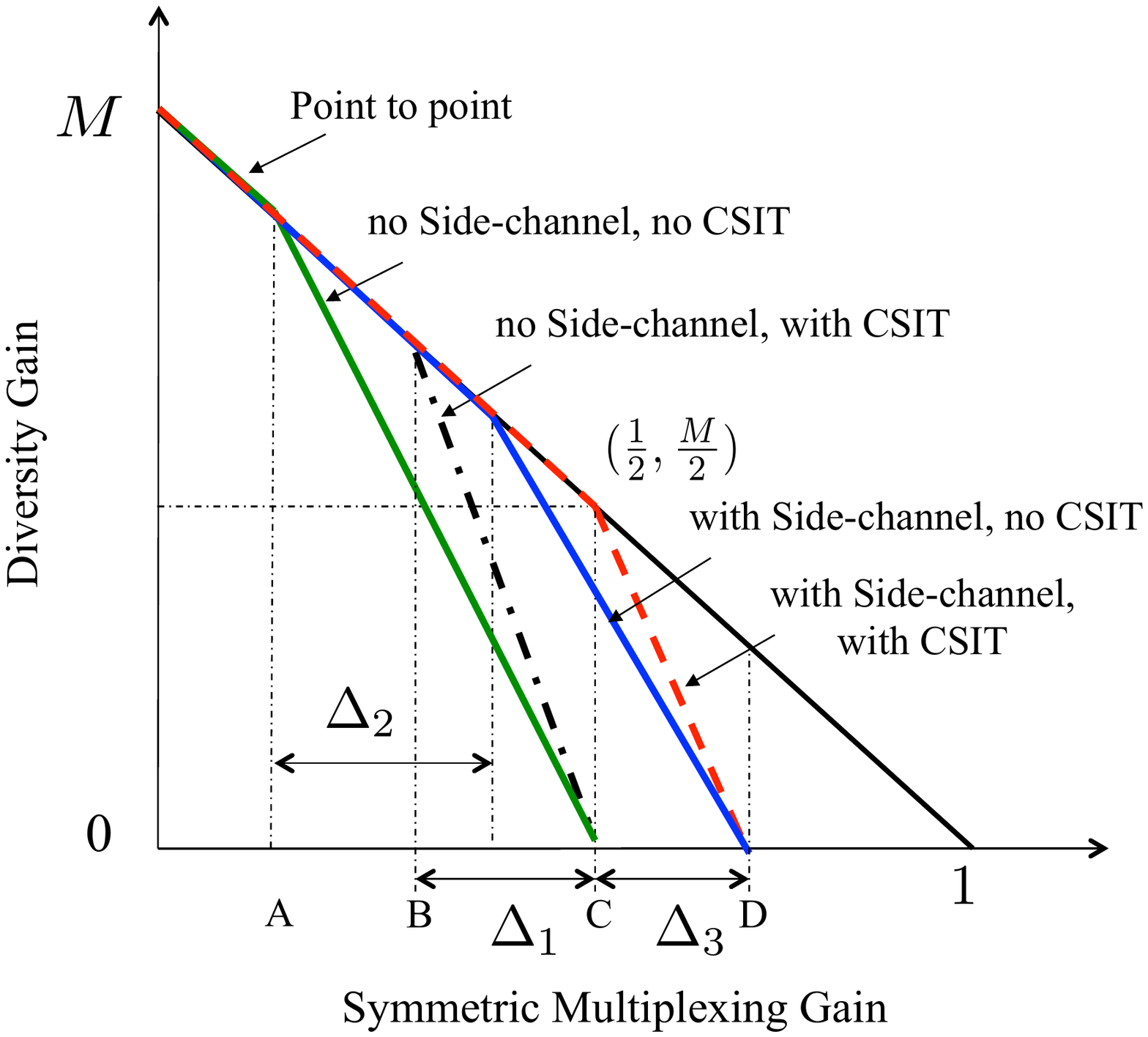}
  \caption{DMT comparison w/wo side-channel w/wo CSIT when $W=\frac{1}{2M+1},\alpha_{\mathrm{S}}=\frac{M}{2}$.}
\label{plotcase1}
\end{figure}

When $W=\frac{1}{2M+1},\alpha_{\mathrm{S}}=\frac{M}{2}$, we have $\Delta_1=\frac{M}{6M+4}$ and $\Delta_3=\frac{M}{8M+4}$. We conclude that in this case, both improvement amount $\Delta_1$ and $\Delta_3$ will scale with the number of antennas at the BS. In the limit of $M$~(as in massive MIMO, BS has unlimited number of antennas), we will have improvement of $\lim_{M\rightarrow\infty}\Delta_1=\frac{1}{6}$ and $\lim_{M\rightarrow\infty}\Delta_3=\frac{1}{8}$.

\subsection{Without CSIT Case}
%
We define the outage event $\mathbf{O}$ in the case without CSIT as the target rate pair does not lie in the achievable  rate region $\mathcal{R}^{\text{No-CSIT}}$: $\mathbf{O}\triangleq\{(R_{\mathrm{dl}},R_{\mathrm{ul}})\notin \mathcal{R} \}$, where $\mathcal{R}$ is given (with $\lambda=\bar{\lambda}=0.5$)
\begin{gather}
\begin{aligned}
\mathcal{R}=&\Bigg\{(R_{\mathrm{dl}},R_{\mathrm{ul}}):R_{\mathrm{dl}}\leq W_m\mathrm{log}\left|I_{N_{\mathrm{dl}}}+\frac{\rho_{\mathrm{dl}}}{M_{\mathrm{dl}}}H_{\mathrm{dl}}H_{\mathrm{dl}}^\dagger\right|;~
R_{\mathrm{ul}}\leq W_m\mathrm{log}\left|I_{N_{\mathrm{ul}}}+\frac{\bar{\lambda}\rho_{{\mathrm{ul}}}}{M_{\mathrm{ul}}}H_{\mathrm{ul}}H_{\mathrm{ul}}^\dagger\right|;\\
R_{\mathrm{dl}}+R_{\mathrm{ul}}&\leq W_m\bigg(\mathrm{log}\left|I_{N_{\mathrm{dl}}}+\frac{\rho_{\mathrm{dl}}}{M_{\mathrm{dl}}}H_{\mathrm{dl}}H_{\mathrm{dl}}^\dagger+\frac{\bar{\lambda}\rho_{\mathrm{I}}}{M_{\mathrm{ul}}}H_{\mathrm{I}}H_{\mathrm{I}}^\dagger\right|+W\mathrm{log}\left|I_{N_{\mathrm{dl}}}+\frac{\lambda\rho_{\mathrm{S}}}{W M_{\mathrm{ul}}}H_{\mathrm{S}}H_{\mathrm{S}}^{\dagger}\right|\bigg)\Bigg\}, \label{nocsitdmtar}
\end{aligned}
\end{gather}
The difference between (\ref{nocsitdmtar}) and the achievable rate region in (\ref{nocsitar}) is that (\ref{nocsitdmtar}) does not have a constraint on $R_{\rm ul}$ for the transmission from up- to downlink mobile. This is because the downlink mobile is not interested in the uplink's message, thereby the failure of decoding uplink's message alone will not be declared as an error event at the downlink receiver.  

Under the no-CSIT assumption, the diversity order of each MIMO downlink/uplink channel is still the same as given in~(\ref{diP2P}).
As for the diversity order for a given sum multiplexing gain, it is characterized by the following lemma.
\begin{lemma}\label{doss2}
The diversity order at a given sum multiplexing gain in the case without CSIT is the minimum of the following objective function:
\begin{gather}
\begin{aligned}
d_{o_{\mathrm{sum}}}(r_{\mathrm{sum}})=&\min_{\bar{\mu},\bar{\theta},\bar{\nu}} \sum_{i=1}^{m_{\mathrm{dl}}}(M_{\mathrm{dl}}+N_{\mathrm{dl}}+1-2i)\mu_i+\sum_{k=1}^{m_{\mathrm{I}}}(M_{\mathrm{ul}}+N_{\mathrm{dl}}+M_{\mathrm{dl}}+1-2k)\theta_k\\
&+\sum_{l=1}^{m_{\mathrm{I}}}(M_{\mathrm{ul}}+N_{\mathrm{dl}}+1-2l)\nu_l-M_{\mathrm{dl}}m_{\mathrm{I}}\alpha_{\mathrm{I}}+\sum_{i=1}^{m_{\mathrm{dl}}}\sum_{k=1}^{\min\{N_{\mathrm{dl}}-i,M_{\mathrm{ul}}\}}(\alpha_{\mathrm{I}}-\mu_i-\theta_k)^+\\
\mathrm{Subject~to}\quad&\sum_{i=1}^{m_{\mathrm{dl}}}(\alpha_{\mathrm{dl}}-\mu_i)^++\sum_{k=1}^{m_{\mathrm{I}}}(\alpha_{\mathrm{I}}-\theta_k)^++W\sum_{l=1}^{m_{\mathrm{I}}}(\alpha_{\mathrm{S}}-\nu_l)^+< r_{\mathrm{sum}};\\
&0\leq\mu_{\mathrm{1}}\leq\cdots\leq\mu_{m_{\mathrm{dl}}};~0\leq\theta_{\mathrm{1}}\leq\cdots\leq\theta_{m_{\mathrm{I}}};~0\leq\nu_{\mathrm{1}}\leq\cdots\leq\nu_{m_{\mathrm{I}}};\\
&\mu_i+\theta_k\geq\alpha_{\mathrm{I}},~\forall (i+k)\geq N_{\mathrm{dl}}+1; \label{disumNo}
\end{aligned}
\end{gather}
\end{lemma}
\begin{proof}
We provide the proof in Appendix~\ref{dos2}.
\end{proof}

\begin{theorem}
A lower bound of the DMT of $(M_{\mathrm{dl}},N_{\mathrm{dl}},M_{\mathrm{ul}},N_{\mathrm{ul}})$ side-channel assisted MIMO full-duplex network without CSIT is given as
 \[
d_{(M_{\mathrm{dl}},N_{\mathrm{dl}},M_{\mathrm{ul}},N_{\mathrm{ul}})}^{\text{No-CSIT}}(r_{\mathrm{dl}},r_{\mathrm{ul}})=\min_{i\in\{{\mathrm{dl}},{\mathrm{ul}},{\mathrm{sum}}\}}d_{o_i}(r_i).
\]
where $d_{o_i}(r_i)$ is given in (\ref{diP2P}) and Lemma~\ref{doss2}.
\end{theorem}
In line with the analysis in Section~\ref{csitcase},  we also give the closed-form no-CSIT DMT in the case of single-antenna mobiles communicating with multiple-antenna BS.
\begin{corollary}\label{coroex2}
In the case of  $(M,1,1,M)$ with symmetric DMT $r_{\rm ul}=r_{\rm dl}=r$ when $\alpha_{\mathrm{dl}}=\alpha_{\mathrm{ul}}=\alpha_{\mathrm{I}}=1$. 
The closed-form lower bound of the DMT without CSIT is given that completely characterizes the achievable DMT under all side-channel conditions:
\begin{itemize}
\item when $W\leq \frac{1}{M+1}$ and $W\alpha_{\rm S}<1$,
\begin{gather}
\begin{aligned}
d^{\text{No-CSIT}}_{(M,1,1,M)}(r)=\left\{
  \begin{array}{l l}
   M(1-r),&0\leq r\leq \frac{1+(M+1)W\alpha_{\rm S}}{M+2}\\
   (M+1)(1+W\alpha_{\mathrm{S}})-(2M+2)r,& \frac{1+(M+1)W\alpha_{\rm S}}{M+2}\leq r \leq \frac{1+W\alpha_{\rm S}}{2}
   \end{array} \right.
\end{aligned}
\end{gather}
\item when $\frac{1}{M+1}\leq W<\frac{2}{M}, \alpha_{\rm S}\geq\frac{M}{2}$, and $W\alpha_{\rm S}< 1$,
\begin{gather}
\begin{aligned}
d^{\text{No-CSIT}}_{{\mathrm{sum}}(M,1,1,M)}(r)=\left\{
  \begin{array}{l l}
   M(1-r),&0\leq r\leq \beta^*\\
 \alpha_{\mathrm{S}}+\frac{1}{W}(1-2r),&\beta^*\leq r\leq \frac{1+W\alpha_{\rm{S}}}{2}
   \end{array} \right.
\end{aligned}
\end{gather}
\item when $W\geq \frac{1}{M+1}, \alpha_{\rm S}< \frac{M}{2}$, and $W\alpha_{\rm S}< 1$,
\begin{gather}
\begin{aligned}
d^{\text{No-CSIT}}_{{\mathrm{sum}}(M,1,1,M)}(r)=\left\{
  \begin{array}{l l}
   M(1-r),&0\leq r\leq \frac{1+\alpha_{\rm S}}{M+2}\\
    M+1+\alpha_{\mathrm{S}}-(2M+2)r,& \frac{1+\alpha_{\rm S}}{M+2}\leq r\leq \frac{1}{2}\\
    \alpha_{\mathrm{S}}+\frac{1}{W}(1-2r),&\frac{1}{2}\leq r\leq \frac{1+W\alpha_{\rm{S}}}{2}
   \end{array} \right.
\end{aligned}
\end{gather}

\item when $W\geq\frac{1}{M+1}, \alpha_{\rm S}< \frac{M}{2}$, and $W\alpha_{\rm S}\geq 1$,
\begin{gather}
\begin{aligned}
d^{\text{No-CSIT}}_{{\mathrm{sum}}(M,1,1,M)}(r)=\left\{
  \begin{array}{l l}
   M(1-r),&0\leq r\leq \frac{1+\alpha_{\rm S}}{M+2}\\
    M+1+\alpha_{\mathrm{S}}-(2M+2)r,& \frac{1+\alpha_{\rm S}}{M+2}\leq r\leq \frac{1}{2}\\
    \alpha_{\mathrm{S}}+\frac{1}{W}(1-2r),&\frac{1}{2}\leq r\leq \beta^*\\
    M(1-r), & \beta^*\leq r \leq 1
       \end{array} \right.
\end{aligned}
\end{gather}

\item  when $\alpha_{\rm S}\geq \frac{M}{2}$ and $W\alpha_{\rm S}\geq 1$,
\begin{gather}
\begin{aligned}
d^{\text{No-CSIT}}_{{\mathrm{sum}}(M,1,1,M)}(r)=
    M(1-r), 0\leq r \leq 1
\end{aligned}
\end{gather}
\end{itemize}
where $\beta^*=\frac{\alpha_{\rm S}+\frac{1}{W}-M}{\frac{2}{W}-M}.$
\end{corollary}
\begin{proof}
The proof is similar to that in Corollary~\ref{coroex1} which uses gradient descent method. 
\end{proof}
\begin{remark}\label{remark2}
The lower bound of the DMT without CSIT in the no side-channel case is a special case of Corollary~\ref{coroex2} when $W=0$, and is given by
\begin{gather}
\begin{aligned}
d^{\text{No-SC,No-CSIT}}_{(M,1,1,M)}(r)=\left\{
  \begin{array}{l l}
   M(1-r),&0\leq r\leq \frac{1}{M+2}\\
  (M+1)(1-2r),& \frac{1}{M+2}\leq r \leq \frac{1}{2}
   \end{array} \right.
\end{aligned}
\end{gather}
\end{remark}
Remark~\ref{remark1} and Remark~\ref{remark2} describe the DMT without side-channel under CSIT and no-CSIT assumptions. One can easily verify that the no-side-channel cases in~\cite{sezgin2009diversity} and \cite{DMTMIMOZ} w/wo CSIT are special cases incorporated in our derivation of DMT.

Now we compare the lower bound of the DMT w/wo side-channel under the no-CSIT assumption. When $W\leq \frac{1}{M+1}$, in the case without CSIT,  with the help of side-channel, the light loading threshold over the no-side-channel system is increased by $\Delta_2$, where $\Delta_2=\frac{(M+1)W\alpha_{\rm S}}{M+2}$.  In Fig.~\ref{plotcase1}, the DMT without CSIT w/wo side-channel is given when $W=\frac{1}{2M+1}$ and $\alpha_{\mathrm{S}}=\frac{M}{2}$. Compared with the light loading improvement under the CSIT assumption, we can see that the side-channel is more effective in increasing the DMT performance in the lack of CSIT as $\Delta_2\geq \Delta_1$.

\subsection{Spatial and Spectral Tradeoff in DMT}

In this section, we will derive symmetric DMT in closed form for a more general case where the mobiles have multiple antennas communicating with the BS with $M$ transmit and receive antennas. Using the closed-form DMT expressions, again we will compare the three systems: with and without CSIT and the no-interference idealized full-duplex network.
We will characterize the relationship between the spatial degrees of freedom of the antenna resources and the extra spectral degrees of freedom due to the side-channels under slow-fading channels. 

We still assume BS has more antennas i.e., $M\geq M_{\mathrm{ul}},~ N_{\mathrm{dl}}$. The closed-form symmetric DMT of the general $(M,N_{\mathrm{dl}},M_{\mathrm{ul}},M)$ system with 
$\alpha_{\mathrm{dl}}=\alpha_{\mathrm{ul}}=\alpha_{\mathrm{I}}=1$ and $r_{\mathrm{dl}}=r_{\mathrm{ul}}=r$ are given under CSIT and no-CSIT assumptions in Lemma~\ref{case3} and Lemma~\ref{case33}~(in Appendix~\ref{caldmt}), respectively. 

First we ask the question that how much side-channel bandwidth is required to compensate for the lack of CSIT such that the DMT of the system without CSIT achieves that of the system with CSIT. The sufficient condition is given in the following theorem.
\begin{theorem} \label{coro1}
In case of $(M,N_{\mathrm{dl}},M_{\mathrm{ul}},M)$,  sufficient conditions such that no CSIT DMT is same as full CSIT DMT are given by 
\vspace{2mm}
\begin{enumerate}
  \item
$W =\min\left\{\frac{N_{\mathrm{dl}}+M_{\mathrm{ul}}-1}{M+N_{\mathrm{dl}}-M_{\mathrm{ul}}+1}, \frac{1}{\alpha_{\mathrm{S}}}\left(2-\frac{N_{\mathrm{dl}}}{M_{\mathrm{ul}}}\right)^+\right\} ~\text{where}~\alpha_{\mathrm{S}}\!\geq\!\frac{d_{M,M_{\mathrm{ul}}}\left(\frac{M_{\mathrm{ul}}}{2}\right)\!-\!M(N_{\mathrm{dl}}\!-\!M_{\mathrm{ul}})}{M_{\mathrm{ul}}N_{\mathrm{dl}}},$  when $N_{\mathrm{dl}}\geq M_{\mathrm{ul}},~M_{\mathrm{ul}}=1,2$;
\vspace{1mm}
  \item $W=0,$ when $N_{\mathrm{dl}}\geq \frac{d_{M_{\mathrm{ul}},M}\left(\frac{M_{\mathrm{ul}}}{2}\right)}{M}+M_{\mathrm{ul}},~M_{\mathrm{ul}}=1,2$.\label{w0}
\end{enumerate}
\end{theorem}
\begin{proof}
With the conditions given above, we can verify that the symmetric DMT with CSIT in Lemma~\ref{case3} is the same as the DMT without CSIT in Lemma~\ref{case33}.
\end{proof}
\begin{corollary}\label{coro5}
When $M_{\mathrm{ul}}> N_{\mathrm{dl}}$, if $W<\frac{1}{\alpha_{\mathrm{S}}}$, the DMT without CSIT is strictly smaller than that with CSIT.
\end{corollary}

Corollary~\ref{coro5} can be readily obtained by comparing Lemma~\ref{case3} and Lemma~\ref{case33}. If $M_{\mathrm{ul}}> N_{\mathrm{dl}}$ and $W<\frac{1}{\alpha_{\mathrm{S}}}$, the availability of CSIT is crucial in performing transmit beamforming to yield higher DMT. 

The next question we will ask is how much side-channel bandwidth is required to eliminate the effect of interference such that the DMT of the system w/wo CSIT achieves that of a system without interference. The following theorem characterizes the effect of the side-channel bandwidth on the performance of the symmetric DMT to reach no-interference DMT. 
\begin{theorem}\label{coro3}
In case of $(M,N_{\mathrm{dl}},M_{\mathrm{ul}},M)$, the sufficient conditions are given under CSIT and no-CIST assumptions, respectively, where the effect of interference can be completely eliminated to achieve the optimal no-interference DMT:
\vspace{1mm}
\begin{enumerate}
\item
$W_{\text{CSIT}}=\frac{1}{\alpha_{\mathrm{S}}}\left(2-\frac{m_X}{m_{\rm I}}\right)^+,~\alpha_{\mathrm{S}}\!\geq\!\frac{(2m_{\rm I}-m_X)(M-m_{\rm I}+1)}{m_{\rm I}(2|N_{\mathrm{dl}}-M_{\mathrm{ul}}|+2)}; $\vspace{3mm}
\item
$
W_{\text{No-CSIT}}=\left\{
  \begin{array}{l l}
  \frac{1}{\alpha_{\mathrm{S}}},~\alpha_{\mathrm{S}}\geq\frac{M-N_{\mathrm{dl}}+1}{2(M_{\mathrm{ul}}-N_{\mathrm{dl}}+1)}, ~\text{when}~M_{\mathrm{ul}}\geq N_{\mathrm{dl}}\vspace{1mm}\\
  \frac{1}{\alpha_{\mathrm{S}}}\left(2-\frac{m_X}{m_{\rm I}}\right)^+,~\alpha_{\mathrm{S}}\!\geq\!\frac{(2m_{\rm I}-m_X)(M-m_{\rm I}+1)}{m_{\rm I}(2|N_{\mathrm{dl}}-M_{\mathrm{ul}}|+2)},~\text{when}~  N_{\mathrm{dl}}\geq M_{\mathrm{ul}}
   \end{array} \right.$\vspace{5mm}
\end{enumerate}
where $m_X=\max(M_{\rm ul},N_{\rm dl}), m_{\rm I}=\min(M_{\rm ul},N_{\rm dl}).$
\end{theorem}
\begin{proof}
We need to show that with the conditions above, the DMT of our system w/wo CSIT is not dominated by the diversity order given sum multiplexing gain $d^{\text{w/wo~CSIT}}_{\mathrm{sum}(M,N_{\rm dl},M_{\rm ul},M)}(r_{\rm sum}),~\forall r\in[0,m_{\rm I}].$
It is sufficient if we show that the conditions above indicate that the decay slope of $d^{\text{w/wo~CSIT}}_{\mathrm{sum}(M,N_{\rm dl},M_{\rm ul},M)}(r_{\rm sum})$ is larger than that of the PTP channel~$d_{M,m_{\rm I}}(r)$ $\forall r$, and the maximum symmetric multiplexing gain of $d^{\text{w/wo~CSIT}}_{\mathrm{sum}(M,N_{\rm dl},M_{\rm ul},M)}(r_{\rm sum})$ is larger than $m_{\rm I}$.

The decay slope of the piecewise linear function $d^k_{M,N}(r)$ is $(M+N-2k+1)$ in each interval $r\in[k-1,k]$, where $k\in[1,\min(M,N)]$ is an integer. Thus the decay slope of $d^k_{M,N}(r)$ decreases as the interval $k$ increases. Also, the decay slope difference between $d^{k-1}_{M,N}(r)$ and $d^k_{M,N}(r)$ is a constant of 2. 
We know that the DMT performance will be improved as side-channel bandwidth ratio $W$ increases. Therefore with $W$ large enough, $d^{\text{w/wo~CSIT}}_{\mathrm{sum}(M,N_{\rm dl},M_{\rm ul},M)}(r_{\rm sum})$ will lastly be dominated by side-channel condition in the last admissible interval. With the special structure of the decay slope, in order to find the conditions where DMT w/wo achieves PTP performance,
it suffices to show: (A)~the decay slope of side-channel given sum multiplexing gain is larger than $d_{M,m_{\rm I}}(r)$ in their last admissible intervals, respectively; (B)~$\max (r_{\rm sum})\geq 2m_{\rm I}$.

Under the CSIT assumption, from Corollary~\ref{sc2}, we know the maximum sum multiplexing gain is $m_{X}+m_{\rm I}W\alpha_{\rm S}$. We set $m_{X}+m_{\rm I}W\alpha_{\rm S}=2m_{\rm I}$ to meet Condition~(B) thus $W=\frac{1}{\alpha_{\mathrm{S}}}\left(2-\frac{m_X}{m_{\rm I}}\right)^+$.
Next to meet condition (A),  the decay slope of the side-channel in the last interval $\alpha_{\mathrm{S}} d_{M_{\mathrm{ul}},N_{\mathrm{dl}}}\left(\frac{r_{\mathrm{sum}}-m_X}{W\alpha_{\mathrm{S}}}\right),~ \forall r_{\mathrm{sum}} \in[m_X, m_X+m_{\mathrm{I}}W\alpha_{\mathrm{S}}]$, i.e., $\frac{2}{W}(|M_{\rm ul}-N_{\rm dl}|+1)$ should be larger than the decay slope of $d_{M,m_{\rm I}}(r),~\forall r\in[0,m_{\rm I}]$ in its last interval, i.e., $M-m_{\rm I}+1$. Hence $\frac{2}{W}(|M_{\rm ul}-N_{\rm dl}|+1)\geq (M-m_{\rm I}+1)$. By substituting $W=\frac{1}{\alpha_{\mathrm{S}}}\left(2-\frac{m_X}{m_{\rm I}}\right)^+$ into the inequality above,  we have $\alpha_{\mathrm{S}}\!\geq\!\frac{(2m_{\rm I}-m_X)(M-m_{\rm I}+1)}{m_{\rm I}(2|N_{\mathrm{dl}}-M_{\mathrm{ul}}|+2)}.$
With the side-channel condition derived above, the DMT with CSIT achieves the PTP DMT.

Under the no-CSIT assumption,  when $N_{\mathrm{dl}}\geq M_{\mathrm{ul}}$, the results can be derived similarly. 
When $M_{\rm ul}>N_{\rm dl}$, the maximum multiplexing gain is $N_{\rm dl}(1+W\alpha_{\rm S})$ according to Corollary~\ref{nocsitgdof}. To satisfy condition II, we set $N_{\rm dl}(1+W\alpha_{\rm S})=2N_{\rm dl}$, hence we have $W=\frac{1}{\alpha_{\rm S}}$. To meet Condition~(A), the decay slope of the side-channel in the last interval $\alpha_{\mathrm{S}} d_{M_{\mathrm{ul}},N_{\mathrm{dl}}}\left(\frac{r_{\mathrm{sum}}-N_{\rm dl}}{W\alpha_{\mathrm{S}}}\right),~\forall r_{\mathrm{sum}} \in[N_{\rm dl}, N_{\rm dl}(1+W\alpha_{\mathrm{S}})]$, i.e., $\frac{2}{W}(M_{\rm ul}-N_{\rm dl}+1)$, should be greater than the decay slope of $d_{M,N_{\rm dl}}(r)$ in its last interval, i.e.,  $(M-N_{\rm dl}+1)$. By substituting 
$W=\frac{1}{\alpha_{\mathrm{S}}}$, we obtain that $\alpha_{\mathrm{S}}\geq\frac{M-N_{\mathrm{dl}}+1}{2(M_{\mathrm{ul}}-N_{\mathrm{dl}}+1)}$. 
\end{proof}
\subsection{Discussion of the Results}
Fig.~\ref{fig} illustrates the comparison of the three systems in DMT as a function of the side-channel bandwidth. When $N_{\mathrm{dl}}\geq M_{\mathrm{ul}}$, there are three regimes in comparison of DMT. In the first regime, the performance the system without CSIT is worse than that with CSIT. In the second regime, with side-channel bandwidth ratio $W$ greater than a threshold, CSIT is of no use. In the last regime, the use of side-channel helps reduce the probability of outage event where all users are in error such that the dominant error event is single-user error. On the other hand, when $M_{\mathrm{ul}}> N_{\mathrm{dl}}$, the availability of CSIT always provides an additional gain in performing transmit beamforming. However,  larger side-channel bandwidth aids the no-CSIT system to achieve the no-interference upper bound. 
Note that the strength of the side-channel level $\alpha_{\mathrm{S}}$ is implicitly incorporated in Theorems~\ref{coro1} and \ref{coro3}, thus is omitted in Fig.~\ref{fig}.

\begin{figure}
\centering
\subfigure[$N_{\mathrm{dl}}\geq M_{\mathrm{ul}}$]{\scalebox{0.62}{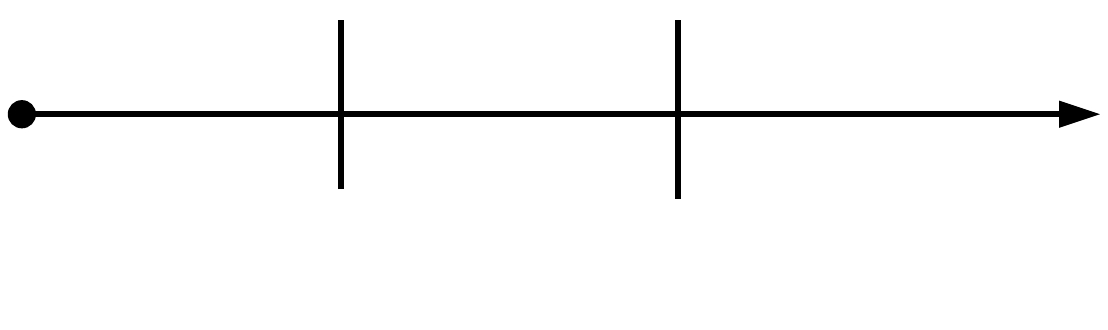}\label{fig:subfig2}}
\hspace{0.6cm}
\subfigure[$M_{\mathrm{ul}}> N_{\mathrm{dl}}$]{\scalebox{0.62}{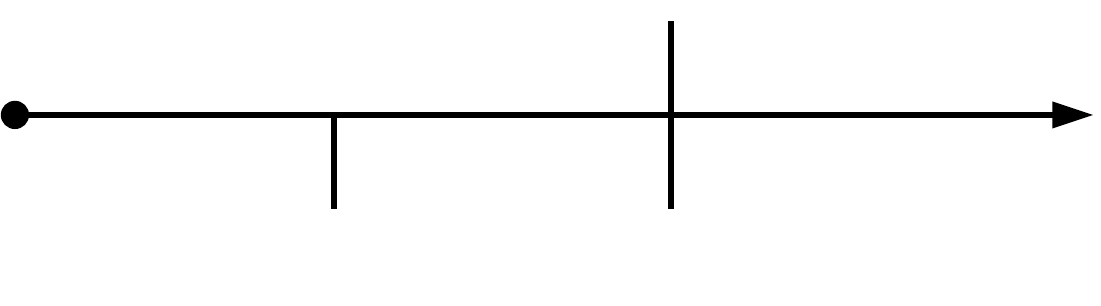}\label{fig:subfig1}}
\caption{Comparison of the three systems in DMT as a function of the side-channel bandwidth.}\label{fig}
\end{figure}

In the following section, we will elaborate the findings in single-antenna mobiles and multiple-antenna mobiles cases, respectively.
\subsubsection{Single-antenna Mobiles}
We first show the symmetric DMT w/wo side-channel and w/wo CSIT when $\alpha_{\mathrm{dl}}=\alpha_{\mathrm{ul}}=\alpha_{\mathrm{I}}=1$. From Fig.~\ref{Spectral}, we can see that in the two-user uplink and downlink system, the full-duplex capable BS  is always superior to its half-duplex~(HD) counterpart where the BS adopts either time-division multiplexing~(TDM) or frequency-division multiplexing (FDM) for uplink and downlink. In the special case of $W=0$, i.e., no side-channel, having CSIT always yields a better DMT performance.
\begin{figure}[h!]
  \centering
    \includegraphics[width=0.5\textwidth,trim = 30mm
      30mm 32mm 20mm, clip]{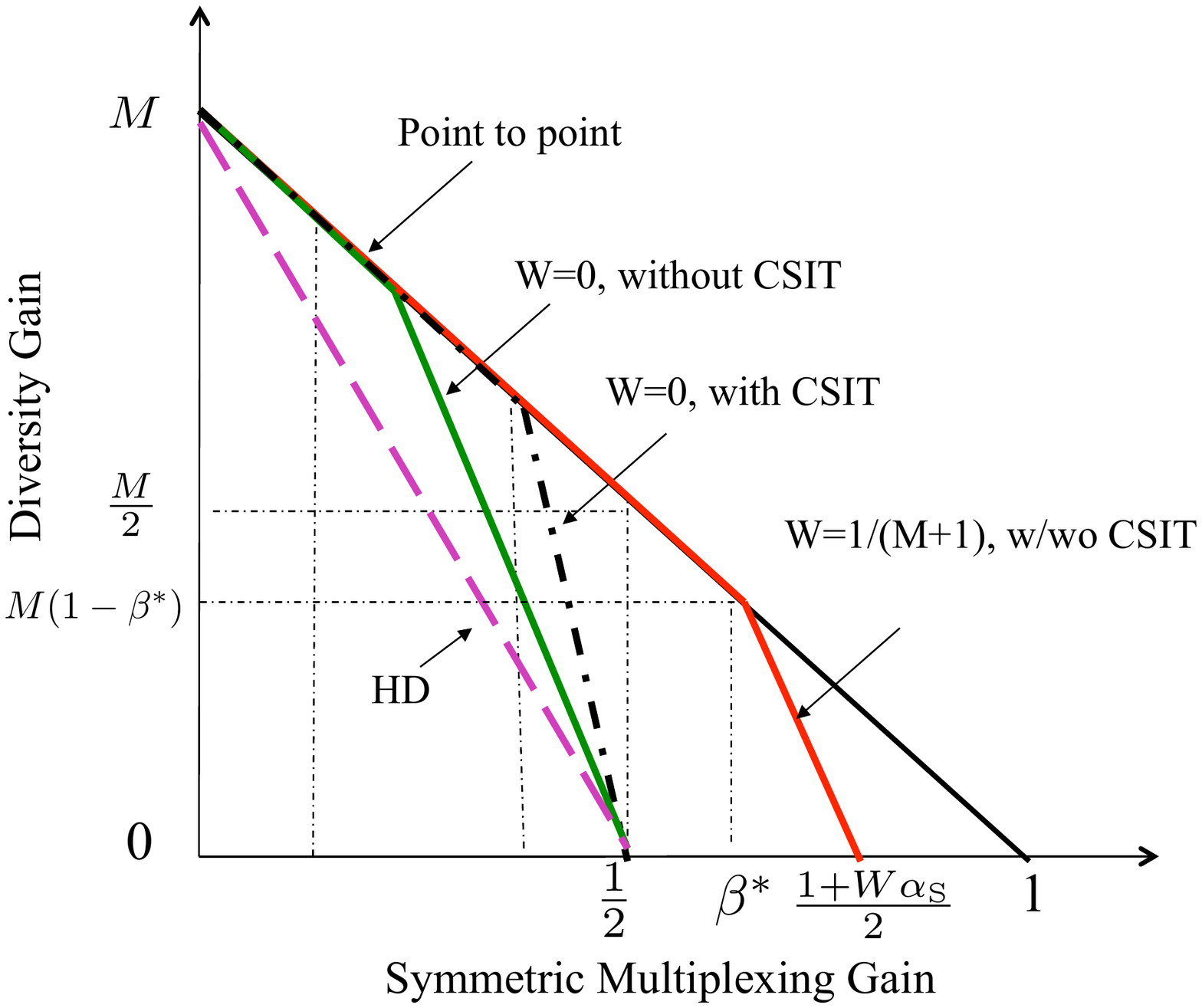}
  \caption{DMT of $(M,1,1,M)$ w/wo side-channel w/wo CSIT when $\alpha_{\mathrm{S}}\geq\frac{M}{2}$, where $\beta^*=\frac{\alpha_{\rm S}+\frac{1}{W}-M}{\frac{2}{W}-M}.$}
\label{Spectral}
\end{figure}
However, with the help of side-channel, as shown in Fig.~\ref{Spectral}, when $W= \frac{1}{M+1}$ and $\alpha_{\mathrm{S}}\geq \frac{M}{2}$, there is no benefit to obtain CSIT as the DMT without CSIT already achieves the optimal DMT with CSIT. Such result indicates that as BS accommodates more antennas (tens or hundreds of BS antennas as in massive MIMO), the required side-channel bandwidth can be reduced superinearly to combat interference.

Fig.~\ref{sst1} illustrates the side-channel bandwidth ratio required to compensate for CSIT as stated in Theorem~\ref{coro1} with single-antenna mobiles. The required $W$ is inversely proportional to the antenna resources at the BS. The caveat is that the side-channel level $\alpha_{\mathrm{S}}$, in the meantime, has to grow with increasing number of antennas at the BS.
\begin{figure}[h!]
  \centering
    \includegraphics[width=0.4\textwidth,trim = 10mm
      60mm 20mm 65mm, clip]{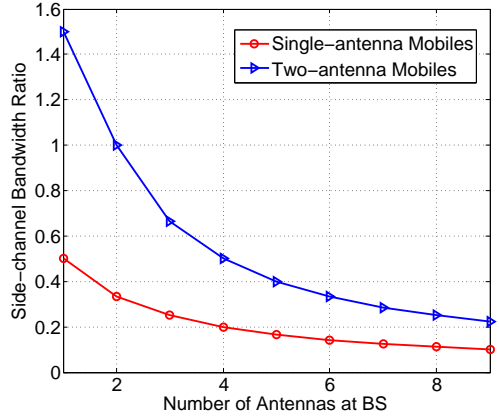}
  \caption{The required side-channel bandwidth ratio to compensate for CSIT as a function of the number of antennas at the BS with equal number of antennas at mobiles when $\alpha_{\mathrm{S}}=\frac{M}{2}$.}
\label{sst1}
\end{figure}

To understand the result above, let us look at the different decay slopes in DMT in the network. From the downlink's viewpoint, the channel is MAC with side-channel. The decay slope of MAC without CSIT is $M+1$, while the the decay slope of the side-channel is $\frac{1}{W}$. When the symmetric multiplexing gain $r\leq \frac{1}{2}$, if $W\geq \frac{1}{M+1}$, the users in MAC will first be in error followed by the users' error event in the side-channel. Moreover, if $\alpha_{\mathrm{S}}\geq \frac{M}{2}$, the error event w/wo CSIT is dominated by single-user performance when $r\leq \frac{1}{2}$. And when $r\geq \frac{1}{2}$, the dominant error event is determined by the side-channel, which is the same for both CSIT and no-CSIT cases.\footnote{The DMT of MAC channel with CSIT is different from that without CSIT as shown in Fig.~\ref{Spectral}.}


In order to eliminate the effect of interference such that the DMT w/wo CSIT achieve no-interference upper bound, it is sufficient if the side-channel condition satisfies that $W\alpha_{\mathrm{S}}\geq 1$ according to Theorem~\ref{coro3}. Hence the required side-channel bandwidth is inversely proportional to the strength of the side-channel as to eliminate the effect of interference.
The implication of such result is that in a highly clustered urban scenario, when the mobile devices are close to each other indicating higher side-channel strength, less side-channel bandwidth is required to achieve the single-user DMT performance.
\subsubsection{Multiple-antenna Mobiles}
Fig.~\ref{MIMO1} shows the DMT in the absence of the side-channel when both the mobiles have multiple antennas. First, we can find out that the gains due to the full-duplex capable BS over half-duplex BS is particularly larger for MIMO channels.
Second, a larger number of downlink receive antennas alone can completely eliminate the effect of CSIT such that the DMT w/wo CSIT have the same performance as stated in Theorem~\ref{coro1}. For example, the DMT of $(3,3,2,3)$ without CSIT is the same as that with CSIT. While in the case of $(3,2,3,3)$, the lack of CSIT will result in significant loss. 
\begin{figure}[h!]
\begin{minipage}[b]{0.5\linewidth}
  \centering
     \scalebox{0.4}{\includegraphics[trim = 50mm
      45mm 28mm 40mm, clip]{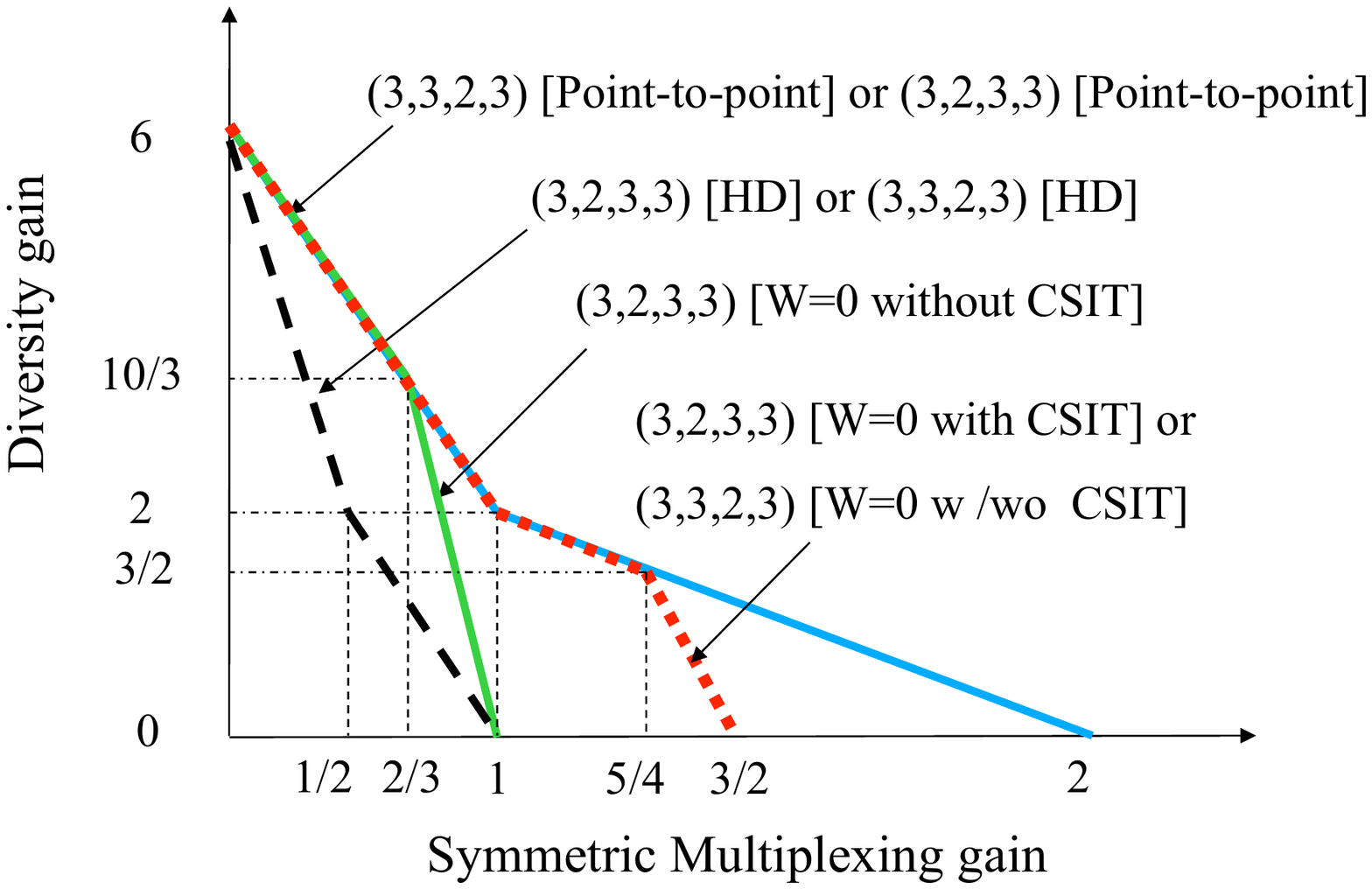}}
  \caption{The symmetric DMT of MIMO full-duplex network without side-channel for $\alpha_{\mathrm{dl}}=\alpha_{\mathrm{ul}}=\alpha_{\mathrm{I}}=1$.}
\label{MIMO1}
\end{minipage}
\hspace{0.1cm}
\begin{minipage}[b]{0.4\linewidth}
  \centering
    \scalebox{0.4}{\includegraphics[trim = 50mm
      45mm 40mm 45mm, clip]{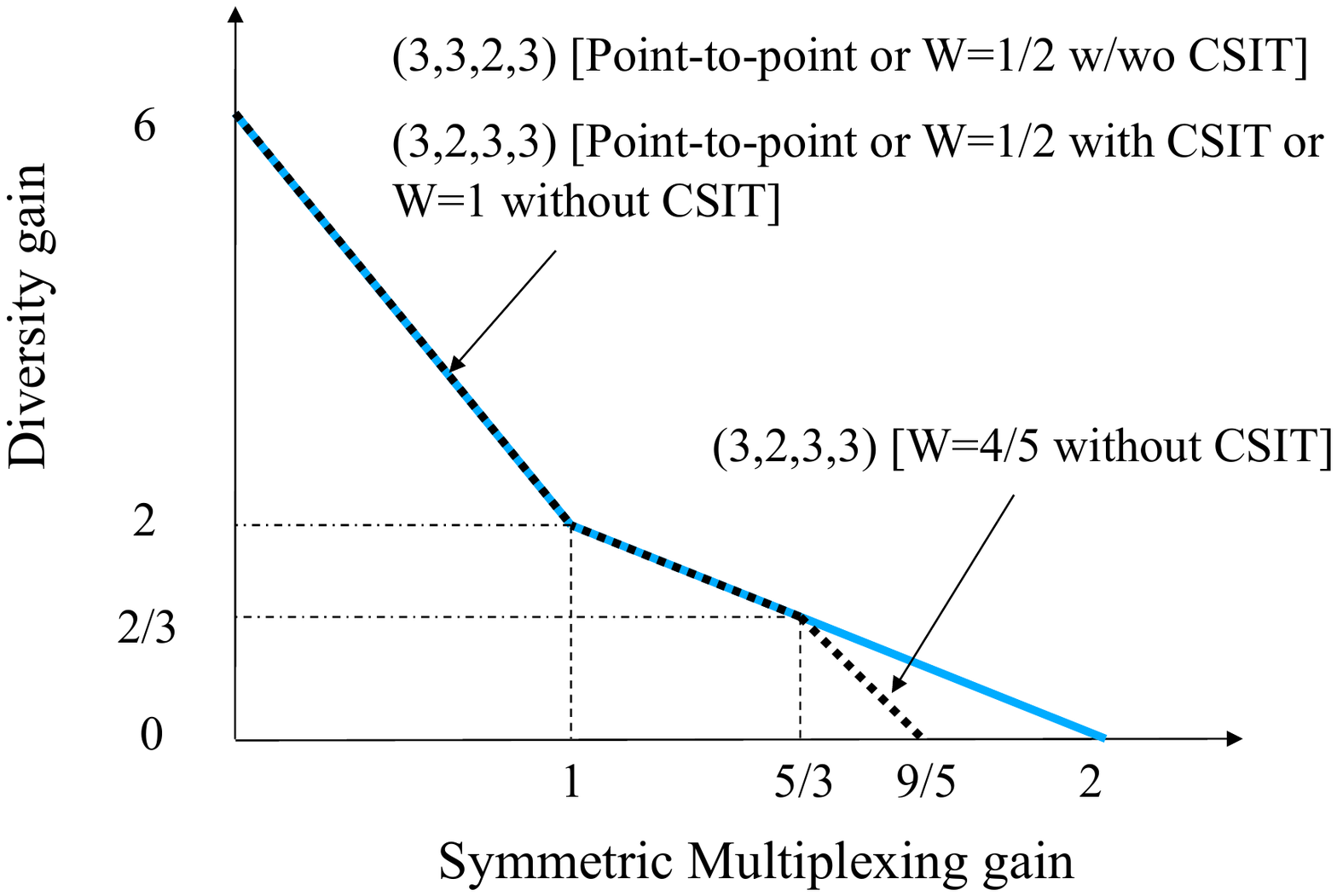}}
  \caption{The symmetric DMT with side-channel for $\alpha_{\mathrm{dl}}=\alpha_{\mathrm{ul}}=\alpha_{\mathrm{I}}=\alpha_{\mathrm{S}}=1$.}
\label{MIMO}
\end{minipage}
\end{figure}

Comparing Fig.~\ref{MIMO1} and Fig.~\ref{MIMO}, we can quantify the gains due to the extra side-channel bandwidth, which is significant especially in MIMO. In the case of $(3,2,3,3)$ when the system is lightly loaded, for instance, $r\leq 2/3$, there is no additional gain due to CSIT or $r\leq 5/4$, there is no gain due to the side-channel since the error event is dominated by single-user error. Beyond those points, the dominant error event is that all users are in error, thus leveraging the CSIT for transmit beamforming or side-channel to perform vector bin-and-cancel will reduce the probability that such outage event happens.

The required side-channel bandwidth ratio for compensation of CSIT in the case of two-antenna mobiles is also depicted in Fig.~\ref{sst1}, which again demonstrates that the required $W\propto \frac{1}{M}$ similar as in the single-antenna-mobile case.

From Theorem~\ref{coro3}, we conclude that with CSIT, as the antenna number ratio $\frac{\max\{M_{\mathrm{ul}},N_{\mathrm{dl}}\}}{\min\{M_{\mathrm{ul}},N_{\mathrm{dl}}\}}$ increases, the side-channel bandwidth required to completely eliminate the effect of interference reduces. 
Hence the spatial resources
of the multiple antennas at mobiles is
interchangeable with the spectral resources of the side-channel
bandwidth to reduce the outage probability at a given multiplexing gain such that single-user DMT can be achieved. 


We also infer from Theorem~\ref{coro3} that when $M_{\mathrm{ul}}> N_{\mathrm{dl}}$, the system with CSIT always outperforms that without CSIT by requiring less side-channel bandwidth to reach single-user performance.\footnote{The system with CSIT also has a weaker requirement of the side-channel strength level as compared to that without CSIT.} However, when $N_{\mathrm{dl}}\geq M_{\mathrm{ul}}$, there is no advantage due to CSIT to achieve the single-user DMT since, with and without CSIT require the same amount of side-channel bandwidth to achieve interference-free performance. Thus we conclude that having more spatial degree-of-freedom at the interfered downlink receiver or larger side-channel bandwidth can simplify transceiver design by ruling out the necessity of obtaining CSIT to null out the effect of inter-mobile interference.

\section{Conclusion}
In this paper, we propose the use of wireless side-channel to manage inter-mobile interference in MIMO full-duplex network where the BS supports both an up- and downlink flow in the same band simultaneously for the half-duplex mobile nodes. We study if and how the antennas resources at nodes will impact the spectral resource from the side-channel under different channel models. For time-invariant channels, we derive a constant-gap capacity region by a vector bin-and-cancel scheme and the corresponding $\mathsf{GDoF}$ region. And for slow-fading channels, we obtain DMT w/wo CSIT of the system. Both the $\mathsf{GDoF}$ and DMT results reveal various insights about the effect of the side-channels and the spatial and spectral tradeoff between antenna resources and bandwidth of the side-channels. Our future work will be to develop practical protocols guided by our analysis. 

\appendix
\subsection{Proof of Lemma~\ref{outerbound}} \label{outer}
First we complete the converse part.
Transmitters uniformly and independently generate the downlink and uplink messages $\omega_{\mathrm{dl}}$ and $\omega_{\mathrm{ul}}$, respectively. The messages will be delivered over $n$ time blocks. Since the full-duplex BS has an implicit feedback of infinite capacity link,  BS encodes the  $\omega_{\mathrm{dl}}$ by codeword $X_{\mathrm{dl},i}$ which is a function of $(\omega_{\mathrm{dl}},Y_{\mathrm{ul}}^{i-1})$,
for $i\in[1,n]$. The point-to-point outer bounds on $R_{\mathrm{dl}}$ and $R_{\mathrm{ul}}$ can be easily obtained following the same argument in Lemma 1 of \cite{MIMOZ}, which are given by
\begin{gather}
\begin{aligned}
R_{\mathrm{dl}}&\leq W_m\bigg(\mathrm{log}\left|I_{N_{\mathrm{dl}}}+\rho_{\mathrm{dl}}H_{\mathrm{dl}}H_{\mathrm{dl}}^\dagger\right|\bigg),\\
R_{\mathrm{ul}}&\leq W_m\bigg( \mathrm{log}\left|I_{N_{\mathrm{ul}}}+\bar{\lambda}\rho_{\mathrm{ul}}H_{\mathrm{ul}}H_{\mathrm{ul}}^\dagger\right|\bigg).
\end{aligned}
\end{gather}
Next we need to prove the sum-capacity upper bound. We define a genie $V_{\rm ul}=\sqrt{\bar{\lambda}\rho_{\mathrm{I}}}H_{\mathrm{I}}X_{\mathrm{ul}}+Z_{\mathrm{dl}}$. The sum-capacity upper bound is derived by providing the genie $V^n_{\mathrm{ul}}$ to the BS. 
By Fano's inequality, for any codebook of block length $n$,
\begin{eqnarray}
n(R_{\mathrm{dl}}+R_{\mathrm{ul}}-\epsilon_n) &\leq &I(\omega_{\mathrm{dl}};Y_{\mathrm{dl}}^n,Y_{\mathrm{S}}^n)+I(\omega_{\mathrm{ul}};Y_{\mathrm{ul}}^n|\omega_{\mathrm{dl}}) \label{a1}\\
&=&I(\omega_{\mathrm{dl}};Y_{\mathrm{dl}}^n)+I(\omega_{\mathrm{dl}};Y_{\mathrm{S}}^n|Y_{\mathrm{dl}}^n)+I(\omega_{\mathrm{ul}};Y_{\mathrm{ul}}{}^n|\omega_{\mathrm{dl}})\\
&=&I(\omega_{\mathrm{dl}};Y_{\mathrm{dl}}^n)+h(Y_{\mathrm{S}}^n|Y_{\mathrm{dl}}^n)-h(Y_{\mathrm{S}}^n|Y_{\mathrm{dl}}^n,\omega_{\mathrm{dl}})\nonumber\\
&&+I(\omega_{\mathrm{ul}};Y_{\rm ul}{}^n|\omega_{\mathrm{dl}}) \label{aa}\\
&\leq&I(\omega_{\mathrm{dl}};Y_{\mathrm{dl}}^n)+h(Y_{\mathrm{S}}^n)-h(Y_{\mathrm{S}}^n|X_{\rm S}^n,Y_{\mathrm{dl}}^n,\omega_{\mathrm{dl}})\nonumber\\
&&+I(\omega_{\mathrm{ul}};Y_{\rm ul}{}^n|\omega_{\mathrm{dl}})\label{bb}\\
&=&h(Y_{\mathrm{dl}}^n)+h(Y_{\mathrm{S}}^n)-h(Y_{\mathrm{S}}^n|X_{\rm S}^n)\nonumber\\
&&+\underbrace{h(Y_{\mathrm{ul}}^n|\omega_{\mathrm{dl}})-h(Y_{\mathrm{dl}}^n|\omega_{\mathrm{dl}})-h(Y_{\mathrm{ul}}^n|\omega_{\mathrm{ul}},\omega_{\mathrm{dl}})}_{U}\label{cc}
\end{eqnarray}
where (\ref{a1}) follows due to the independence of messages; (\ref{bb}) follows because conditioning reduces entropy; (\ref{cc}) follows because $(Y_{\rm dl}^n,\omega_{\rm dl})\rightarrow X_{\rm S}^n\rightarrow Y_{\rm S}^n$ forms a Markov chain.

We can rewrite $h(Y_{\mathrm{ul}}^n|\omega_{\mathrm{ul}},\omega_{\mathrm{dl}})$ in (\ref{cc}) in $U$ as
\begin{eqnarray}
h(Y_{\mathrm{ul}}^n|\omega_{\mathrm{ul}},\omega_{\mathrm{dl}})&=&\sum_{i=1}^{n}h(Y_{\mathrm {ul},i}|Y_{\rm ul}^{i-1},\omega_{\rm ul},\omega_{\rm dl})\\
&=&\sum_{i=1}^n h(Y_{\mathrm {ul},i}|X_{\mathrm{ul},i},Y_{\mathrm {ul}}^{i-1},\omega_{\rm ul},\omega_{\rm dl})\label{11}\\
&=&\sum_{i=1}^{n}h(Z_{\mathrm {ul},i})
\end{eqnarray}
where (\ref{11}) follows because $X_{\mathrm {ul},i}$ is a function of $\omega_{\rm ul}$ and conditioned on $X_{\mathrm {ul},i}$, $Y_{\mathrm {ul},i}$ is independent of everything else.

We also rewrite $h(Y_{\mathrm{ul}}^n|\omega_{\mathrm{dl}})-h(Y_{\mathrm{dl}}^n|\omega_{\mathrm{dl}})$ in (\ref{cc}) in $U$ as
\begin{eqnarray}
&&\!\!\!\!\!\!\!\!\!\!\!\!\!\!\!\!\!\!h(Y_{\mathrm{ul}}^n|\omega_{\mathrm{dl}})-h(Y_{\mathrm{dl}}^n|\omega_{\mathrm{dl}})\\
&=&h(Y_{\mathrm{ul}}^n,Y_{\rm dl}^n|\omega_{\mathrm{dl}})-h(Y_{\mathrm{dl}}^n|Y_{\rm ul}^n,\omega_{\mathrm{dl}})-\big(h(Y_{\mathrm{ul}}^n,Y_{\rm dl}^n|\omega_{\mathrm{dl}})-h(Y_{\mathrm{ul}}^n|Y_{\rm dl}^n,\omega_{\mathrm{dl}})\big)\\
&=& h(Y_{\mathrm{ul}}^n|Y_{\rm dl}^n,\omega_{\mathrm{dl}})-h(Y_{\mathrm{dl}}^n|Y_{\rm ul}^n,\omega_{\mathrm{dl}})\\
&=&\sum_{i=1}^n h(Y_{\mathrm{ul},i}|Y_{\rm dl}^n,Y_{\rm ul}^{i-1},\omega_{\rm dl})-\sum_{i=1}^n h(Y_{\mathrm {dl},i}|Y_{\rm ul}^n,Y_{\rm dl}^{i-1},\omega_{\rm dl}) \\
&\leq& \sum_{i=1}^n h(Y_{\mathrm {ul},i}|Y_{\rm dl}^n,Y_{\rm ul}^{i-1},\omega_{\rm dl})-\sum_{i=1}^n h(Y_{\rm dl,i}|X_{\mathrm{ul},i},X_{\mathrm{dl},i},Y_{\rm ul}^n,Y_{\rm dl}^{i-1},\omega_{\rm dl}) \label{22}\\
&=& \sum_{i=1}^n h(Y_{\mathrm {ul},i}|X_{\mathrm{dl},i},V_{\mathrm{ul},i},Y_{\rm dl}^n,Y_{\rm ul}^{i-1},\omega_{\rm dl})-\sum_{i=1}^n h(Z_{\mathrm {dl},i}) \label{33}\\
&\leq& \sum_{i=1}^n h(Y_{\mathrm {ul},i}|V_{\mathrm {ul},i})-\sum_{i=1}^n h(Z_{\mathrm {dl},i}), \label{44}
\end{eqnarray}
where (\ref{22}) follows because conditioning reduces entropy; (\ref{33}) follows since $X_{\mathrm {dl},i}$ is a function of $(\omega_{\rm dl}, Y_{\rm ul}^{i-1})$ and the genie $V_{\mathrm{ul},i}$ can be determined by $X_{\mathrm{dl},i}$ and $Y_{\mathrm{dl},i}$ as $Y_{\mathrm{dl}}=\sqrt{\rho_{\rm dl}}H_{\rm dl}X_{\rm dl}+V_{\rm ul}$. Also conditioned on $(X_{\mathrm {ul},i},X_{\mathrm {dl},i})$, $Y_{\mathrm {dl},i}$ is independent of everything else; (\ref{44}) follows since removing condition does not reduce entropy.

Thus $U$ can be upper bounded as
\begin{eqnarray}
U&\leq&  \sum_{i=1}^n h(Y_{{\mathrm{ul}},i}|V_{{\mathrm{ul}},i})-\sum_{i=1}^n\big(h(Z_{{\mathrm{dl}},i})+h(Z_{{\mathrm{ul}},i})\big). \label{e}
\end{eqnarray}
Combining the results above and applying the chain rule, we have
\begin{gather}
\begin{aligned}
R_{\mathrm{dl}}+R_{\mathrm{ul}}-\epsilon_n&\leq \frac{1}{n}\sum_{i=1}^n\bigg(h(Y_{{\mathrm{dl}},i})+h(Y_{{\mathrm{S}},i})+h(Y_{{\mathrm{ul}},i}|V_{{\mathrm{ul}},i})-\big[h(Z_{{\mathrm{dl}},i})\\
&+h(Z_{{\mathrm{ul}},i})+h(Z_{{\mathrm{S}},i})\big]\bigg).\nonumber
\end{aligned}
\end{gather}

Now by applying the standard time sharing argument, we can obtain
\begin{eqnarray}
R_{\mathrm{dl}}+R_{\mathrm{ul}}
&\!\!\leq\!\!&\underbrace{h(Y_{\mathrm{dl}})-h(Z_{\mathrm{dl}})}_{R_{us,1}}+\underbrace{h(Y_{\mathrm{ul}}|V_{\mathrm{ul}})-h(Z_{\mathrm{ul}})}_{R_{us,2}}+\underbrace{h(Y_{\mathrm{S}})-h(Z_{\mathrm{S}})}_{R_{us,3}}\label{prov4}.
\end{eqnarray}
We denote the covariance matrix of $Y_{\rm dl}$ as $K_{Y_{\rm dl}}=\mathbb{E}(Y_{\rm dl}Y_{\rm dl}^\dagger)$ that is maximized by Gaussian input in the presence of Gaussian noise. It can be easily shown that
\begin{eqnarray}
K_{Y_{\rm dl}}=I_{N_{\mathrm{dl}}}+\rho_{\mathrm{dl}}H_{\mathrm{dl}}Q_{\mathrm{dl}}H_{\mathrm{dl}}^\dagger+\bar{\lambda}\rho_{\mathrm{I}}H_{\mathrm{I}}Q_{\mathrm{ul}}H_{\mathrm{I}}^\dagger+\sqrt{\bar{\lambda}\rho_{\mathrm{dl}}\rho_{\mathrm{I}}}H_{\mathrm{dl}}Q_{\mathrm{d,u}}H_{\mathrm{I}}^\dagger \label{cov_y}
+\sqrt{\bar{\lambda}\rho_{\mathrm{dl}}\rho_{\mathrm{I}}}H_{\mathrm{I}}Q_{\mathrm{\mathrm{u,d}}}H_{\mathrm{dl}}^\dagger,
\end{eqnarray}
where $Q_{\mathrm{dl}}=\mathbb{E}(X_{\mathrm{dl}}X_{\mathrm{dl}}^\dagger), Q_{\mathrm{ul}}=\mathbb{E}(X_{\mathrm{ul}}X_{\mathrm{ul}}^\dagger), Q_{\mathrm{d,u}}=\mathbb{E}(X_{\mathrm{dl}}X_{\mathrm{ul}}^\dagger), Q_{\mathrm{\mathrm{u,d}}}=\mathbb{E}(X_{\mathrm{ul}}X_{\mathrm{dl}}^\dagger)$. 

Let $J= \begin{bmatrix} V_{\rm ul}  \\ Y_{\rm ul}  \end{bmatrix}$, the covariance matrix of $J$ denoted by $K_J$ can be maximized with Gaussian inputs, it can be verified that
\begin{gather}
K_{J}=\mathbb{E}(JJ^\dagger)=
\begin{bmatrix} 
I_{N_{\mathrm{dl}}}+\bar{\lambda}\rho_{\mathrm{I}}H_{\mathrm{I}}Q_{\mathrm{ul}}H_{\mathrm{I}}^\dagger & \bar{\lambda}\sqrt{\rho_{\mathrm{ul}}\rho_{\mathrm{I}}}H_{\mathrm{I}}Q_{\mathrm{\mathrm{ul}}}H_{\mathrm{ul}}^\dagger\\
 \bar{\lambda}\sqrt{\rho_{\mathrm{ul}}\rho_{\mathrm{I}}}H_{\mathrm{ul}}Q_{\mathrm{\mathrm{ul}}}H_{\mathrm{I}}^\dagger&
 I_{N_{\mathrm{ul}}}+\bar{\lambda}\rho_{\mathrm{ul}}H_{\mathrm{ul}}Q_{\mathrm{ul}}H_{\mathrm{ul}}^\dagger \label{cov_joint}
\end{bmatrix}.
\end{gather}

Likewise, the covariance matrix of  $Y_{\rm S}$ will be maximized by Gaussian input and computed as
 \begin{eqnarray}
K_{Y_{\rm S}}=\mathbb{E}(Y_{\rm S}Y_{\rm S}^\dagger)=WI_{N_{\mathrm{dl}}}+\lambda\rho_{\mathrm{S}}H_{\mathrm{S}}Q_{\mathrm{S}}H_{\mathrm{S}}^\dagger, \label{cov_side}
\end{eqnarray}
where $Q_{\mathrm{S}}=\mathbb{E}(X_{\mathrm{S}}X_{\mathrm{S}}^\dagger)$. 

Using the result in (\ref{cov_y}), we can upper bound the first term $R_{us,1}$ (bit/s) in (\ref{prov4}) as
\begin{eqnarray}
\frac{R_{us,1}}{W_m}&\leq&\mathrm{log}\bigg|I_{N_{\mathrm{dl}}}+\rho_{\mathrm{dl}}H_{\mathrm{dl}}Q_{\mathrm{dl}}H_{\mathrm{dl}}^\dagger+\bar{\lambda}\rho_{\mathrm{I}}H_{\mathrm{I}}Q_{\mathrm{ul}}H_{\mathrm{I}}^\dagger+\sqrt{\bar{\lambda}\rho_{\mathrm{dl}}\rho_{\mathrm{I}}}H_{\mathrm{dl}}Q_{\mathrm{d,u}}H_{\mathrm{I}}^\dagger
\nonumber\\
&&+\sqrt{\bar{\lambda}\rho_{\mathrm{dl}}\rho_{\mathrm{I}}}H_{\mathrm{I}}Q_{\mathrm{\mathrm{u,d}}}H_{\mathrm{dl}}^\dagger \bigg|\\
&\leq&\mathrm{log}\left|I_{N_{\mathrm{dl}}}+G_{\mathrm{dl}}+G_{\mathrm{ul}}\right|\label{prov6}\\
&=&\mathrm{log}\left|(I_{N_{\mathrm{dl}}}+G_{\mathrm{dl}})(I_{N_{\mathrm{dl}}}+(I_{N_{\mathrm{dl}}}+G_{\mathrm{dl}})^{-1}G_{\mathrm{ul}})\right|\\
&=&\mathrm{log}\left|I_{N_{\mathrm{dl}}}+G_{\mathrm{dl}}\right|+\mathrm{log}\left|I_{N_{\mathrm{dl}}}+(I_{N_{\mathrm{dl}}}+G_{\mathrm{dl}})^{-1}G_{\mathrm{ul}}\right|\\
&\leq &\mathrm{log}\left|I_{N_{\mathrm{dl}}}+G_{\mathrm{dl}}\right|+\mathrm{log}\left|2I_{N_{\mathrm{dl}}}\right|\label{prov7}\\
&=&\mathrm{log}\left|I_{N_{\mathrm{dl}}}+G_{\mathrm{dl}}\right|+N_{\mathrm{dl}},
\end{eqnarray}
where $G_{\mathrm{dl}}\!\!=\!\!\rho_{\mathrm{dl}}H_{\mathrm{dl}}H_{\mathrm{dl}}^\dagger+\bar{\lambda}\rho_{\mathrm{I}}H_{\mathrm{I}}H_{\mathrm{I}}^\dagger, G_{\mathrm{ul}}\!\!=\!\!\sqrt{\bar{\lambda}\rho_{\mathrm{dl}}\rho_{\mathrm{I}}}H_{\mathrm{dl}}Q_{\mathrm{d,u}}H_{\mathrm{I}}^\dagger
+\sqrt{\bar{\lambda}\rho_{\mathrm{dl}}\rho_{\mathrm{I}}}H_{\mathrm{I}}Q_{\mathrm{u,d}}H_{\mathrm{dl}}^\dagger$; (\ref{prov6}) follows because $\text{trace}(Q_i)\leq1, i\in\{\mathrm{ul,dl}\},$ thus $Q_i\preceq I$, and $\mathrm{log}|.|$ is an increasing function on the cone of positive-definite matrices; (\ref{prov7}) follows from the following lemma.
\begin{lemma} \label{Lemmamatrix}
For p.s.d. matrices $G_{\mathrm{dl}}$ and $G_{\mathrm{ul}}$, we have
\begin{eqnarray}
\mathrm{log}\left|I_{N_{\mathrm{dl}}}+(I_{N_{\mathrm{dl}}}+G_{\mathrm{dl}})^{-1}G_{\mathrm{ul}}\right|\leq\mathrm{log}\left|2I_{N_{\mathrm{dl}}}\right|.
\end{eqnarray}
\end{lemma}
\begin{proof}
First we show that $G_{\mathrm{ul}}\preceq G_{\mathrm{dl}}$. Let $A=\sqrt{\rho_{\mathrm{dl}}}H_{\mathrm{dl}}Q_{\rm d,u}-\sqrt{\bar{\lambda}\rho_{\mathrm{I}}}H_{\mathrm{I}}$, the product of matrices $AA^\dagger$ is always p.s.d., because for any vector $x$, 
$x^\dagger AA^\dagger x= (A^\dagger x)^\dagger (A^\dagger x)\geq 0.$
Hence we have the following
\begin{gather}
\begin{aligned}
\rho_{\rm dl}H_{\rm dl}Q_{\rm d,u}Q_{\rm d,u}^\dagger H_{\mathrm{dl}}^\dagger+\bar{\lambda}\rho_{\mathrm{I}}H_{\mathrm{I}}H_{\mathrm{I}}^\dagger\geq \sqrt{\bar{\lambda}\rho_{\mathrm{dl}}\rho_{\mathrm{I}}}H_{\mathrm{dl}}Q_{\mathrm{d,u}}H_{\mathrm{I}}^\dagger
+\sqrt{\bar{\lambda}\rho_{\mathrm{dl}}\rho_{\mathrm{I}}}H_{\mathrm{I}}Q_{\mathrm{d,u}}^\dagger H_{\mathrm{dl}}^\dagger.
\end{aligned}
\end{gather}
Since $Q_{\rm d,u}Q_{\rm d,u}^\dagger\preceq I$, and $Q_{\rm d,u}^\dagger=Q_{\rm u,d}$, now we can obtain that
\begin{gather}
\begin{aligned}
&\rho_{\rm dl}H_{\rm dl}H_{\mathrm{dl}}^\dagger+\bar{\lambda}\rho_{\mathrm{I}}H_{\mathrm{I}}H_{\mathrm{I}}^\dagger\geq \sqrt{\bar{\lambda}\rho_{\mathrm{dl}}\rho_{\mathrm{I}}}H_{\mathrm{dl}}Q_{\mathrm{d,u}}H_{\mathrm{I}}^\dagger
+\sqrt{\bar{\lambda}\rho_{\mathrm{dl}}\rho_{\mathrm{I}}}H_{\mathrm{I}}Q_{\mathrm{u,d}}H_{\mathrm{dl}}^\dagger.
\end{aligned}
\end{gather}
Hence we have verified that $G_{\mathrm{ul}}\preceq G_{\mathrm{dl}}$, which also leads to the fact that $G_{\mathrm{ul}}\preceq I_{N_{\mathrm{dl}}}+G_{\mathrm{dl}}$. As a result, for any given vector $x$, we have
\begin{gather}
\begin{aligned}
x^\dagger\left[\left(I+G_{\mathrm{dl}}\right)^{\frac{1}{2}}\left(I-(I+G_{\mathrm{dl}})^{-\frac{1}{2}}G_{\mathrm{ul}}(I+G_{\mathrm{dl}})^{-\frac{1}{2}}\right)(I+G_{\mathrm{dl}})^{\frac{1}{2}}\right]x&\geq0,~\text{or} \\
\left((I+G_{\rm dl})^{\frac{1}{2}}x\right)^\dagger\left[I-(I+G_{\mathrm{dl}})^{-\frac{1}{2}}G_{\mathrm{ul}}(I+G_{\mathrm{dl}})^{-\frac{1}{2}}\right]\left((I+G_{\rm dl})^{\frac{1}{2}}x\right)&\geq0.
\end{aligned}
\end{gather}
From the definition of partial order of p.s.d. matrices~\cite{horn1990matrix}, we have $(I+G_{\mathrm{dl}})^{-\frac{1}{2}}G_{\mathrm{ul}}(I+G_{\mathrm{dl}})^{-\frac{1}{2}}\preceq I$.
Hence we have that
\begin{eqnarray}
&&\mathrm{log}\left|I_{N_{\mathrm{dl}}}+(I_{N_{\mathrm{dl}}}+G_{\mathrm{dl}})^{-1}G_{\mathrm{ul}}\right|\\
&=&\mathrm{log}\left|(I_{N_{\mathrm{dl}}}+G_{\mathrm{dl}})^{-\frac{1}{2}}\left(I_{N_{\mathrm{dl}}}+(I_{N_{\mathrm{dl}}}+G_{\mathrm{dl}})^{-\frac{1}{2}}G_{\mathrm{ul}}(I_{N_{\mathrm{dl}}}+G_{\mathrm{dl}})^{-\frac{1}{2}}\right)(I_{N_{\mathrm{dl}}}
+G_{\mathrm{dl}})^{\frac{1}{2}}\right|\nonumber\\
\!\!\!\!\!&=\!\!\!\!\!&\mathrm{log}\left|I_{N_{\mathrm{dl}}}+(I_{N_{\mathrm{dl}}}+G_{\mathrm{dl}})^{-\frac{1}{2}}G_{\mathrm{ul}}(I_{N_{\mathrm{dl}}}+G_{\mathrm{dl}})^{-\frac{1}{2}}\right|\leq\mathrm{log}\left|2I_{N_{\mathrm{dl}}}\right|.\nonumber
\end{eqnarray}
\end{proof}
The second term in (\ref{prov4}) can be written as 
\begin{eqnarray}
R_{us,2}=h(Y_{\mathrm{ul}},V_{\mathrm{ul}})-h(V_{\rm ul})-h(Z_{\mathrm{ul}}).\label{sec}
\end{eqnarray}
Using the covariance matrix we derived in (\ref{cov_joint}), and invoking Lemma 8 in \cite{MIMOZ}, we can upper bound (\ref{sec}) as follows,
\begin{eqnarray}
\frac{R_{us,2}}{W_m}\leq \mathrm{log}\left|I_{N_{\mathrm{ul}}}+\bar{\lambda}\rho_{\mathrm{ul}}H_{\mathrm{ul}}(I_{M_{\mathrm{ul}}}+\bar{\lambda}\rho_{\mathrm{I}}H_{\mathrm{I}}^\dagger H_{\mathrm{I}})^{-1}H_{\mathrm{ul}}^\dagger\right|.
\end{eqnarray}

Finally, we can upper bound the third term in (\ref{prov4}) with the covariance matrix in (\ref{cov_side}),
\begin{eqnarray}
\frac{R_{us,3}}{W_s}&\leq &\mathrm{log}\left|WI_{N_{\mathrm{dl}}}+\lambda\rho_{\mathrm{S}}H_{\mathrm{S}}H_{\mathrm{S}}^\dagger\right|-\mathrm{log}\left|W I_{N_{\rm dl}}\right|\\
&=&\mathrm{log}\left|I_{N_{\mathrm{dl}}}+\frac{\lambda\rho_{\mathrm{S}}}{W}H_{\mathrm{S}}H_{\mathrm{S}}^\dagger\right|.
\end{eqnarray}
Combining all the results we derived above, we can prove Lemma~\ref{outerbound}.

\subsection{Rate Calculation in Lemma~\ref{innerbound}} \label{inner}
For the Gaussian inputs with the covariance matrices given in (\ref{powerallocate}), the achievable rate in (\ref{discreteAR}) can be calculated as
\begin{eqnarray}
I(X_{\mathrm{dl}};Y_{\mathrm{dl}}|S_{\mathrm{ul}})&=&W_m\bigg(\mathrm{log}\left|I_{N_{\mathrm{dl}}}+\frac{\rho_{\mathrm{dl}}}{M_{\mathrm{dl}}}H_{\mathrm{dl}}H_{\mathrm{dl}}^\dagger+\bar{\lambda}\rho_{\mathrm{I}}H_{\mathrm{I}}K_u H_{\mathrm{I}}^\dagger\right|\nonumber\\
&&-\mathrm{log}\left|I_{N_{\mathrm{dl}}}+\bar{\lambda}\rho_{\mathrm{I}}H_{\mathrm{I}}K_u H_{\mathrm{I}}^\dagger\right|\bigg)\\
&\geq&W_m\bigg(\mathrm{log}\left|I_{N_{\mathrm{dl}}}+\frac{\rho_{\mathrm{dl}}}{M_{\mathrm{dl}}}H_{\mathrm{dl}}H_{\mathrm{dl}}^\dagger+\bar{\lambda}\rho_{\mathrm{I}}H_{\mathrm{I}}K_u H_{\mathrm{I}}^\dagger\right|-\hat{m}_{\mathrm{I}}\bigg)\label{step2}\\
&\geq&W_m\bigg(\mathrm{log}\left|I_{N_{\mathrm{dl}}}+\rho_{\mathrm{dl}}H_{\mathrm{dl}}H_{\mathrm{dl}}^\dagger\right|-m_{\mathrm{dl}}\mathrm{log}M_{\mathrm{dl}}-\hat{m}_{\mathrm{I}}\bigg),
\end{eqnarray}
where $K_u=\frac{1}{M_{\mathrm{ul}}}(I_{M_{\mathrm{ul}}}+\bar{\lambda}\rho_{\mathrm{I}}H_{\mathrm{I}}^\dagger H_{\mathrm{I}})^{-1}, \hat{m}_{\mathrm{I}}=m_{\mathrm{I}}\mathrm{log}\left(1+\frac{1}{M_{\mathrm{ul}}}\right)$; $m_{\mathrm{I}}=\min\{M_{\mathrm{ul}},N_{\mathrm{dl}}\}, m_{\mathrm{dl}}=\min\{M_{\mathrm{dl}},N_{\mathrm{dl}}\}$, which are the rank of $H_{\mathrm{I}}$ and $H_{\mathrm{dl}}$, respectively.
Step~(\ref{step2}) is established because of the following argument: applying the singular value decomposition to $H_{\mathrm{I}}$ such that $H_{\mathrm{I}}=U\Lambda V^{\dagger}$, where $U$ and $V$ are unitary matrices, $\Lambda$ is $N_{\mathrm{dl}}\times M_{\mathrm{ul}}$ diagonal matrix containing singular values of $H_{\mathrm{I}}$. Now we can rewrite $\bar{\lambda}\rho_{\mathrm{I}}H_{\mathrm{I}}K_u H_{\mathrm{I}}^\dagger$ as
\begin{eqnarray}
\bar{\lambda}\rho_{\mathrm{I}}H_{\mathrm{I}}K_u H_{\mathrm{I}}^\dagger=\frac{\bar{\lambda}\rho_{\rm I }}{M_{\rm ul}}U\Lambda(I_{M_{\mathrm{ul}}}+\bar{\lambda}\rho_{\mathrm{I}}\Lambda^\dagger \Lambda)^{-1}\Lambda^\dagger U^\dagger.
\end{eqnarray}
Since $\bar{\lambda}\rho_{\rm I}\Lambda(I_{M_{\mathrm{ul}}}+\bar{\lambda}\rho_{\mathrm{I}}\Lambda^\dagger \Lambda)^{-1}\Lambda^\dagger \leq I_{N_{\rm dl}}$, for p.s.d. matrices, given a vector $x$, we can show that
\begin{eqnarray}
x^\dagger\left(\bar{\lambda}\rho_{\mathrm{I}}H_{\mathrm{I}}K_u H_{\mathrm{I}}^\dagger\right)x&=&\frac{1}{M_{\rm ul}}(U^\dagger x)^\dagger \bar{\lambda}\rho_{\rm I}\Lambda(I_{M_{\mathrm{ul}}}+\bar{\lambda}\rho_{\mathrm{I}}\Lambda^\dagger \Lambda)^{-1}\Lambda^\dagger (U^\dagger x)\\
&\leq&\frac{1}{M_{\rm ul}}(U^\dagger x)^\dagger I_{N_{\rm dl}} (U^\dagger x).
\end{eqnarray}
Thus $\bar{\lambda}\rho_{\mathrm{I}}H_{\mathrm{I}}K_u H_{\mathrm{I}}^\dagger\leq \frac{1}{M_{\rm ul}}I_{N_{\rm dl}}$, which implies that
\begin{eqnarray}
\mathrm{log}\left|I_{N_{\mathrm{dl}}}+\bar{\lambda}\rho_{\mathrm{I}}H_{\mathrm{I}}K_u H_{\mathrm{I}}^\dagger\right|\leq \min\{M_{\mathrm{ul}},N_{\mathrm{dl}}\}\mathrm{log}\left(1+\frac{1}{M_{\rm ul}}\right)\triangleq \hat{m}_{\mathrm{I}}.\label{psd}
\end{eqnarray}

Next we compute $I(S_{\mathrm{ul}},U_{\mathrm{ul}};Y_{\mathrm{ul}})$ as follows,
\begin{eqnarray}
I(S_{\mathrm{ul}},U_{\mathrm{ul}};Y_{\mathrm{ul}})&\!\!\!\!=\!\!\!\!&W_m\mathrm{log}\left|I_{N_{\mathrm{ul}}}+\frac{\bar{\lambda}\rho_{\mathrm{ul}}}{M_{\mathrm{ul}}}H_{\mathrm{ul}}H_{\mathrm{ul}}^\dagger\right|\\
&\!\!\!\!\geq\!\!\!\!&W_m\left(\mathrm{log}\left|I_{N_{\mathrm{ul}}}+\bar{\lambda}\rho_{\mathrm{ul}}H_{\mathrm{ul}}H_{\mathrm{ul}}^\dagger\right|-m_{\mathrm{ul}}\mathrm{log}M_{\mathrm{ul}}\right),
\end{eqnarray}
where $m_{\mathrm{ul}}=\min\{M_{\mathrm{ul}},N_{\mathrm{ul}}\}$, which is the rank of $H_{\mathrm{ul}}$.
Similarly,
\begin{eqnarray}
I(X_{\mathrm{S}};Y_{\mathrm{S}})&\!\!\!\!=\!\!\!\!&W_s\mathrm{log}\left|I_{N_{\mathrm{dl}}}+\frac{\lambda\rho_{\mathrm{S}}}{M_{\mathrm{ul}}W}H_{\mathrm{S}}H_{\mathrm{S}}^{\dagger}\right|\\
&\!\!\!\!\geq\!\!\!\!&W_m\left(W\mathrm{log}\left|I_{N_{\mathrm{dl}}}+\frac{\lambda\rho_{\mathrm{S}}}{W}H_{\mathrm{S}}H_{\mathrm{S}}^{\dagger}\right|-m_{\mathrm{I}}W\mathrm{log}M_{\mathrm{ul}}\right).
\end{eqnarray}
And
\begin{eqnarray}
I(U_{\mathrm{ul}};Y_{\mathrm{ul}}|S_{\mathrm{ul}})&\!\!\!\!=\!\!\!\!&W_m\mathrm{log}\left|I_{N_{\mathrm{ul}}}+\bar{\lambda}\rho_{\mathrm{ul}}H_{\mathrm{ul}}K_u H_{\mathrm{ul}}^\dagger\right|\\
&\!\!\!\!\geq\!\!\!\!&W_m\left(\mathrm{log}\left|I_{N_{\mathrm{ul}}}+\bar{\lambda}\rho_{\mathrm{ul}}H_{\mathrm{ul}}(I_{M_{\mathrm{ul}}}+\bar{\lambda}\rho_{\mathrm{I}}H_{\mathrm{I}}^\dagger H_{\mathrm{I}})^{-1}H_{\mathrm{ul}}^\dagger\right|-m_{\mathrm{ul}}\mathrm{log}M_{\mathrm{ul}}\right),
\end{eqnarray}
\begin{eqnarray}
I(S_{\mathrm{ul}};Y_{\mathrm{dl}}|X_{\mathrm{dl}})&\!\!\!\!=\!\!\!\!&W_m\bigg(\mathrm{log}\left|I_{N_{\mathrm{dl}}}+\frac{\bar{\lambda}\rho_{\mathrm{I}}}{M_{\mathrm{ul}}}H_{\mathrm{I}}H_{\mathrm{I}}^\dagger\right|
-\mathrm{log}\left|I_{N_{\mathrm{dl}}}+\bar{\lambda}\rho_{\mathrm{I}}H_{\mathrm{I}}K_u H_{\mathrm{I}}^\dagger\right|\bigg) \\
&\geq&W_m\bigg(\mathrm{log}\left|I_{N_{\mathrm{dl}}}+\frac{\bar{\lambda}\rho_{\mathrm{I}}}{M_{\mathrm{ul}}}H_{\mathrm{I}}H_{\mathrm{I}}^\dagger\right|
-\hat{m}_{\mathrm{I}}\bigg) \label{step13}\\ 
&\geq&W_m\bigg(\mathrm{log}\left|I_{N_{\mathrm{dl}}}+\bar{\lambda}\rho_{\mathrm{I}}H_{\mathrm{I}}H_{\mathrm{I}}^\dagger\right|
-m_{\mathrm{I}}\mathrm{log}M_{\mathrm{ul}}-\hat{m}_{\mathrm{I}}\bigg)\\
&=&W_m\bigg(\mathrm{log}\left|I_{N_{\mathrm{dl}}}+\bar{\lambda}\rho_{\mathrm{I}}H_{\mathrm{I}}H_{\mathrm{I}}^\dagger\right|-m_{\mathrm{I}}\mathrm{log}(\!M_{\mathrm{ul}}+1)\bigg)
\end{eqnarray}
where (\ref{step13}) follows from step~(\ref{psd}).

Now we can calculate $I(U_{\mathrm{ul}};Y_{\mathrm{ul}}|S_{\mathrm{ul}})+I(S_{\mathrm{ul}};Y_{\mathrm{dl}}|X_{\mathrm{dl}})$ as
\begin{eqnarray}
&&\!\!\!\!\!\!\!\!\!I(U_{\mathrm{ul}};Y_{\mathrm{ul}}|S_{\mathrm{ul}})+I(S_{\mathrm{ul}};Y_{\mathrm{dl}}|X_{\mathrm{dl}})\\
&\geq&W_m\bigg(\mathrm{log}\left|I_{N_{\mathrm{ul}}}+\bar{\lambda}\rho_{\mathrm{ul}}H_{\mathrm{ul}}(I_{M_{\mathrm{ul}}}+\bar{\lambda}\rho_{\mathrm{I}}H_{\mathrm{I}}^\dagger H_{\mathrm{I}})^{-1}H_{\mathrm{ul}}^\dagger\right|+\mathrm{log}\left|I_{N_{\mathrm{dl}}}+\bar{\lambda}\rho_{\mathrm{I}}H_{\mathrm{I}}H_{\mathrm{I}}^\dagger\right|\nonumber\\
&&-m_{\mathrm{ul}}\mathrm{log}M_{\mathrm{ul}}-m_{\mathrm{I}}\mathrm{log}(\!M_{\mathrm{ul}}+1)\bigg)\\
&=&W_m\bigg(\mathrm{log}\left|I_{M_{\mathrm{ul}}}+\bar{\lambda}\rho_{\mathrm{ul}}H_{\mathrm{ul}}^\dagger H_{\mathrm{ul}}(I_{M_{\mathrm{ul}}}+\bar{\lambda}\rho_{\mathrm{I}}H_{\mathrm{I}}^\dagger H_{\mathrm{I}})^{-1}\right|+\mathrm{log}\left|I_{M_{\mathrm{ul}}}+\bar{\lambda}\rho_{\mathrm{I}}H_{\mathrm{I}}^\dagger H_{\mathrm{I}}\right|\nonumber\label{step3}\\
&&-m_{\mathrm{ul}}\mathrm{log}M_{\mathrm{ul}}-m_{\mathrm{I}}\mathrm{log}(\!M_{\mathrm{ul}}+1)\bigg)\\
&=&W_m\bigg(\mathrm{log}\left|I_{M_{\mathrm{ul}}}+\bar{\lambda}\rho_{\mathrm{ul}}H_{\mathrm{ul}}^\dagger H_{\mathrm{ul}}+\bar{\lambda}\rho_{\mathrm{I}}H_{\mathrm{I}}^\dagger H_{\mathrm{I}}\right|-m_{\mathrm{ul}}\mathrm{log}M_{\mathrm{ul}}+m_{\mathrm{I}}\mathrm{log}(M_{\mathrm{ul}}+1)\bigg) \\
&\geq&W_m\bigg(\mathrm{log}\!\left|I_{N_{\mathrm{ul}}}+\bar{\lambda}\rho_{\mathrm{ul}}H_{\mathrm{ul}}H_{\mathrm{ul}}^\dagger\right|-m_{\mathrm{ul}}\mathrm{log}M_{\mathrm{ul}}-m_{\mathrm{I}}\mathrm{log}(\!M_{\mathrm{ul}}+1)\bigg),\label{step4}
\end{eqnarray}
where (\ref{step3}) and (\ref{step4}) follow from Sylvester's determinant theorem.

Finally, we compute $I(X_{\mathrm{dl}},S_{\mathrm{ul}};Y_{\mathrm{dl}})$ as follows,
\begin{eqnarray}
I(X_{\mathrm{dl}},S_{\mathrm{ul}};Y_{\mathrm{dl}})&=&W_m\bigg(\mathrm{log}\left|I_{N_{\mathrm{dl}}}+\frac{\rho_{\mathrm{dl}}}{M_{\mathrm{dl}}}H_{\mathrm{dl}}H_{\mathrm{dl}}^\dagger+\frac{\bar{\lambda}\rho_{\mathrm{I}}}{M_{\mathrm{ul}}}H_{\mathrm{I}} H_{\mathrm{I}}^\dagger\right|\nonumber\\
&&-\mathrm{log}\left|I_{N_{\mathrm{dl}}}+\bar{\lambda}\rho_{\mathrm{I}}H_{\mathrm{I}}K_u H_{\mathrm{I}}^\dagger\right|\bigg)\\
&\geq&W_m\Big(\mathrm{log}\left|I_{N_{\mathrm{dl}}}+\rho_{\mathrm{dl}}H_{\mathrm{dl}}H_{\mathrm{dl}}^\dagger+\bar{\lambda}\rho_{\mathrm{I}}H_{\mathrm{I}} H_{\mathrm{I}}^\dagger\right|-\hat{m}_{\mathrm{I}}\nonumber\\
&&-\min\{M_{\mathrm{dl}}+M_{\mathrm{ul}},N_{\mathrm{dl}}\}\mathrm{log}(\max\{M_{\mathrm{dl}},M_{\mathrm{ul}}\})\Big)\label{last},
\end{eqnarray}
where (\ref{last}) holds because the rank of the matrix $I_{N_{\mathrm{dl}}}+\frac{\rho_{\mathrm{dl}}}{M_{\mathrm{dl}}}H_{\mathrm{dl}}H_{\mathrm{dl}}^\dagger+\frac{\bar{\lambda}\rho_{\mathrm{I}}}{M_{\mathrm{ul}}}H_{\mathrm{I}} H_{\mathrm{I}}^\dagger$ is less than the rank of an enhanced multiple-access channel matrix by allowing full-cooperation between transmitters which is $\min\{M_{\mathrm{dl}}+M_{\mathrm{ul}},N_{\mathrm{dl}}\}$.

Combining all the expressions we derived above, we can obtain the capacity region inner bound as
 \footnote{When $W=0$, we define $W\mathrm{log}\left(1+\frac{x}{W}\right)\triangleq0$.}
 $\mathcal{R}_{\rm BC}(\mathcal{H})=\left\{(R_{\rm dl},R_{\rm ul}): R_{\rm dl}\leq\overline{C}_{\rm dl}-c_1,R_{\rm ul}\leq\overline{C}_{\rm ul}-c_2,R_{\rm sum}\leq\overline{C}_{\rm sum}-(c_1+c_2)\right\}$,
where $c_1$ and $c_2$ are given in (\ref{gap}).

\subsection{Useful Lemmas}\label{citelemma}
Random matrix theory plays a critical role in the analysis of MIMO wireless networks. Here we will restate some important properties of random matrices in the following lemmas which will be used for our derivation. 
\begin{lemma}(Lemma 3 in \cite{DMTMIMOZ})\label{lemmacite}
For a PTP channel, where $H\in \mathbb{C}^{M\times N}$ with i.i.d, $\mathcal{CN}(0,1)$ entries and the channel level is $\alpha$, the optimal DMT is equivalent to the minimum of the following optimization problem, 
\begin{gather}
\begin{aligned}
d(r)=&\min\sum_{i=1}^{\min(M,N)}(M+N+1-2i)x_i\\
\text{s.t}&~\sum_{i=1}^{\min(M,N)}(\alpha-x_i)^+\leq r\\
& 0\leq x_1\leq \cdots \leq x_{\min(M,N)},
\end{aligned}
\end{gather}
and the optimal solution is $d(r)=\alpha d_{M,N}(\frac{r}{\alpha})$, for $0\leq r\leq \min(M,N)\alpha$,
where $d_{M,N}(r)=(M-r)(N-r)$ is a piecewise linear curve joining the integer point $r\in[0,\min(M,N)]$.
\end{lemma}
\begin{lemma}(Theorem 4 in \cite{DMTTse})\label{lemma2cite}
Let $H\in \mathbb{C}^{M\times N}$ have i.i.d, $\mathcal{CN}(0,1)$ entries. Suppose the nonzero ordered eigenvalues of $R=HH^\dagger$ are denoted by $\beta_1\geq\cdots\beta_q>0$, where $q=\min(M,N)$. Let $\beta_i=\rho^{-\mu_i},i\in[1,q]$, assuming that all the eigenvalues vary exponentially with SNR. Let $\bar{\mu}=\{\mu_{\mathrm{1}},\cdots,\mu_q\}$, thus the asymptotic   distribution of $\bar{\mu}$ is 
\begin{gather}
\begin{aligned}
p(\bar{\mu})\doteq
  \begin{cases}
   \rho^{-\sum_{i=1}^{q}(M+N+1-2i)\mu_i}&  ~~~~\text{if}~~0\leq\mu_1\leq\cdots\mu_q \\
   0  &~~~~\text{Otherwise},
   \end{cases}
\end{aligned}
\end{gather}
\end{lemma}

\begin{lemma}(Theorem 1 and 2 in \cite{dmtdf})\label{lemma3cite}
Let $H_1\in \mathbb{C}^{N_2\times N_1}$ and $H_2\in \mathbb{C}^{N_2\times N_3}$ be two mutually independent random matrices with i.i.d, $\mathcal{CN}(0,1)$ entries.  Suppose the ordered eigenvalues of $V_1=H_1^\dagger(I_{N_2}+\rho^{\alpha} H_2H_2^\dagger)^{-1}H_1, V_2=H_2H_2^\dagger$ are denoted by $\beta_1\geq\cdots\beta_q>0, \eta_{\mathrm{1}}\geq\cdots\eta_p>0$ where $q=\min(N_1,N_2),p=\min(N_2,N_3)$. Let $\beta_i=\rho^{-\mu_i},i\in[1,q];\eta_k=\rho^{-\theta_k},k\in[0,p]$, and $\bar{\mu}=\{\mu_{\mathrm{1}},\cdots,\mu_q\},~\bar{\theta}=\{\theta_1,\cdots,\theta_p\}$. Hence the conditional distribution of $\bar{\mu}$ given $\bar{\theta}$ is
\begin{gather}
p(\bar{\mu}|\bar{\theta})\doteq
  \begin{cases}
   \rho^{-E_1(\bar{\mu},\bar{\theta})}&  ~~~~\text{if}~~(\bar{\mu},\bar{\theta})\in\mathcal{D}_1 \\
   0  &~~~~\text{Otherwise},
   \end{cases}
\end{gather}
where 
\begin{gather}
\begin{aligned}
E_1(\bar{\mu},\bar{\theta})=&\sum_{i=1}^{q}(N_1+N_2+1-2i)\mu_i+\sum_{i=1}^{q}\sum_{k=1}^{\min(N_2-i,N_3)}(\alpha-\mu_i-\theta_k)^+-N_1\sum_{k=1}^{p}(\alpha-\theta_k)^+\\
\mathcal{D}_1=&\left\{0\leq\mu_{\mathrm{1}}\leq\cdots\leq\mu_q;~0\leq\theta_{\mathrm{1}}\leq\cdots\leq\theta_p;~\mu_i+\theta_k\geq \alpha,\forall (i+k)\geq N_2+1\right\}. 
\end{aligned}
\end{gather}
\end{lemma}
\subsection{Proof of Lemma~\ref{doss}}\label{dos}
From Corollary~1, we can express the high SNR asymptotic sum-capacity as 
\begin{gather}
\begin{aligned}
 {C}_{\rm sum}(\mathcal{H})\doteq&\max_{\substack{0\leq\lambda\leq 1}}~ F(\mathcal{H},\lambda,\bar{\lambda})
\end{aligned}
\end{gather}
where 
\begin{gather}
\begin{aligned}
F(\mathcal{H},\lambda,\bar{\lambda})&=W_m\bigg(\mathrm{log}\left|I_{N_{\mathrm{dl}}}+\rho_{\mathrm{dl}}H_{\mathrm{dl}}H_{\mathrm{dl}}^\dagger+\bar{\lambda}\rho_{\mathrm{I}}H_{\mathrm{I}}H_{\mathrm{I}}^\dagger\right|+W\mathrm{log}\left|I_{N_{\mathrm{dl}}}+\frac{\lambda\rho_{\mathrm{S}}}{W}H_{\mathrm{S}}H_{\mathrm{S}}^{\dagger}\right|\\
&+\mathrm{log}\left|I_{N_{\mathrm{ul}}}+\bar{\lambda}\rho_{\mathrm{ul}}H_{\mathrm{ul}}(I_{M_{\mathrm{ul}}}+\bar{\lambda}\rho_{\mathrm{I}}H_{\mathrm{I}}^\dagger H_{\mathrm{I}})^{-1}H_{\mathrm{ul}}^\dagger\right|\bigg),\\
&=W_m\bigg(\mathrm{log}\left|I_{M_{\mathrm{dl}}}+\rho_{\mathrm{dl}}H_{\mathrm{dl}}^\dagger(I_{N_{\mathrm{dl}}}+\bar{\lambda}\rho_{\mathrm{I}}H_{\mathrm{I}} H_{\mathrm{I}}^\dagger)^{-1}H_{\mathrm{dl}}\right|+\mathrm{log}\left|I_{N_{\mathrm{dl}}}+\bar{\lambda}\rho_{\mathrm{I}}H_{\mathrm{I}}H_{\mathrm{I}}^\dagger\right|\\
&+\mathrm{log}\left|I_{N_{\mathrm{ul}}}+\bar{\lambda}\rho_{\mathrm{ul}}H_{\mathrm{ul}}(I_{M_{\mathrm{ul}}}+\bar{\lambda}\rho_{\mathrm{I}}H_{\mathrm{I}}^\dagger H_{\mathrm{I}})^{-1}H_{\mathrm{ul}}^\dagger\right|+W\mathrm{log}\left|I_{N_{\mathrm{dl}}}+\frac{\lambda\rho_{\mathrm{S}}}{W}H_{\mathrm{S}}H_{\mathrm{S}}^{\dagger}\right|\bigg).\nonumber
\end{aligned}
\end{gather}
The ordered eigenvalues of $G_1=H_{\mathrm{dl}}^\dagger(I_{N_{\mathrm{dl}}}+\bar{\lambda}\rho_{\mathrm{I}}H_{\mathrm{I}}H_{\mathrm{I}}^\dagger )^{-1}H_{\mathrm{dl}}, G_2=H_{\mathrm{ul}}(I_{M_{\mathrm{ul}}}+\bar{\lambda}\rho_{\mathrm{I}}H_{\mathrm{I}}^\dagger H_{\mathrm{I}})^{-1}H_{\mathrm{ul}}^\dagger$, $G_3=H_{\mathrm{I}}H_{\mathrm{I}}^\dagger$ and $G_4=H_{\mathrm{S}}H_{\mathrm{S}}^{\dagger}$ are denoted by $\beta_{\mathrm{1}}\geq\cdots\beta_{m_{\mathrm{dl}}}>0, \gamma_{\mathrm{1}}\geq\cdots\gamma_{m_{\mathrm{ul}}}>0, \eta_{\mathrm{1}}\geq\cdots\eta_{m_{\mathrm{I}}}>0$ and $\xi_{\mathrm{1}}\geq\cdots\xi_{m_{\mathrm{I}}}>0$. Let $\beta_i=\rho^{-\mu_i},i\in[1,m_{\mathrm{dl}}]; \gamma_j=\rho^{-\sigma_j}, j\in[1,m_{\mathrm{ul}}];\eta_k=\rho^{-\theta_k},k\in[0,m_{\mathrm{I}}];~\xi_l=\rho^{-\nu_l},l\in[0,m_{\mathrm{I}}].$ When $\rho\rightarrow\infty$, we have
\begin{gather}
\begin{aligned}
&\rho^{-d_{\mathfrak{B}_{\mathrm{sum}}}(r_{\mathrm{sum}})}\doteq\mathrm{Pr}\left(C_{\mathrm{sum}}<W_mr_{\mathrm{sum}}\mathrm{log}\rho\right)\\
&\doteq\mathrm{Pr}\bigg(\max_{0\leq \lambda\leq 1}~\prod_{i=1}^{m_{\mathrm{dl}}}(1+\rho^{\alpha_{\mathrm{dl}}}\beta_i)\prod_{j=1}^{m_{\mathrm{ul}}}(1+\bar{\lambda}\rho^{\alpha_{\mathrm{ul}}}\gamma_j)\prod_{k=1}^{m_{\mathrm{I}}}(1+\bar{\lambda}\rho^{\alpha_{\mathrm{I}}}\eta_k)
\left(\prod_{l=1}^{m_{\mathrm{I}}}(1+\frac{\lambda}{W}\rho^{\alpha_{\mathrm{S}}}\xi_l)\right)^W<\rho^{r_{\mathrm{sum}}}\bigg) \\
&\doteq\mathrm{Pr}\bigg(\max_{0\leq \lambda\leq 1}~(\bar{\lambda})^{m_{\mathrm{ul}}+m_{\mathrm{I}}}\left(\frac{\lambda}{W}\right)^{Wm_{\mathrm{I}}}\prod_{i=1}^{m_{\mathrm{dl}}}\rho^{(\alpha_{\mathrm{dl}}-\mu_i)^+}\prod_{j=1}^{m_{\mathrm{ul}}}\rho^{(\alpha_{\mathrm{ul}}-\sigma_j)^+}\prod_{k=1}^{m_{\mathrm{I}}}\rho^{(\alpha_{\mathrm{I}}-\theta_k)^+}
\prod_{l=1}^{m_{\mathrm{I}}}\rho^{W(\alpha_{\mathrm{S}}-\nu_l)^+}<\rho^{r_{\mathrm{sum}}}\bigg) \label{p0}
\end{aligned}
\end{gather}
where $m_{\mathrm{dl}},~m_{\mathrm{ul}}$ and $m_{\mathrm{I}}$ are defined in (\ref{gap}).

For any channel realization $\mathcal{H}$ in a particular fade period, we have $F(\mathcal{H}, \lambda=\bar{\lambda}=0.5)\leq {C}_{\rm sum}(\mathcal{H})< F(\mathcal{H}, \lambda=\bar{\lambda}=1)$, hence the sum-capacity outage event $\mathfrak{B}_{\rm sum}\triangleq\{R_{\mathrm{sum}}\notin {C}_{\rm sum}(\mathcal{H}) \}$ can be bounded as $\{R_{\mathrm{sum}}\notin F(\mathcal{H}, \lambda=\bar{\lambda}=1)\}\subset\mathfrak{B}_{\rm sum}\subseteq\{R_{\mathrm{sum}}\notin F(\mathcal{H}, \lambda=\bar{\lambda}=0.5)\}$. Consequently, we have 
\begin{gather}
\text{Pr}\left(R_{\mathrm{sum}}\notin F(\mathcal{H}, \lambda=\bar{\lambda}=1)\right)<\rho^{-d_{\mathfrak{B}_{\mathrm{sum}}}(r_{\mathrm{sum}})}\leq \text{Pr}\left(R_{\mathrm{sum}}\notin F(\mathcal{H}, \lambda=\bar{\lambda}=0.5)\right). \label{sumpr}
\end{gather}

From (\ref{sumpr}), we can see that when $\rho\rightarrow\infty$, $\rho^{-d_{\mathfrak{B}_{\mathrm{sum}}}(r_{\mathrm{sum}})}$
converges to the following result as $\frac{\lambda}{W},\bar{\lambda}$ do not grow at the same rate as $\rho$ thus can be ignored on the scale of interest
\begin{gather}
\begin{aligned}
\rho^{-d_{\mathfrak{B}_{\mathrm{sum}}}(r_{\mathrm{sum}})}&\doteq\mathrm{Pr}\left(\sum_{i=1}^{m_{\mathrm{dl}}}(\alpha_{\mathrm{dl}}-\mu_i)^++\sum_{j=1}^{m_{\mathrm{ul}}}(\alpha_{\mathrm{ul}}-\sigma_j)^++\sum_{k=1}^{m_{\mathrm{I}}}(\alpha_{\mathrm{I}}-\theta_k)^++W\sum_{l=1}^{m_{\mathrm{I}}}(\alpha_{\mathrm{S}}-\nu_l)^+
<r_{\mathrm{sum}}\right). \label{keydmt}
\end{aligned}
\end{gather}
Let $\bar{\mu}=\{\mu_{\mathrm{1}},\cdots,\mu_{m_{\mathrm{dl}}}\},~\bar{\sigma}=\{\sigma_{\mathrm{1}},\cdots,\sigma_{m_{\mathrm{ul}}}\},~\bar{\theta}=\{\theta_{\mathrm{1}},\cdots,\theta_{m_{\mathrm{I}}}\}$ and $\bar{\nu}=\{\nu_{\mathrm{1}},\cdots,\nu_{m_{\mathrm{I}}}\}$. The joint distribution of $\bar{\mu}, \bar{\sigma}, \bar{\theta}$ and $\bar{\nu}$ can be calculated as
\begin{eqnarray}
p(\bar{\mu},\bar{\sigma},\bar{\theta},\bar{\nu})&=&p(\bar{\mu},\bar{\sigma},\bar{\theta})p(\bar{\nu})\label{p1}\\
&=&p(\bar{\mu}\bar{\sigma}|\bar{\theta})p(\bar{\theta})p(\bar{\nu})\\
&=&p(\bar{\mu}|\bar{\theta})p(\bar{\sigma}|\bar{\theta})p(\bar{\theta})p(\bar{\nu}) \label{p2}
\end{eqnarray}
where (\ref{p1}) follows from the fact that matrix $G_4$ is independent of other matrices; (\ref{p2}) follows from random matrix theory that the dependence of $G_1$ and $G_2$ is only through the eigenvalues of $G_3$. Thus given the eigenvalues of $G_3$, the eigenvalues of $G_1$ and $G_2$ are conditionally independent.  

By invoking Lemma~\ref{lemma2cite} and Lemma~\ref{lemma3cite}, we have
\[p(\bar{\mu},\bar{\sigma},\bar{\theta},\bar{\nu})\doteq
  \begin{cases}
   \rho^{-E(\bar{\mu},\bar{\sigma},\bar{\theta},\bar{\nu})}&  ~~~\text{if}~~(\bar{\mu},\bar{\sigma},\bar{\theta},\bar{\nu})\in\mathcal{D} \\
   0  &~~~\text{Otherwise},
   \end{cases}
\]
where \begin{gather}
\begin{aligned}
E(\bar{\mu},\bar{\sigma},\bar{\theta},\bar{\nu})=&\left\{\sum_{i=1}^{m_{\mathrm{dl}}}(M_{\mathrm{dl}}+N_{\mathrm{dl}}+1-2i)\mu_i+\sum_{j=1}^{m_{\mathrm{ul}}}(M_{\mathrm{ul}}+N_{\mathrm{ul}}+1-2j)\sigma_j-(M_{\mathrm{dl}}+N_{\mathrm{ul}})m_{\mathrm{I}}\alpha_{\mathrm{I}}\right.\\
&\left.+\sum_{k=1}^{m_{\mathrm{I}}}(M_{\mathrm{dl}}+N_{\mathrm{ul}}+M_{\mathrm{ul}}+N_{\mathrm{dl}}+1-2k)\theta_k+\sum_{l=1}^{m_{\mathrm{I}}}(M_{\mathrm{ul}}+N_{\mathrm{dl}}+1-2l)\nu_l\right.\\
&\left. \sum_{i=1}^{m_{\mathrm{dl}}}\sum_{k=1}^{\min\{N_{\mathrm{dl}}-i,M_{\mathrm{ul}}\}}(\alpha_{\mathrm{I}}-\mu_i-\theta_k)^++\sum_{j=1}^{m_{\mathrm{ul}}}\sum_{k=1}^{\min\{M_{\mathrm{ul}}-j,N_{\mathrm{dl}}\}}(\alpha_{\mathrm{I}}-\sigma_j-\theta_k)^+\right\},
\end{aligned}
\end{gather}
\begin{gather}
\begin{aligned}
\mathcal{D}=&\Bigg\{\sum_{i=1}^{m_{\mathrm{dl}}}(\alpha_{\mathrm{dl}}-\mu_i)^++\sum_{j=1}^{m_{\mathrm{ul}}}(\alpha_{\mathrm{ul}}-\sigma_j)^++\sum_{k=1}^{m_{\mathrm{I}}}(\alpha_{\mathrm{I}}-\theta_k)^++W\sum_{l=1}^{m_{\mathrm{I}}}(\alpha_{\mathrm{S}}-\nu_l)^+
<r_{\mathrm{sum}};\\
&0\leq\mu_{\mathrm{1}}\leq\cdots\leq\mu_{m_{\mathrm{dl}}};~0\leq\sigma_{\mathrm{1}}\leq\cdots\leq\sigma_{m_{\mathrm{ul}}};~0\leq\theta_{\mathrm{1}}\leq\cdots\leq\theta_{m_{\mathrm{I}}};~0\leq\nu_{\mathrm{1}}\leq\cdots\leq\nu_{m_{\mathrm{I}}};\\
&\mu_i+\theta_k\geq\alpha_{\mathrm{I}},\forall (i+k)\geq N_{\mathrm{dl}}+1;~\sigma_j+\theta_k\geq\alpha_{\mathrm{I}},\forall (j+k)\geq M_{\mathrm{ul}}+1\Bigg\}. 
\end{aligned}
\end{gather}
With the joint distribution of $p(\bar{\mu},\bar{\sigma},\bar{\theta},\bar{\nu})$ we have obtained above, the outage probability is:
\begin{gather}
\begin{aligned}
\mathrm{Pr}(\mathfrak{B}_{\rm sum})&\doteq \int_{\mathcal{D}} p(\bar{\mu},\bar{\sigma},\bar{\theta},\bar{\nu})\doteq\int_{\mathcal{D}} \rho^{-E(\bar{\mu},\bar{\sigma},\bar{\theta},\bar{\nu})}.
\end{aligned}
\end{gather}
Using Laplace's principle, step (\ref{keydmt}) can be calculated by minimizing the SNR exponent $E(\bar{\mu},\bar{\sigma},\bar{\theta},\bar{\nu})$ which has the dominant probability.
Thus we have 
\begin{gather}
\begin{aligned}
d_{\mathfrak{B}_{\mathrm{sum}}}=\min_{(\bar{\mu},\bar{\sigma},\bar{\theta},\bar{\nu})\in\mathcal{D}} E(\bar{\mu},\bar{\sigma},\bar{\theta},\bar{\nu}),
\end{aligned}
\end{gather}
which proves Lemma~\ref{doss}.
\subsection{Proof of Lemma~\ref{doss2}}\label{dos2}
We first express the asymptotic achievable sum-rate (with $\lambda=\bar{\lambda}=0.5$) as follows
\begin{gather}
\begin{aligned}
I_{\mathrm{sum}}\doteq&W_m\bigg(\mathrm{log}\left|I_{M_{\mathrm{dl}}}+\rho_{\mathrm{dl}}H_{\mathrm{dl}}^\dagger(I_{N_{\mathrm{dl}}}+\bar{\lambda}\rho_{\mathrm{I}}H_{\mathrm{I}}H_{\mathrm{I}}^\dagger )^{-1}H_{\mathrm{dl}}\right|\\
&+\mathrm{log}\left|I_{N_{\mathrm{dl}}}+\bar{\lambda}\rho_{\mathrm{I}}H_{\mathrm{I}}H_{\mathrm{I}}^\dagger\right|+W\mathrm{log}\left|I_{N_{\mathrm{dl}}}+\frac{\lambda\rho_{\mathrm{S}}}{W}H_{\mathrm{S}}H_{\mathrm{S}}^{\dagger}\right|\bigg).
\end{aligned}
\end{gather}
We still use the same notations defined in Appendix~\ref{dos} to represent the ordered eigenvalue of $G_1=H_{\mathrm{dl}}^\dagger(I_{N_{\mathrm{dl}}}+\bar{\lambda}\rho_{\mathrm{I}}H_{\mathrm{I}}H_{\mathrm{I}}^\dagger )^{-1}H_{\mathrm{dl}}, G_3=H_{\mathrm{I}}H_{\mathrm{I}}^\dagger$ and $G_4=H_{\mathrm{S}}H_{\mathrm{S}}^{\dagger}$.
Thus we obtain that
\[\rho^{-d_{\mathfrak{O}_{\mathrm{sum}}}(r_{\mathrm{sum}})}\doteq\mathrm{Pr}\left(\sum_{i=1}^{m_{\mathrm{dl}}}(\alpha_{\mathrm{dl}}-\mu_i)^++\sum_{k=1}^{m_{\mathrm{I}}}(\alpha_{\mathrm{I}}-\theta_k)^++W\sum_{l=1}^{m_{\mathrm{I}}}(\alpha_{\mathrm{S}}-\nu_l)^+< r_{\mathrm{sum}}\right).\]
The joint distribution of $(\bar{\mu},\bar{\theta},\bar{\nu})$ can be derived by following the same steps in Appendix~\ref{dos}. Likewise, Lemma~\ref{doss2} can be proved and we omit the steps to avoid redundancy.

\subsection{DMT calculation of $(M,N_{\rm dl},M_{\rm ul},M)$ with and without CSIT}\label{caldmt}
\begin{lemma}\label{case3}
For the $(M,N_{\mathrm{dl}},M_{\mathrm{ul}},M)$ side-channel assisted full-duplex network with $ \alpha_{\mathrm{dl}}=\alpha_{\mathrm{ul}}=\alpha_{\mathrm{I}}=1$ and with CSIT, the optimal DMT at multiplexing gain pair $(r_{\mathrm{dl}},r_{\mathrm{ul}})$ is
 \begin{equation}
d^{\text{CSIT,opt}}_{(M,N_{\mathrm{dl}},M_{\mathrm{ul}},M)}(r_{\mathrm{dl}},r_{\mathrm{ul}})=\min\{d_{M,N_{\mathrm{dl}}}(r_{\mathrm{dl}}), d_{M_{\mathrm{ul}},M}(r_{\mathrm{ul}}),d_{{\mathrm{sum}}(M,N_{\mathrm{dl}},M_{\mathrm{ul}},M)}^{\text{CSIT}}(r_{\mathrm{sum}})\}.
\end{equation}
where $d_{{\mathrm{sum}}(M,N_{\mathrm{dl}},M_{\mathrm{ul}},M)}^{\text{CSIT}}(r_{\mathrm{sum}})$ is given as:\begin{itemize}
\item when $W\leq\frac{\left|M_{\mathrm{ul}}-N_{\mathrm{dl}}\right|+1}{2M+M_{\mathrm{ul}}+N_{\mathrm{dl}}-1}$,
\[\!\!\!\!\!\!\!\!\!\!\!\!\!\!d_{{\mathrm{sum}}(M,N_{\mathrm{dl}},M_{\mathrm{ul}},M)}^{\text{CSIT}}(r_{\mathrm{sum}})\!\!=\!\!\left\{
  \begin{array}{l l}
  \!\!   \alpha_{\mathrm{S}} d_{M_{\mathrm{ul}},N_{\mathrm{dl}}}\left(\frac{r_{\mathrm{sum}}}{W\alpha_{\mathrm{S}}}\right)\!+\!M_{\mathrm{ul}}N_{\mathrm{dl}}\!+\!M(M_{\mathrm{ul}}\!+\!N_{\mathrm{dl}}),~r_{\mathrm{sum}}\! \leq \!m_{\mathrm{I}}W\alpha_{\mathrm{S}}\\
   \!\!  d_{m_{\mathrm{I}},2M+m_X}\left(r_{\mathrm{sum}}\!-\!m_{\mathrm{I}}W\alpha_{\mathrm{S}}\right)\!+\!M|M_{\mathrm{ul}}\!-\!N_{\mathrm{dl}}|, ~ m_{\mathrm{I}}W\alpha_{\mathrm{S}}\!\leq\! r_{\mathrm{sum}}\!\leq\! m_{\mathrm{I}}(1\!+\!W\alpha_{\mathrm{S}})\\
  \!\!   d_{|M_{\mathrm{ul}}\!-\!N_{\mathrm{dl}}|,M}\left(r_{\mathrm{sum}}\!-\!m_{\mathrm{I}}(1\!+\!W\alpha_{\mathrm{S}})\right), ~ m_{\mathrm{I}}(1\!\!+\!\!W\alpha_{\mathrm{S}})\!\leq \!r_{\mathrm{sum}}\!\leq \! m_X\!+\!m_{\mathrm{I}}W\alpha_{\mathrm{S}} \end{array} \right.
\]
\item when $W\in\left[\frac{M_{\mathrm{ul}}+N_{\mathrm{dl}}-1}{2M+|M_{\mathrm{ul}}-N_{\mathrm{dl}}|+1},\frac{|M_{\mathrm{ul}}-N_{\mathrm{dl}}|+1}{M+|M_{\rm ul}-N_{\rm dl}|-1}\right]$,
\[\!\!\!\!\!\!\!\!\!\!\!\!\!\!d_{{\mathrm{sum}}(M,N_{\mathrm{dl}},M_{\mathrm{ul}},M)}^{\text{CSIT}}(r_{\mathrm{sum}})\!\!=\!\!\left\{
  \begin{array}{l l}
\!\!   d_{m_{\mathrm{I}},2M+m_X}\left(r_{\mathrm{sum}}\right)\!+\!M_{\mathrm{ul}}N_{\mathrm{dl}}\alpha_{\mathrm{S}}\!+\!M|M_{\mathrm{ul}}\!-\!N_{\mathrm{dl}}|, ~ r_{\mathrm{sum}}\!\leq \!m_{\mathrm{I}}\\
 \!\!     \alpha_{\mathrm{S}} d_{M_{\mathrm{ul}},N_{\mathrm{dl}}}\left(\frac{r_{\mathrm{sum}}-m_{\mathrm{I}}}{W\alpha_{\mathrm{S}}}\right)\!+\!M|M_{\mathrm{ul}}\!-\!N_{\mathrm{dl}}|, ~ m_{\mathrm{I}}\!\leq\! r_{\mathrm{sum}} \!\leq\! m_{\mathrm{I}}(1\!+\!W\alpha_{\mathrm{S}})\\
  \!\!   d_{|M_{\mathrm{ul}}\!-\!N_{\mathrm{dl}}|,M}\left(r_{\mathrm{sum}}\!-\!m_{\mathrm{I}}(1\!+\!W\alpha_{\mathrm{S}})\right), ~ m_{\mathrm{I}}(1\!\!+\!\!W\alpha_{\mathrm{S}})\!\leq \!r_{\mathrm{sum}}\!\leq \! m_X\!+\!m_{\mathrm{I}}W\alpha_{\mathrm{S}}
   \end{array} \right.
\]
\item when $W\geq \frac{M_{\mathrm{ul}}+N_{\mathrm{dl}}-1}{M-|M_{\rm ul}-N_{\rm dl}|+1}$,
\[\!\!\!\!\!\!\!\!\!\!\!\!\!\!d_{{\mathrm{sum}}(M,1,1,M)}^{\text{CSIT}}(r_{\mathrm{sum}})\!\!=\!\!\left\{
  \begin{array}{l l}
 \!\!  d_{m_{\mathrm{I}},2M+m_X}\left(r_{\mathrm{sum}}\right)\!\!+\!\!M_{\mathrm{ul}}N_{\mathrm{dl}}\alpha_{\mathrm{S}}\!+\!M|M_{\mathrm{ul}}\!-\!N_{\mathrm{dl}}|, ~  r_{\mathrm{sum}}\!\leq\! m_{\mathrm{I}}\\
  \!\!   d_{|M_{\mathrm{ul}}-N_{\mathrm{dl}}|,M}\left(r_{\mathrm{sum}}\!-\!m_{\mathrm{I}}\right)\!+\!M_{\mathrm{ul}}N_{\mathrm{dl}}\alpha_{\mathrm{S}}, ~ m_{\mathrm{I}}\!\leq \!r_{\mathrm{sum}}\!\leq \!m_X\\
 \!\!    \alpha_{\mathrm{S}} d_{M_{\mathrm{ul}},N_{\mathrm{dl}}}\left(\frac{r_{\mathrm{sum}}-m_X}{W\alpha_{\mathrm{S}}}\right),~ m_X\!\leq \!r_{\mathrm{sum}}\! \leq\! m_X\!+\!m_{\mathrm{I}}W\alpha_{\mathrm{S}}
   \end{array} \right.
\]
\end{itemize}
where $m_{\mathrm{I}}=\min\{M_{\mathrm{ul}},N_{\mathrm{dl}}\}, m_X=\max\{M_{\mathrm{ul}},N_{\mathrm{dl}}\}.$
\end{lemma}
\begin{proof}
The details of the proof can be found in Appendix~\ref{dos4}.
\end{proof}

The achievable DMT of $(M,N_{\mathrm{dl}},M_{\mathrm{ul}},M)$ without CSIT is given in the following lemma.
\begin{lemma} \label{case33}
Consider the case in Lemma~\ref{case3} under no-CSIT assumption, the achievable DMT at multiplexing gain pair $(r_{\mathrm{dl}},r_{\mathrm{ul}})$ is
 \[
d_{(M,N_{\mathrm{dl}},M_{\mathrm{ul}},M)}^{\text{No-CSIT}}(r_{\mathrm{dl}},r_{\mathrm{ul}})=\min\{d_{M,N_{\mathrm{dl}}}(r_{\rm dl}), d_{M_{\mathrm{ul}},M}(r_{\mathrm{ul}}),d_{{\mathrm{sum}}(M,N_{\mathrm{dl}},M_{\mathrm{ul}},M)}^{\text{No-CSIT}}(r_{\mathrm{sum}})\}.
\]
where $d_{{\mathrm{sum}}(M,N_{\mathrm{dl}},M_{\mathrm{ul}},M)}^{\text{No-CSIT}}(r_{\mathrm{sum}})$ is given:
if $M_{\mathrm{ul}}\geq 2(N_{\mathrm{dl}}-1)$,
\begin{itemize}
\item when $W\leq \frac{M_{\mathrm{ul}}-N_{\mathrm{dl}}+1}{M+M_{\mathrm{ul}}+N_{\mathrm{dl}}-1}$,
\[\!\!\!\!d_{{\mathrm{sum}}(M,N_{\mathrm{dl}},M_{\mathrm{ul}},M)}^{\text{CSIT}}(r_{\mathrm{sum}})\!\!=\!\!\left\{
  \begin{array}{l l}
    \!\! \alpha_{\mathrm{S}} d_{M_{\mathrm{ul}},N_{\mathrm{dl}}}\left(\frac{r_{\mathrm{sum}}}{W\alpha_{\mathrm{S}}}\right)\!+\!N_{\mathrm{dl}}(M_{\mathrm{ul}}\!+\!M), ~ r_{\mathrm{sum}} \!\leq \!N_{\mathrm{dl}}W\alpha_{\mathrm{S}}\\
     \!\! d_{N_{\mathrm{dl}},M+M_{\mathrm{ul}}}\left(r_{\mathrm{sum}}\!-\!N_{\mathrm{dl}}W\alpha_{\mathrm{S}}\right), ~  N_{\mathrm{dl}}W\alpha_{\mathrm{S}}\!\leq\! r_{\mathrm{sum}}\leq N_{\mathrm{dl}}(1+W\alpha_{\mathrm{S}})
   \end{array} \right.
\]
\item when $W\geq \frac{M_{\mathrm{ul}}+N_{\mathrm{dl}}-1}{M+M_{\mathrm{ul}}-N_{\mathrm{dl}}+1}$,
\[\!\!\!\!d_{{\mathrm{sum}}(M,N_{\mathrm{dl}},M_{\mathrm{ul}},M)}^{\text{No-CSIT}}(r_{\mathrm{sum}})\!\!=\!\!\left\{
  \begin{array}{l l}
\!\!     d_{N_{\mathrm{dl}},M+M_{\mathrm{ul}}}\left(r_{\mathrm{sum}}\right)+M_{\mathrm{ul}}N_{\mathrm{dl}}\alpha_{\mathrm{S}}, ~  r_{\mathrm{sum}}\!\leq\! N_{\mathrm{dl}}\\
    \!\!  \alpha_{\mathrm{S}}d_{M_{\mathrm{ul}},N_{\mathrm{dl}}}\left(\frac{r_{\mathrm{sum}}-N_{\mathrm{dl}}}{W\alpha_{\mathrm{S}}}\right), ~ N_{\mathrm{dl}}\leq r_{\mathrm{sum}}\! \leq \!N_{\mathrm{dl}}(1+W\alpha_{\mathrm{S}})
   \end{array} \right.
\]
\end{itemize}
And if $N_{\mathrm{dl}}\geq M_{\mathrm{ul}}$, and $M_{\rm ul}\leq 2$:
\begin{itemize}
\item when $W\leq\frac{N_{\mathrm{dl}}-M_{\mathrm{ul}}+1}{M+M_{\mathrm{ul}}+N_{\mathrm{dl}}-1}$,
\[\!\!\!\!\!\!\!\!\!\!\!\!\!\!\!d_{{\mathrm{sum}}(M,N_{\mathrm{dl}},M_{\mathrm{ul}},M)}^{\text{No-CSIT}}(r_{\mathrm{sum}})\!\!=\!\!\left\{
  \begin{array}{l l}
\!\!     \alpha_{\mathrm{S}} d_{M_{\mathrm{ul}},N_{\mathrm{dl}}}\left(\frac{r_{\mathrm{sum}}}{W\alpha_{\mathrm{S}}}\right)\!+\!N_{\mathrm{dl}}(M_{\mathrm{ul}}\!+\!M), ~ r_{\mathrm{sum}} \!\leq \!M_{\mathrm{ul}}W\alpha_{\mathrm{S}}\\
  \!\!   d_{M_{\mathrm{ul}},M+N_{\mathrm{dl}}}\left(r_{\mathrm{sum}}\!-\!M_{\mathrm{ul}}W\alpha_{\mathrm{S}}\right)\!+\!M(N_{\mathrm{dl}}\!-\!M_{\mathrm{ul}}), M_{\mathrm{ul}}W\alpha_{\mathrm{S}}\!\leq \!r_{\mathrm{sum}}\!\leq \!M_{\mathrm{ul}}(1\!+\!W\alpha_{\mathrm{S}})\\
  \!\!   d_{N_{\mathrm{dl}}-M_{\mathrm{ul}},M}\left(r_{\mathrm{sum}}\!-\!M_{\mathrm{ul}}(1\!+\!W\alpha_{\mathrm{S}})\right),~ M_{\mathrm{ul}}(1\!+\!W\alpha_{\mathrm{S}})\!\leq\! r_{\mathrm{sum}}\!\leq\! N_{\mathrm{dl}}\!+\!M_{\mathrm{ul}}W\alpha_{\mathrm{S}}
   \end{array} \right.
\]
\item when $W\in\left[\frac{M_{\mathrm{ul}}+N_{\mathrm{dl}}-1}{M+N_{\mathrm{dl}}-M_{\mathrm{ul}}+1},\frac{N_{\mathrm{dl}}-M_{\mathrm{ul}}+1}{M+N_{\rm dl}-M_{\mathrm{ul}}-1}\right]$,
\[\!\!\!\!\!\!\!\!\!\!\!\!\!\!\!d_{{\mathrm{sum}}(M,N_{\mathrm{dl}},M_{\mathrm{ul}},M)}^{\text{No-CSIT}}(r_{\mathrm{sum}})\!\!=\!\!\left\{
  \begin{array}{l l}
  \!\! d_{M_{\mathrm{ul}},M+N_{\mathrm{dl}}}\left(r_{\mathrm{sum}}\right)\!+\!M_{\mathrm{ul}}N_{\mathrm{dl}}\alpha_{\mathrm{S}}\!+\!M(N_{\mathrm{dl}}\!-\!M_{\mathrm{ul}}), ~ r_{\mathrm{sum}}\!\leq\! M_{\mathrm{ul}}\\
   \!\!   \alpha_{\mathrm{S}} d_{M_{\mathrm{ul}},N_{\mathrm{dl}}}\left(\frac{r_{\mathrm{sum}}-M_{\mathrm{ul}}}{W\alpha_{\mathrm{S}}}\right)\!+\!M(N_{\mathrm{dl}}\!-\!M_{\mathrm{ul}}), ~ M_{\mathrm{ul}}\!\leq \!r_{\mathrm{sum}}\! \leq\! M_{\mathrm{ul}}(1\!+\!W\alpha_{\mathrm{S}})\\
 \!\!     d_{N_{\mathrm{dl}}-M_{\mathrm{ul}},M}\left(r_{\mathrm{sum}}\!-\!M_{\mathrm{ul}}(1\!+\!W\alpha_{\mathrm{S}})\right), ~  M_{\mathrm{ul}}(1\!+\!W\alpha_{\mathrm{S}})\!\leq\! r_{\mathrm{sum}}\!\leq \!N_{\mathrm{dl}}\!+\!M_{\mathrm{ul}}W\alpha_{\mathrm{S}}
   \end{array} \right.
\]
\item when $W\geq \frac{M_{\mathrm{ul}}+N_{\mathrm{dl}}-1}{M-N_{\mathrm{dl}}+M_{\rm ul}+1}$,
\[\!\!\!\!\!\!\!\!\!\!\!\!\!\!\!d_{{\mathrm{sum}}(M,N_{\mathrm{dl}},M_{\mathrm{ul}},M)}^{\text{CSIT}}(r_{\mathrm{sum}})=\left\{
  \begin{array}{l l}
   d_{M_{\mathrm{ul}},M+N_{\mathrm{dl}}}\left(r_{\mathrm{sum}}\right)+M_{\mathrm{ul}}N_{\mathrm{dl}}\alpha_{\mathrm{S}}+M(N_{\mathrm{dl}}-M_{\mathrm{ul}}), ~ r_{\mathrm{sum}}\leq M_{\mathrm{ul}}\\
    d_{N_{\mathrm{dl}}-M_{\mathrm{ul}},M}\left(r_{\mathrm{sum}}-M_{\mathrm{ul}})\right)+M_{\mathrm{ul}}N_{\mathrm{dl}}\alpha_{\mathrm{S}}, ~  M_{\mathrm{ul}}\leq r_{\mathrm{sum}}\leq N_{\mathrm{dl}}\\
      \alpha_{\mathrm{S}} d_{M_{\mathrm{ul}},N_{\mathrm{dl}}}\left(\frac{r_{\mathrm{sum}}-N_{\mathrm{dl}}}{W\alpha_{\mathrm{S}}}\right), ~N_{\mathrm{dl}}\leq r_{\mathrm{sum}}\leq N_{\mathrm{dl}}+M_{\mathrm{ul}}W\alpha_{\mathrm{S}}
   \end{array} \right.
\]
\end{itemize}
\end{lemma}
\begin{proof}
The results can be derived by following the similar steps in the proof of Lemma~\ref{case3}.
\end{proof}

\subsection{Proof of Lemma~\ref{case3}}\label{dos4}
As demonstrated in the proof of Corollary~\ref{coroex1}, we use gradient descent method to find the local optimum value  for each value of the multiplexing gain which is equivalent to global optimum in the convex optimization problem. This method is also used in \cite{DMTMIMOZ} to derive the DMT for MIMO Z-interference channel for some special cases. In our setting of $(M,N_{\mathrm{dl}},M_{\mathrm{ul}},M)$ with $\alpha_{\mathrm{dl}}=\alpha_{\mathrm{ul}}=\alpha_{\mathrm{I}}=1$ and $r_{\mathrm{dl}}=r_{\mathrm{ul}}=r$, we can simplify the objective function in Lemma~\ref{doss} given sum multiplexing gain. We will first give the analysis when $M_{\mathrm{ul}}\geq N_{\mathrm{dl}}$. By substituting $\nu_l^\prime=W\nu_l$ in (\ref{disum}), we can express the objective function as
\begin{gather}
\begin{aligned}
 d^{\text{CSIT}}_{\mathrm{sum}}=&\min \sum_{k=1}^{N_{\mathrm{dl}}}(2M+M_{\mathrm{ul}}+N_{\mathrm{dl}}+1-2k)\theta_k+\frac{1}{W}\sum_{l=1}^{N_{\mathrm{dl}}}(M_{\mathrm{ul}}+N_{\mathrm{dl}}+1-2l)\nu_l^\prime \\ &+\sum_{i=1}^{N_{\mathrm{dl}}}(M+N_{\mathrm{dl}}+1-2i)\mu_i+\sum_{j=1}^{M_{\mathrm{ul}}}(M+M_{\mathrm{ul}}+1-2j)\sigma_j-2MN_{\mathrm{dl}}\\
&+\sum_{i=1}^{N_{\mathrm{dl}}}\sum_{k=1}^{N_{\mathrm{dl}}-i}(1-\mu_i-\theta_k)^++\sum_{j=1}^{M_{\mathrm{ul}}}\sum_{k=1}^{\min\{M_{\mathrm{ul}}-j,N_{\mathrm{dl}}\}}(1-\sigma_j-\theta_k)^+,\\\label{dmt1}
\mathrm{Subject~to}\quad&\sum_{i=1}^{N_{\mathrm{dl}}}(1-\mu_i)^++\sum_{j=1}^{M_{\mathrm{ul}}}(1-\sigma_j)^++\sum_{k=1}^{N_{\mathrm{dl}}}(1-\theta_k)^++\sum_{l=1}^{N_{\mathrm{dl}}}(W\alpha_{\mathrm{S}}-\nu_l^\prime)^+< r_{\mathrm{sum}};\\
&0\leq\mu_{\mathrm{1}}\leq\cdots\leq\mu_{N_{\mathrm{dl}}};~0\leq\sigma_{\mathrm{1}}\leq\cdots\leq\sigma_{M_{\mathrm{ul}}};~0\leq\theta_{\mathrm{1}}\leq\cdots\leq\theta_{N_{\mathrm{1}}};~0\leq\nu_{\mathrm{1}}^\prime\leq\cdots\leq\nu_{N_{\mathrm{dl}}}^\prime;\\
&\mu_i+\theta_k\geq1,~\forall (i+k)\geq N_{\mathrm{dl}}+1;\\
&\sigma_j+\theta_k\geq1,~\forall (j+k)\geq M_{\mathrm{ul}}+1.
\end{aligned}
\end{gather}
Next, we differentiate the objective function in (\ref{dmt1}) with respect to different variables,
\begin{eqnarray}
&&\frac{ \partial  d^{\text{CSIT}}_{\mathrm{sum}}}{\partial \nu_l^\prime }=\frac{1}{W}(M_{\mathrm{ul}}+N_{\mathrm{dl}}+1-2l),~ l\leq N_{\mathrm{dl}};\\
&&\frac{ \partial  d^{\text{CSIT}}_{\mathrm{sum}}}{\partial \theta_k }\bigg|_{\mu_i=\sigma_j=1,\forall i,j}=2M+M_{\mathrm{ul}}+N_{\mathrm{dl}}+1-2k,~ k\leq N_{\mathrm{dl}};\\
&&\frac{ \partial  d^{\text{CSIT}}_{\mathrm{sum}}}{\partial \mu_1 }\bigg|_{\theta_k=1, \forall k}=M+N_{\mathrm{dl}}-1<\frac{ \partial  d^{\text{CSIT}}_{\mathrm{sum}}}{\partial \theta_k },\forall k;\\
&&\frac{ \partial  d^{\text{CSIT}}_{\mathrm{sum}}}{\partial \sigma_1 }\bigg|_{\theta_k=1, \forall k}=M+M_{\mathrm{ul}}-1< \frac{ \partial  d^{\text{CSIT}}_{\mathrm{sum}}}{\partial \theta_k }, \forall k.
\end{eqnarray}
Since the slope of the objective function decreases with the increasing index of $\mu_i,\sigma_j$, it suffices to only consider the decay of the function with $\mu_{\mathrm{1}},\sigma_{\mathrm{1}}$.
We can also easily verify that the decay slopes of $\mu_{\mathrm{1}}$ and $\sigma_{\mathrm{1}}$ are smaller than that of $\theta_k, \forall k$. 
\subsubsection{Case 1}\label{case1}
In this case, $\nu_l^\prime$ has the steepest descent, i.e., $\frac{ \partial  d^{\text{CSIT}}_{\mathrm{sum}}}{\partial \nu_{N_{\rm dl}}^\prime }\geq \frac{ \partial  d^{\text{CSIT}}_{\mathrm{sum}}}{\partial \theta_1 }$.
Thus when $W\leq\frac{M_{\mathrm{ul}}-N_{\mathrm{dl}}+1}{2M+M_{\mathrm{ul}}+N_{\rm dl}-1}$,  for $(l-1)W\alpha_{\mathrm{S}}\leq r_{\mathrm{sum}}\leq lW\alpha_{\mathrm{S}}~,\forall l$, the steepest descent of the objective function is along the decreasing value of $\nu_l^\prime$ with $\mu_i=\sigma_j=\theta_k=1,\forall i,j,k$. Now the optimization problem becomes
\begin{gather}
\begin{aligned}
d^{\text{CSIT}}_{\mathrm{sum}}=&\min \frac{1}{W}\sum_{l=1}^{N_{\mathrm{dl}}}(M_{\mathrm{ul}}+N_{\mathrm{dl}}+1-2l)\nu_l^\prime+M_{\mathrm{ul}}N_{\mathrm{dl}}+M(M_{\mathrm{ul}}+N_{\mathrm{dl}}),\\
\mathrm{Subject~to}\quad&\sum_{l=1}^{N_{\mathrm{dl}}}(W\alpha_{\mathrm{S}}-\nu_l^\prime)^+\leq r_{\mathrm{sum}};\\
&0\leq\nu_{\mathrm{1}}^\prime\cdots\leq\nu_{N_{\mathrm{dl}}}^\prime.\nonumber
\end{aligned}
\end{gather}
Invoking Lemma~\ref{lemmacite}, the solution to the optimization problem above is
\[d_{\mathfrak{B}_{\mathrm{sum}}}^{\text{CSIT}}=\alpha_{\mathrm{S}} d_{M_{\mathrm{ul}},N_{\mathrm{dl}}}\left(\frac{r_{\mathrm{sum}}}{W\alpha_{\mathrm{S}} }\right)+M_{\mathrm{ul}}N_{\mathrm{dl}}+M(M_{\mathrm{ul}}+N_{\mathrm{dl}}),~\forall r_{\mathrm{sum}}\leq N_{\mathrm{dl}}W\alpha_{\mathrm{S}}.\]
If $r_{\mathrm{sum}}\geq N_{\mathrm{dl}}W\alpha_{\mathrm{S}}$, it can be implied from the solution above that $\nu_l^\prime=0,~\forall l$ are in the optimal solution. We can see that now the steepest descent of the objective function in (\ref{dmt1}) is along the decreasing value of $\theta_k$ with $\mu_i=\sigma_j=1,\forall i,j$, and the corresponding optimization function becomes
\begin{gather}
\begin{aligned}
d^{\text{CSIT}}_{\mathrm{sum}}=&\min \sum_{k=1}^{N_{\mathrm{dl}}}(2M+M_{\mathrm{ul}}+N_{\mathrm{dl}}+1-2k)\theta_k+MM_{\mathrm{ul}}-MN_{\mathrm{dl}}\\
\mathrm{Subject~to}\quad&\sum_{k=1}^{N_{\mathrm{dl}}}(1-\theta_k)^+\leq r_{\mathrm{sum}}-N_{\mathrm{dl}}W\alpha_{\mathrm{S}};\\
&0\leq\theta_{\mathrm{1}}\leq\cdots\leq\theta_{N_{\mathrm{dl}}}.
\end{aligned}
\end{gather}
Again, invoking Lemma~\ref{lemmacite}, we have
\[d_{\mathfrak{B}_{\mathrm{sum}}}^{\text{CSIT}}=d_{N_{\mathrm{dl}},2M+M_{\mathrm{ul}}}(r_{\mathrm{sum}}-N_{\mathrm{dl}}W\alpha_{\mathrm{S}})+M(M_{\mathrm{ul}}-N_{\mathrm{dl}}),~ N_{\mathrm{dl}}W\alpha_{\mathrm{S}}\leq r_{\mathrm{sum}}\leq N_{\mathrm{dl}}W\alpha_{\mathrm{S}}+N_{\mathrm{dl}}.\]
Likewise, when $r_s\geq N_{\mathrm{dl}}W\alpha_{\mathrm{S}}+N_{\mathrm{dl}}$, $\theta_k=0~\forall k$, the optimization problem is given as
\begin{gather}
\begin{aligned}
d^{\text{CSIT}}_{\mathrm{sum}}=&\min \sum_{i=1}^{N_{\mathrm{dl}}}(M+N_{\mathrm{dl}}+1-2i)\mu_i+\sum_{j=1}^{M_{\mathrm{ul}}}(M+M_{\mathrm{ul}}+1-2j)\sigma_j\\
&-2MN_{\mathrm{dl}}+\sum_{i=1}^{N_{\mathrm{dl}}}\sum_{k=1}^{N_{\mathrm{dl}}-i}(1-\mu_i)^++\sum_{j=1}^{M_{\mathrm{ul}}}\sum_{k=1}^{\min\{M_{\mathrm{ul}}-j,N_{\mathrm{dl}}\}}(1-\sigma_j)^+;\\
\mathrm{Subject~to}\quad&\sum_{i=1}^{N_{\mathrm{dl}}}(1-\mu_i)^++\sum_{j=1}^{M_{\mathrm{ul}}}(1-\sigma_j)^+\leq r_{\mathrm{sum}}-N_{\mathrm{dl}}W\alpha_{\mathrm{S}}-N_{\mathrm{dl}};\\
&0\leq\mu_{\mathrm{1}}\leq\cdots\leq\mu_{N_{\mathrm{dl}}};~0\leq\sigma_{\mathrm{1}}\leq\cdots\leq\sigma_{M_{\mathrm{ul}}};\\
&\mu_i\geq1,~\forall i+k\geq N_{\mathrm{dl}}+1,\forall k\\
&\sigma_j\geq 1,\forall j+k\geq M_{\mathrm{ul}}+1,\forall k.
\label{obj1}
\end{aligned}
\end{gather}
Apparently, to minimize the objective function above, we should let $\mu_i=1, \forall i$ and $\sigma_j=1, \forall j\geq M_{\mathrm{ul}}-N_{\mathrm{dl}}+1$. Hence the last term in (\ref{obj1}) can be rewritten as
\begin{gather}
\begin{aligned}
\sum_{j=1}^{M_{\mathrm{ul}}}\sum_{k=1}^{\min\{M_{\mathrm{ul}}-j,N_{\mathrm{dl}}\}}(1-\sigma_j)^+&=\sum_{j=1}^{M_{\mathrm{ul}}-N_{\mathrm{dl}}}\min\{M_{\mathrm{ul}}-j,N_{\mathrm{dl}}\}(1-\sigma_j)^+\\
&= \sum_{j=1}^{M_{\mathrm{ul}}-N_{\mathrm{dl}}}N_{\mathrm{dl}}(1-\sigma_j)^+.\nonumber
\end{aligned}
\end{gather}
Combining the results above, the objective function in (\ref{obj1}) reduces to
\begin{gather}
\begin{aligned}
d^{\text{CSIT}}_{\mathrm{sum}}&=\min \sum_{j=1}^{M_{\mathrm{ul}}-N_{\mathrm{dl}}}(M+M_{\mathrm{ul}}+1-2j)\sigma_j+N_{\rm dl}(N_{\mathrm{dl}}-M_{\rm ul})+\sum_{j=1}^{M_{\mathrm{ul}}-N_{\mathrm{dl}}}N_{\mathrm{dl}}(1-\sigma_j)^+\\
&=\sum_{j=1}^{M_{\mathrm{ul}}-N_{\mathrm{dl}}}(M+M_{\mathrm{ul}}-N_{\mathrm{dl}}+1-2j)\sigma_j\\
\mathrm{Subject~to}\quad&\sum_{j=1}^{M_{\mathrm{ul}}-N_{\mathrm{dl}}}(1-\sigma_j)^+\leq r_{\mathrm{sum}}-N_{\mathrm{dl}}W\alpha_{\mathrm{S}}-N_{\mathrm{dl}},\\
&0\leq\sigma_{\mathrm{1}}\leq\cdots\leq\sigma_{M_{\mathrm{ul}}-N_{\mathrm{dl}}}
\end{aligned}
\end{gather}
Thus the optimization problem above has the following solution 
\[d_{\mathfrak{B}_{\mathrm{sum}}}^{\text{CSIT}}= d_{M_{\mathrm{ul}}-N_{\mathrm{dl}},M}\left(r_{\mathrm{sum}}-N_{\mathrm{dl}}(W\alpha_{\mathrm{S}}+1)\right),~N_{\mathrm{dl}}(W\alpha_{\mathrm{S}}+1)\leq r_{\mathrm{sum}}\leq N_{\mathrm{dl}}W\alpha_{\mathrm{S}}+M_{\mathrm{ul}}.\]
\subsubsection{Case 2}
In this case, $\theta_k$ has the steepest descent, i.e., $\frac{ \partial  d^{\text{CSIT}}_{\mathrm{sum}}}{\partial \theta_{N_{\rm dl}}}\geq \frac{ \partial  d^{\text{CSIT}}_{\mathrm{sum}}}{\partial \nu_1^\prime }$.
Thus when $W\geq \frac{M_{\mathrm{ul}}+N_{\mathrm{dl}}-1}{2M+M_{\mathrm{ul}}-N_{\rm dl}+1}$, for $k-1\leq r_{\mathrm{sum}}\leq k$, the objective function in (\ref{dmt1}) decays fastest first along the decreasing values of $\theta_k$ with $\mu_i=\sigma_j=1,\nu_{l}^\prime=W\alpha_{\mathrm{S}},\forall i,j,l$.  The optimization problem becomes
\begin{gather}
\begin{aligned}
d^{\text{CSIT}}_{\mathrm{sum}}=&\min \sum_{k=1}^{N_{\mathrm{dl}}}(2M+M_{\mathrm{ul}}+N_{\mathrm{dl}}+1-2k)\theta_k+M_{\mathrm{ul}}N_{\mathrm{dl}}\alpha_{\rm S}+M(M_{\mathrm{ul}}-N_{\mathrm{dl}}),\\
\mathrm{Subject~to}\quad&\sum_{k=1}^{N_{\mathrm{dl}}}(1-\theta_k)^+\leq r_{\mathrm{sum}},\\
&0\leq\theta_{\mathrm{1}}\leq\cdots\leq\theta_{N_{\mathrm{dl}}}.
\end{aligned}
\end{gather}
Invoking Lemma~\ref{lemmacite}, the solution to the  optimization problem above is
\[d_{\mathfrak{B}_{\mathrm{sum}}}^{\text{CSIT}}=d_{N_{\mathrm{dl}},2M+M_{\mathrm{ul}}}(r_{\mathrm{sum}})+M_{\mathrm{ul}}N_{\mathrm{dl}}\alpha_{\rm S}+M(M_{\mathrm{ul}}-N_{\mathrm{dl}}),~\forall r_{\mathrm{sum}}\leq N_{\mathrm{dl}}.\]

If $r_{\mathrm{sum}}\geq N_{\mathrm{dl}}$, the optimal solution has $\theta_k=0~\forall k$. We rewrite the objective function as
\begin{gather}
\begin{aligned}
 d^{\text{CSIT}}_{\mathrm{sum}}=&\min \sum_{i=1}^{N_{\mathrm{dl}}}(M+N_{\mathrm{dl}}+1-2i)\mu_i+\sum_{j=1}^{M_{\mathrm{ul}}}(M+M_{\mathrm{ul}}+1-2j)\sigma_j-2MN_{\mathrm{dl}}\\
&+\frac{1}{W}\sum_{l=1}^{N_{\mathrm{dl}}}(M_{\mathrm{ul}}+N_{\mathrm{dl}}+1-2l)\nu_l^\prime +\sum_{i=1}^{N_{\mathrm{dl}}}\sum_{k=1}^{N_{\mathrm{dl}}-i}(1-\mu_i)^++\sum_{j=1}^{M_{\mathrm{ul}}}\sum_{k=1}^{\min\{M_{\mathrm{ul}}-j,N_{\mathrm{dl}}\}}(1-\sigma_j)^+,\\\label{obj2}
\mathrm{Subject~to}\quad&\sum_{i=1}^{N_{\mathrm{dl}}}(1-\mu_i)^++\sum_{j=1}^{M_{\mathrm{ul}}}(1-\sigma_j)^++\sum_{l=1}^{N_{\mathrm{dl}}}(W\alpha_{\mathrm{S}}-\nu_l^\prime)^+\leq r_{\mathrm{sum}}-N_{\rm dl},\\
&0\leq\mu_{\mathrm{1}}\leq\cdots\leq\mu_{N_{\mathrm{dl}}};~0\leq\sigma_{\mathrm{1}}\leq\cdots\leq\sigma_{M_{\mathrm{ul}}};~0\leq\nu_{\mathrm{1}}^\prime\leq\cdots\leq\nu_{N_{\mathrm{dl}}}^\prime,\\
&\mu_i\geq1,~\forall (i+k)\geq N_{\mathrm{dl}}+1,\forall k,\\
&\sigma_j\geq1,~\forall (j+k)\geq M_{\mathrm{ul}}+1,\forall k.
\end{aligned}
\end{gather}
Again, in order to minimize the objective function above, it is clearly that $\mu_i=1, \forall i$ and $\sigma_j=1, \forall j\geq M_{\mathrm{ul}}-N_{\mathrm{dl}}+1$. Hence the objective function in (\ref{obj2}) reduces to
\begin{gather}
\begin{aligned}
d^{\text{CSIT}}_{\mathrm{sum}}&=\min \sum_{j=1}^{M_{\mathrm{ul}}-N_{\mathrm{dl}}}(M+M_{\mathrm{ul}}-N_{\mathrm{dl}}+1-2j)\sigma_j+\frac{1}{W}\sum_{l=1}^{N_{\mathrm{dl}}}(M_{\mathrm{ul}}+N_{\mathrm{dl}}+1-2l)\nu_l^\prime\\
\mathrm{Subject~to}\quad&\sum_{j=1}^{M_{\mathrm{ul}}-N_{\mathrm{dl}}}(1-\sigma_j)^++\sum_{l=1}^{N_{\mathrm{dl}}}(W\alpha_{\mathrm{S}}-\nu_l^\prime)^+\leq r_{\mathrm{sum}}-N_{\rm dl},\\ \label{obj3}
&0\leq\sigma_{\mathrm{1}}\leq\cdots\leq\sigma_{M_{\mathrm{ul}}-N_{\mathrm{dl}}},~0\leq\nu_{\mathrm{1}}^\prime\leq\cdots\leq\nu_{N_{\mathrm{dl}}}^\prime.
\end{aligned}
\end{gather}
Now we have two subcases for the optimization problem in (\ref{obj3}) when $r_{\mathrm{sum}}\geq N_{\mathrm{dl}}$. 

\emph{Subcase A:} Let $\nu_l^\prime$ have steeper descent than $\sigma_1$, i.e., $\frac{ \partial  d^{\text{CSIT}}_{\mathrm{sum}}}{\partial \nu_{N_{\rm dl}}^\prime }\geq \frac{ \partial  d^{\text{CSIT}}_{\mathrm{sum}}}{\partial \sigma_1 }$. Thus when $W\leq \frac{M_{\rm ul}-N_{\rm dl}+1}{M+M_{\rm ul}-N_{\rm dl}-1}$,
the steepest descent of the objective function in (\ref{obj3}) is along the decreasing value of $\nu_l^\prime$ with $\sigma_j=1, \forall j.$
Thus the solution to the optimization problem above is
\[d_{\mathfrak{B}_{\mathrm{sum}}}^{\text{CSIT}}=\alpha_{\mathrm{S}} d_{M_{\mathrm{ul}},N_{\mathrm{dl}}}\left(\frac{r_{\mathrm{sum}}-N_{\rm dl}}{W\alpha_{\mathrm{S}} }\right)+M(M_{\mathrm{ul}}-N_{\mathrm{dl}}),~N_{\rm dl}\leq r_{\mathrm{sum}}\leq N_{\mathrm{dl}}(1+W\alpha_{\mathrm{S}}).\]
It is obvious that when $r\geq N_{\mathrm{dl}}(1+W\alpha_{\mathrm{S}})$, $\nu_l^\prime=0, \forall l$. We can further simplify the optimization problem in~(\ref{obj3}) as
\begin{gather}
\begin{aligned}
d^{\text{CSIT}}_{\mathrm{sum}}&=\min\sum_{j=1}^{M_{\mathrm{ul}}-N_{\mathrm{dl}}}(M+M_{\mathrm{ul}}-N_{\mathrm{dl}}+1-2j)\sigma_j\\
\mathrm{Subject~to}\quad&\sum_{j=1}^{M_{\mathrm{ul}}-N_{\mathrm{dl}}}(1-\sigma_j)^+\leq r_{\mathrm{sum}}-N_{\mathrm{dl}}(1+W\alpha_{\mathrm{S}}),\\
&0\leq\sigma_{\mathrm{1}}\leq\cdots\leq\sigma_{M_{\mathrm{ul}}-N_{\mathrm{dl}}}
\end{aligned}
\end{gather}
Hence the solution to the optimization problem above is
\[d_{\mathfrak{B}_{\mathrm{sum}}}^{\text{CSIT}}= d_{M_{\mathrm{ul}}-N_{\mathrm{dl}},M}\left(r_{\mathrm{sum}}-N_{\mathrm{dl}}(W\alpha_{\mathrm{S}}+1)\right),~N_{\mathrm{dl}}(W\alpha_{\mathrm{S}}+1)\leq r_{\mathrm{sum}}\leq N_{\mathrm{dl}}W\alpha_{\mathrm{S}}+M_{\mathrm{ul}}.\]

\emph{Subcase B:} Let $\sigma_j$ have steeper descent than $\nu_1^\prime$, i.e., $ \frac{ \partial  d^{\text{CSIT}}_{\mathrm{sum}}}{\partial \sigma_{M_{\rm ul}-N_{\rm dl}} }\geq \frac{ \partial  d^{\text{CSIT}}_{\mathrm{sum}}}{\partial \nu_1^\prime }$. Thus when $W\geq \frac{M_{\rm ul}+N_{\rm dl}-1}{M-M_{\rm ul}+N_{\rm dl}+1}$,
the steepest descent of the objective function in (\ref{obj3}) is along the decreasing value of $\sigma_j$ with $\nu_l^\prime=W\alpha_{\rm S}, \forall l.$  
Now the solution is given as
\[d_{\mathfrak{B}_{\mathrm{sum}}}^{\text{CSIT}}= d_{M_{\mathrm{ul}}-N_{\mathrm{dl}},M}\left(r_{\mathrm{sum}}-N_{\mathrm{dl}}\right)+M_{\rm ul}N_{\mathrm{dl}}\alpha_{\mathrm{S}},~N_{\mathrm{dl}}\leq r_{\mathrm{sum}}\leq M_{\mathrm{ul}}.\]
The result above implies that when $r\geq M_{\rm ul}$, $\sigma_j=0, \forall j$, hence the optimization problem in~(\ref{obj3}) further reduces to 
\begin{gather}
\begin{aligned}
d^{\text{CSIT}}_{\mathrm{sum}}=&\min \frac{1}{W}\sum_{l=1}^{N_{\mathrm{dl}}}(M_{\mathrm{ul}}+N_{\mathrm{dl}}+1-2l)\nu_l^\prime\\
\mathrm{Subject~to}\quad&\sum_{l=1}^{N_{\mathrm{dl}}}(W\alpha_{\mathrm{S}}-\nu_l^\prime)^+\leq r_{\mathrm{sum}}-M_{\rm ul},\\
&0\leq\nu_{\mathrm{1}}^\prime\cdots\leq\nu_{N_{\mathrm{dl}}}^\prime.\nonumber
\end{aligned}
\end{gather}
Consequently, we have
\[d_{\mathfrak{B}_{\mathrm{sum}}}^{\text{CSIT}}=\alpha_{\mathrm{S}} d_{M_{\mathrm{ul}},N_{\mathrm{dl}}}\left(\frac{r_{\mathrm{sum}}-M_{\rm ul}}{W\alpha_{\mathrm{S}} }\right),~M_{\rm ul}\leq r_{\mathrm{sum}}\leq M_{\rm ul}+N_{\mathrm{dl}}W\alpha_{\mathrm{S}}.\]

The proof will be complete with the analysis for $N_{\mathrm{dl}}> M_{\mathrm{ul}}$, which can be derived following the same argument and thus is skipped to avoid redundancy. By combining all the cases above, we will obtain the results in Lemma~\ref{case3}.

\bibliographystyle{IEEEtran}
\scriptsize
\bibliography{Reference}
\end{document}